\documentclass[a4,10pt]{article}       % onecolumn (second format)

\usepackage[dvips]{graphicx}
\usepackage{amssymb}
\usepackage{amsmath}
\usepackage{amsthm}
\usepackage{bm}
\usepackage{bbm}
\usepackage[caption=false]{subfig}
\usepackage{vmargin}
\setmargnohfrb{20mm}{20mm}{20mm}{20mm}
\usepackage{sublabel}

\newcommand{\be}{\begin{equation}}
\newcommand{\ee}{\end{equation}}
\newcommand{\ba}{\begin{equation} \begin{aligned}}
\newcommand{\ea}{\end{aligned} \end{equation}}
\newcommand{\ddt}[1]{\frac{\mathrm{d}#1}{\mathrm{d}t}}
\newcommand{\dint}[1]{\mathrm{d}#1}
\newcommand{\kmax}{k_{\mathrm{max}}}
\newcommand{\dpt}{{\tilde{p}}}
\newcommand{\dd}{{\tilde{p}}}

\newcommand{\EX}[1]{\mathbb{E}\left[#1\right]}
\newcommand{\myvec}[1]{ \mathbf{#1} }
\newcommand{\mymat}[1]{ \boldsymbol{#1} }
\newcommand{\mygmat}[1]{ \boldsymbol{#1} }
\newcommand{\mygvec}[1]{ \boldsymbol{#1} }

\newcommand{\ones}{\mathbf{1}}
\newcommand{\TR}[1]{#1^{\!\top}}

\newcommand{\bZ}{\mathbf{Z}}

\newtheorem{prop}{Proposition}

\numberwithin{equation}{section}

\begin{document}
\title{\sc Heterogeneous network epidemics: real-time growth, variance and extinction of infection}
\date{}

\author{Frank Ball \and Thomas House}

\maketitle

\begin{abstract}
Recent years have seen a large amount of interest in epidemics on networks as a
way of representing the complex structure of contacts capable of spreading
infections through the modern human population.  The configuration model is a
popular choice in theoretical studies since it combines the ability to specify
the distribution of the number of contacts (degree) with analytical
tractability. Here we consider the early real-time behaviour of the Markovian
SIR epidemic model on a configuration model network using a multitype
branching process. We find closed-form analytic expressions for the mean and
variance of the number of infectious individuals as a function of time and the
degree of the initially infected individual(s), and write down a system of
differential equations for the probability of extinction by time $t$ that are numerically
fast compared to Monte Carlo simulation. We show that these quantities are all
sensitive to the degree distribution -- in particular we confirm that the
mean prevalence of infection depends on the first two moments of the degree
distribution and the variance in prevalence depends on the first three moments
of the degree distribution.  In contrast to most existing analytic approaches,
the accuracy of these results does not depend on having a large number of
infectious individuals, meaning that in the large population limit they would
be asymptotically exact even for one initial infectious individual.\\

\noindent{}Keywords: SIR epidemic; Configuration model; Branching process
\end{abstract}

\section{Introduction}
\label{sec:intro}
\subsection{Background}
\label{sec:background}

Models of infectious disease transmission have, from relatively modest
beginnings (e.g.~Bailey~\cite{Bailey:1957}), developed a rich domain of
applicability covering the whole spectrum of human, animal and plant pathogens,
and informing the study of questions from viral evolution, through epidemiology
of infectious diseases, to public health policy (see Heesterbeek {\it et
al.}~\cite{Heesterbeek:2015}). Increasingly, networks have been seen as a way
of modelling the complex, heterogeneous patterns of contacts between
individuals (Danon {\it et al.}~\cite{Danon:2011}).

In theoretical studies, the configuration model has been a popular choice due
to the ability to specify the number of contacts each individual has that are
capable of spreading disease, while allowing for analytic results to be
obtained (e.g.~Molloy and Reed~\cite{Molloy:1995} and Newman~\cite{Newman:2002}).
Ball and Neal~\cite{Ball:2008} used an effective degree approach (which we
describe in Section~\ref{subsec:effectivedegree} below -- cf.\ Lindquist {\it et
al.}~\cite{Lindquist:2010}) to derive a system of ordinary differential
equations that describes the deterministic limit of the epidemic model as the
population size $N \to \infty$. A much simpler (equivalent) system of only $4$
ordinary differential equations was obtained by Volz~\cite{Volz:2008} and
subsequently shown by Miller {\it et al.}~\cite{Miller:2011,Miller:2012} to be
essentially one-dimensional (the 4 ODEs were also shown by House and
Keeling~\cite{House:2010} to be a special case of the much higher dimensional
pair approximation model of Eames and Keeling~\cite{Eames:2002}, in which the
degree structure is explicit). Fully rigorous proofs of convergence in
probability of the scaled stochastic model to the deteministic limit are given
by Decreusefond {\it et al.}~\cite{Decreusefond:2012}, Bohman and
Picollelli~\cite{Bohman:2012}, Barbour and Reinert~\cite{Barbour:2013} and
Janson {\it et al.}~\cite{Janson:2014}. These works are primarily
concerned with the temporal behaviour of \emph{proportions} of the population
in different epidemiological compartments (susceptible, infectious and removed)
over the main body of a large epidemic. Here, we are also concerned with
temporal behaviour, but focus on \emph{numbers} infected early in the epidemic,
including the possibility of early stochastic extinction.

In a recent paper, Graham and House~\cite{Graham:2014} use a pairwise
approximation in conjunction with the central limit theorem for density
dependent population processes (Ethier and Kurtz~\cite{Ethier:1986}, Chapter
11) to obtain a closed-form approximation to the mean and variance of
prevalence in the linearised model which approximates the early asymptotic
exponential growth phase of a Markovian SIR epidemic on a configuration network.
In particular, they find that, under these approximations, the variance in
disease prevalence is determined by the first three moments of the network
degree distribution. In this paper, we use the effective degree approach of
Ball and Neal~\cite{Ball:2008} to approximate the early stages of the epidemic
by a continuous-time, multitype Markovian branching process, which is then
analysed in detail.  For $t\ge 0$, let $Z(t)$ denote the total number of
individuals alive in this branching process at time $t$, so $Z(t)$
approximates disease prevalence in the epidemic model during its early
asymptotic exponential growth phase.  Explicit closed-form expressions are
derived for the mean and variance of $Z(t)$, the covariance of $Z(t)$ and
$Z(s)$ to give the behaviour over time, and also for the probability of
extinction $\pi(t) = \mathbb{P}(Z(t)=0)$.  As in Graham and
House~\cite{Graham:2014}, the mean and variance in disease prevalence depends on the
degree distribution only through its first two and three moments, respectively.

The results in Graham and House~\cite{Graham:2014} assume implicitly that the
initial number of infectives is sufficiently large for the density dependent
population process central limit theorem to yield a good approximation.  In
contrast, our results assume any arbitrary, but specified, initial number of
infectives.  The  asymptotic distribution of types in the branching process,
when it does not go extinct, is also available in closed-form and enables us to
obtain a Gaussian process approximation, with explicit mean and covariance
function, for the prevalence in the early asymptotic exponential growth phase
of an SIR epidemic, with few initial infectives, which takes off and becomes
established. We show that this approximation can be applied together with the
methods of Ross {\it et al.}~\cite{Ross:2006} to estimate
epidemiological parameters from early prevalence data of a simulated epidemic
provided the first three moments of the degree distribution are known.

\subsection{Outline of the paper}

The paper is organised as follows.  The configuration network model and a
Markov SIR epidemic on that network are described in
Section~\ref{subsec:model}.  The effective degree construction of this epidemic
is outlined in Section~\ref{subsec:effectivedegree}.  Approximation of the
early stages of this epidemic by a branching process is outlined in
Section~\ref{subsec:BP}, where conditions are given for the mean, variance and
covariance functions of the number of infectives in the epidemic process to
converge to the corresponding quantities of the approximating branching process
as the population size tend to infinity.   The representation of the
approximating branching process as a continuous-time, multitype Markov
branching process is outlined in Section~\ref{subsec:BP} and described more
explicitly in Section~\ref{subsec:BPexplicit}.  The mean, variance and covariance
functions of the total number of individuals alive in the branching process are
considered in Sections~\ref{sec:means},~\ref{sec:var} and~\ref{sec:cov},
respectively.  Explicit closed-form expressions are obtained for each of these
quantities and for their limits as time $t \to \infty$.  The arguments in
Sections~\ref{sec:means},~\ref{sec:var} and~\ref{sec:cov} assume that
underlying degree distribution has a maximum degree.  In
Section~\ref{sec:unbounded}, we show that these expressions continue to hold in
the unbounded degree setting, subject to the degree distribution satisfying
suitable moment conditions.  The probability that the branching process is
extinct at time $t$ is studied in Section~\ref{sec:ext}.  Closed-form
expressions for this probability, given the initial state of the branching
process, are not available so asymptotic results as $t \to 0$ and $t \to
\infty$ are considered.

The mean, variance and covariance functions derived in
Sections~\ref{sec:means},~\ref{sec:var} and~\ref{sec:cov} are unconditional, so
they include realisations of the branching process which result in
extinction.  However, in the epidemic setting, we are often interested in
analysing the behaviour of epidemics that take off and become established,
which correspond to non-extinction of the branching process.  In
Section~\ref{sec:fluc}, we first derive the mean and variance of the total
number of individuals alive in the branching process at time $t$, conditional
upon the process having survived to time $t$; fully closed-form results are not
available owing to the absence of a closed-form expression for the survival
probability.  We then consider realisations of the branching process which
reach some specified size, $K$ say, with time being set to zero the first time
the total number of individuals alive is $K$.  The results in
Section~\ref{sec:means} yield an explicit expression for the asymptotic
distribution of types, given that the branching process does not go extinct,
which, provided $K$ is sufficiently large, enables the above branching process starting
from $K$ individuals to be approximated by a Gaussian process whose mean and
covariance functions are determined explicitly.  The theory is illustrated by
numerical examples of both forward simulation and inference in
Section~\ref{sec:numerical} and some concluding comments are given in
Section~\ref{sec:conc}.

In general, we define notation as it is introduced; we also collect notation
that is used in multiple sections in Table~\ref{tab:notation}.

\section{Model and approximating branching process}
\label{sec:model-BP}

\subsection{Model}
\label{subsec:model}

We consider the spread of an SIR epidemic on a network of $N$ individuals,
labelled $1,2,\ldots,N$, constructed using the configuration model as follows
(see e.g.\ Newman \cite{Newman:2002}).  Let $D$ be a random variable which
describes the degree of a typical individual and let $p_k=\mathbb{P}(D=k)$
$(k=0,1,\ldots)$.  Let $D_1,D_2,\ldots,D_N$ be independent realisations of $D$
and, for $i=1,2,\ldots, N$, attach $D_i$ stubs (half-edges) to individual $i$.
Pair up these stubs uniformly at random to form the edges in the network. If
$D_1+D_2+\ldots+D_N$ is odd, there will be a left-over stub, which is ignored;
the resulting network may have other `defects' such as self-loops and multiple
edges between pairs of individuals but, provided that $D$ has finite variance,
such imperfections become sparse in the network as $N \to \infty$ (see
e.g.~Durrett~\cite{Durrett:2007}, Theorem 3.1.2).  An alternative to the
degrees $D_1,D_2,\ldots,D_N$ being random is, for each $N=1,2,\cdots$, to replace
$\myvec{D}=(D_1,D_2,\ldots,D_N)$ by
$\myvec{D}^{(N)}=\left(D_1^{(N)},D_2^{(N)},\ldots,D_N^{(N)}\right)$, where the degree
sequences $\myvec{D}^{(N)}$ $(N=1,2,\ldots)$ are prescribed and satisfy
$p_k^{(N)}=N^{-1} \sum_{i=1}^N \delta_{k,D_i^{(N)}} \to p_k$ as $N \to \infty$ $(k=0,1,\ldots)$,
where the Kronecker delta $\delta_{k,j}$ is 1 if $k=j$ and $0$ otherwise (see  e.g.~Molloy and
Reed~\cite{Molloy:1995}).
%The
%results of the paper all apply also to this setting.

The epidemic is defined as follows.  Initially some individuals are infective
and the remaining individuals are susceptible.  Infective individuals have
independent infectious periods, each having an exponential distribution with
rate $\gamma$ (and hence mean $\gamma^{-1}$), after which they become recovered
and play no further role in the epidemic.  Throughout its infectious period, an
infective contacts each of its susceptible neighbours in the network at the
points of independent Poisson processes having rate $\tau$, so the probability
that a given infective contacts a given neighbour before the infective recovers
is $\tau/(\gamma+\tau)$.  Any contacted susceptible immediately becomes
infective and may transmit the infection to any of its neighbouring
susceptibles; i.e.~there is no latent period.  All the infectious periods and
Poisson processes governing transmission of infection are mutually independent.
The epidemic ends as soon as there is no infective present in the network.

\subsection{The effective degree model}
\label{subsec:effectivedegree}

Ball and Neal~\cite{Ball:2008} introduced an `effective-degree' construction of
the above epidemic, in which the network is constructed as the epidemic
progresses.  The process starts with some individuals infective and the
remaining individuals susceptible, but with none of the stubs paired up.  For
$i=1,2,\ldots,N$, the effective degree of individual $i$ is initially $D_i$.
Infected individuals transmit infection by pairing their stubs with stubs
attached to susceptible individuals in the following fashion.  An infected
individual makes infectious contacts down its unpaired stubs independently at
rate $\tau$ and is removed at rate $\gamma$. When an infective, individual $i$
say, transmits infection down a stub that stub is paired with a stub (attached
to individual $j$, say) chosen independently and uniformly at random from all
the unpaired stubs, to form an edge.  The effective degrees of individuals $i$
and $j$ are both reduced by $1$.  If $i=j$ then the effective degree of
individual $i$ is reduced by $2$ but this will not significantly affect the
dynamics for large populations since the probability of it happening is
$O(N^{-1}$).  If individual $j$ is susceptible then it becomes infective and
can transmit infection down any of its unattached stubs.  As before, the
epidemic ends as soon as there is no infective present.  The network is then
typically only partially constructed but that does not matter if interest is
focussed on properties of the epidemic.  In the original formulation of Ball
and Neal~\cite{Ball:2008}, when an infective recovers its unpaired stubs, if
any, were paired with stubs chosen uniformly at random without replacement from
the set of unpaired stubs but that is unnecessary; the stubs from such an
infective can simply be left in the set of unpaired stubs.

\subsection{Approximating multitype branching process}
\label{subsec:BP}

Suppose that the size $N$ of the network is large and the initial number of
infectives is small.  Then during the early stages of an epidemic it is very
likely that each time an infective individual transmits infection down a stub
that stub is paired with a stub belonging to a susceptible individual.  It
follows that the early stages of such an epidemic can be approximated by a
branching process in which each newly-infected individual has their ``full"
effective degree (i.e. their actual degree minus one for the stub that is
paired with their infector).  This approximation can be made fully rigorous by
considering a sequence of epidemics, indexed by $N$, and using a coupling
argument; see e.g.~Ball and Neal~\cite{Ball:2008}, which treats a more general
model in which infective individuals also make contacts with individuals chosen
uniformly at random from the population.  Let $E_N$ denote the epidemic on a
network of $N$ individuals and let $\mathcal{B}$ denote the approximating
branching process.  Then following  Ball and Neal~\cite{Ball:2008} (see
Appendix A) if $\mu_D=\EX{D}$ is finite then the epidemics $E_1,E_2\dots$ and
the branching process $\mathcal{B}$ can be constructed on a common probability
space so that, with probability one, over any finite time interval $[0,t]$ the
process of infectives in $E_N$ and the branching process $\mathcal{B}$ coincide
for all sufficiently large $N$.  The same result holds for the model with
prescribed degree sequences provided that $p^{(N)}_k \to p_k$ $(k=0,1,\ldots)$
and $\mu_D^{(N)}=\sum_{k=0}^\infty p^{(N)}_k k \to \mu_D$ as $N \to \infty$,
where $\sum_{k=0}^{\infty} p_k=1$ and $\mu_D < \infty$.

As indicated in Appendix~\ref{app:moments}, the branching process $\mathcal{B}$
is not an almost sure upperbound for the process of infectives in $E_N$, so
unlike in Theorem 3.1 of Ball and Donnelly~\cite{Ball:1995} which considers
homogeneously mixing epidemics, one cannot simply use the dominated convergence
theorem to deduce convergence of moments of the number of infectives in the
epidemic process to corresponding moments of the branching process as $N \to
\infty$.   If there is a maximum degree $\kmax$ (i.e.\ $p_k=0$ for all
$k>\kmax$, or $p_k^{(N)}=0$ for all $k>\kmax$ and all $N$ in the model with
prescribed degrees) then, for all $N$, the process of infectives in $E_N$ is
bounded above by a branching process in which each newly-infected individual
has the maximun effective degree $\kmax-1$, so in that case the dominated
convergence theorem can be used to prove convergence of moments.  In
Appendix~\ref{app:moments}, we consider the case when there is no maximum
degree and use uniform integrability arguments to determine sufficient
conditions for the mean and variance of the number of infectives at any given
time $t \ge 0$, and the covariance of the number of infectives at any given
times $t,s \ge 0$, in the epidemic $E_N$ to converge to the corresponding mean,
variance and covariance of the branching process $\mathcal{B}$ as $N \to
\infty$.  Specifically, we prove that (i) in the model with prescribed degrees
these moments converge if, in addition to the conditions given above, there
exists $\delta>0$ such that $\mu_{D^{3+\delta}}^{(N)}=\sum_{k=0}^\infty
p^{(N)}_k k^{3+\delta} \to \mu_{D^{3+\delta}}=\sum_{k=0}^\infty p_k
k^{3+\delta} $ as $N \to \infty$, where $\mu_{D^{3+\delta}} < \infty$; and (ii)
in the model with random degrees they converge if the moment-generating
function $M_{D^2}(\theta)=\EX{\exp \left(\theta D^2\right)}$ of $D^2$ is finite
for some $\theta_0>0$. Note that the latter condition implies that
$\EX{D^\alpha}< \infty$ for all $\alpha \ge 0$.

In this context, we note that, for the model with prescribed degrees, the
weakest conditions obtained on the moments of the degree distribution for
convergence of the scaled stochastic epidemic on to its deterministic limit are
given by Janson {\it et al.}~\cite{Janson:2014}, who require uniform
boundedness of the second moment of $\myvec{D}^{(N)}$.  However that paper, and
the other related papers cited in the second paragraph of
Section~\ref{sec:background}, (i) are concerned with the entire time course of
the epidemic; (ii) assume that either the epidemic starts with a positive
fraction of the population infected in the limit as $N \to \infty$, or if that
limiting fraction is zero then the convergence is for epidemics which take off
and involves a random time translation describing when the epidemic becomes
suitably established; and (iii) consider the evolution of the proportion of the
population that is susceptible, infective or recovered.  By contrast, this
paper is concerned with epidemics initiated by few infectives and considers the
number, rather than proportion, of infectives during the early phase of such an
epidemic.  Under the coupling mentioned above, in the limit as $N \to \infty$,
if an epidemic takes off then the duration of its early (exponentially growing)
phase tends almost surely to infinity.

The limiting branching process may be described by a continuous-time multitype
Markov branching process, with the type of an infective corresponding to its
effective degree.  Let $\tilde{D}$ denote the (size-biased) degree of a typical
neighbour of a typical individual in the network and let
$\tilde{p}_k=\mathbb{P}(\tilde{D}=k)$ $(k=1,2,\ldots)$. Then
$\tilde{p}_k=\mu_D^{-1} k p_k$, where $\mu_D=\EX{D}$, since when a stub is
paired it is $k$ times as likely to be paired with a stub from a given
individual having degree $k$ than it is with a stub from a given individual
having degree $1$.  Under the branching process approximation, the effective
degree of a newly infected individual is distributed according to
$\tilde{D}-1$, since one of that individual's stubs is `used up' when it is
infected.  Note for future reference that $\mu_{\tilde{D}}=\mu_D^{-1}\mu_{D^2}$
and, more generally, $\mu_{f(\tilde{D})}=\mu_D^{-1} \mu_{Df(D)}$ for any
real-valued function $f$.  (For a random variable, $X$ say, we use $\mu_X$ to
denote its expectation $\EX{X}$.  Thus, for example, $\mu_{Df(D)}=\EX{Df(D)}$.)

\subsection{Explicit form for the multitype branching process}
\label{subsec:BPexplicit}

We now assume that there is a maximum degree $\kmax$.
%i.e.\ that $p_k=0$ for all $k>\kmax$.
We show in Section~\ref{sec:unbounded} that our results for moments of the branching process
extend to the case of no maximum degree size, subject to suitable moment conditions on $D$.
%We show later that almost all of our results extend to the case
%of no maximum degree size, subject to suitable moment conditions on $D$.
Thus
the type space for the branching process is $\mathcal{K}=\{0,1,\ldots,\kmax\}$.
Note that only initial infectives can have type $\kmax$.  For
$k\in\mathcal{K}$, an individual of type $k$ dies if either its infectious
period comes to an end or it transmits infection down one of its unattached
stubs, whichever happens first.  If the former happens first then the
individual has no offspring, otherwise it has two offspring, namely an
individual of type $k-1$ and an individual whose type is distributed according
to  $\tilde{D}-1$.  Note that an individual of type $0$ necessarily has no
offspring when it dies.  Thus, for $k\in\mathcal{K}$, the lifetime of an
individual of type $k$ has an exponential distribution with rate
\be
\label{omegak}
 \omega_k = \gamma + \tau k \text{ ,}
\ee
and when it dies its offspring is distributed as follows (recall that
$\dpt_{\kmax+1}=0$):
\ba
\mathbb{P}\left(\mathrm{Offspring} = \varnothing | \mathrm{Parent\ type} = k \right) &
 = \frac{\gamma}{\gamma + \tau k} \text{ ,} \\
 \mathbb{P}\left(\mathrm{Offspring} = \{k-1,l\} | \mathrm{Parent\ type} = k \right) &
 = \frac{\tau k \dpt_{l+1}}{\gamma + \tau k} \quad(l=0,1,\ldots,\kmax-1) \text{ .}
\ea
The joint probability-generating function (PGF) for offspring of a type-$k$
individual is therefore
\be
\label{PGFPk}
P_k(\myvec{s}) = \frac{1}{\omega_k} \left(
	\gamma + \tau k s_{k-1} \sum_{l=0}^{\kmax-1} \dpt_{l+1} s_l
\right) \text{ ,}
\ee
where $\myvec{s}=(s_k)$. In general we will write $\myvec{v} = (v_k) =
\TR{(v_0,v_1,\ldots,v_{\kmax})}$ for a column vector in $\mathbb{R}^{\kmax+1}$,
where $^\top$ denotes transpose. Verbally, we will call $v_0$ the $0$-th element of
such a vector, $v_1$ the first element etc. For $k\in\mathcal{K}$, let
$\mygvec{\partial}P_k(\myvec{s})$ be the column vector whose $i$-th element is
$\frac{\partial P_k(\myvec{s})}{\partial s_i}$ and let $\mygmat{\partial}^2
P_{k}(\myvec{s})$ be the matrix whose $(i,j)$-th element is $\frac{{\partial}^2
P_{k}(\myvec{s})}{\partial s_i \partial s_j}$.  We note for future reference
that, for $i,j,k\in\mathcal{K}$,
\ba
\label{fac2off}
\left(\mygvec{\partial}P_k(\ones)\right)_i&=\frac{\tau k}{\gamma+\tau k}\left(\dpt_{i+1}+\delta_{k-1,i}\right) \text{ ,} \\
\left[\mygmat{\partial}^2 P_{k}(\ones)\right]_{i,j}&=\frac{\tau k}{\gamma+\tau k}\left(\dpt_{i+1}\delta_{k-1,j}+\dpt_{j+1}\delta_{k-1,i}\right) \text{ ,}
\ea
where $\ones$ is the length-$(\kmax+1)$ column vector of ones.\\

\noindent{}For $t \ge 0$, let $\bZ(t)=\left(Z_i(t)\right)$, where $Z_i(t)$
denotes the number of individuals of type $i$ alive at time $t$, and let
$Z(t)=Z_0(t)+Z_1(t)+\ldots+Z_{\kmax}(t)= \TR{\ones}\bZ(t) $ denote the total
number individuals alive at time $t$.  For $k\in\mathcal{K}$, we use the
notation $\{\bZ^{(k)}(t):t \ge 0\}$, where $\bZ^{(k)}(t)=(Z^{(k)}_i(t))$, to
denote a process starting with a single individual, whose type is $k$, at time
$0$ (i.e.\ where $Z_i^{(k)}(0) = \delta_{i,k}$, $i \in \mathcal{K}$).  Further,
$Z^{(k)}(t)=\TR{\ones}{\bZ^{(k)}(t)}$ denotes the total number of individuals
at time $t$ in such a process.

\section{Behaviour of means}
\label{sec:means}
In the next three sections we consider the behaviour of the mean, variance and covariance function of the total number of individuals over time in the branching process which approximates the initial phase of an epidemic.
For $t \ge 0$ and $i,j,k\in\mathcal{K}$, let
\ba
\nonumber
\mymat{M}(t)& =[m_{i,j}(t)] \text{ , where}\quad
m_{i,j}(t)=\EX{Z^{(i)}_j(t)}\text{ ,} \\
\myvec{m}^{(k)}(t)& =\EX{\bZ^{(k)}(t)}=\TR{\mygmat{M}(t)}{\myvec{u}_k}
\text{ ,} \\
m^{(k)}(t)& =\EX{Z^{(k)}(t)}=\TR{\ones}\myvec{m}^{(k)}(t)
=\TR{\myvec{u}_k} \mymat{M}(t)\myvec{1}\text{ ,}
\ea
where $\myvec{u}_k$ is a length-$(\kmax+1)$ column vector with $k$-th element
equal to 1 and other elements equal to 0. A standard argument using the
Kolmogorov forward equation (see e.g.~Dorman {\it et al.}~\cite{Dorman:2004},
Section 7 and recall that $\dd_{\kmax+1}=0$) then yields that
\be
\label{Meq}
\ddt{}\mygmat{M}(t)=\mygmat{M}(t) \mygmat{\Omega}\text{ ,}\qquad \mygmat{M}(0)=\mygmat{I}\text{ ,}
\ee
where $\mygmat{I}$ denotes the $(\kmax+1) \times (\kmax+1)$ identity matrix and $\mygmat{\Omega}=[\Omega_{l,k}]$ is the
$(\kmax+1) \times (\kmax+1)$ matrix with elements given by
\be
\nonumber
\Omega_{l,k} = \tau l \left(\dd_{k+1} + \delta_{l,k+1}\right)
 -(\gamma + \tau l) \delta_{l,k} \qquad (l,k \in \mathcal{K}) \text{ .}
\ee
The solution to~\eqref{Meq} is then straightforwardly given by
\be
\label{meanM}
\mygmat{M}(t)={\rm e}^{\mygmat{\Omega}t}
=\sum_{l=0}^\infty \frac{t^l \mygmat{\Omega}^l}{l!} \text{ .}
\ee
We show in Appendix~\ref{app:eigenvalues} that the eigenvalues of $\mygmat{\Omega}$ are
\be
\label{eigomega}
\lambda_i = \begin{cases}
	- \gamma - i \tau & \text{ for } i \in \{0,2,3,\ldots,\kmax\} \text{ ,}\\
	\tau(({\textstyle \sum_{l=0}^{\kmax}} l \dd_{l+1}) -1)  - \gamma & \text{ for }
	i = 1 \text{ .}
\end{cases}
\ee
We denote the dominant eigenvalue, $\lambda_1$, by $r$, so
\be
\label{growthrate}
r= \tau(({\textstyle \sum_{l=0}^{\kmax}} l \dd_{l+1}) -1)  - \gamma
= \tau \mu_{\tilde{D}-2}  - \gamma \text{ .}
\ee
If $r \le 0$, the branching process $\{\bZ(t):t \ge 0\}$ goes extinct almost
surely.  If $r>0$, then $r$ gives the asymptotic exponential growth rate of
$\{Z(t):t \ge 0\}$ (and also of $\{Z_i(t):t \ge 0\}$ for $i \in \mathcal{K} \setminus \{\kmax\}$) when $\{\bZ(t):t \ge 0\}$ does not go extinct.\\

\noindent{}For $i\in \mathcal{K}$, let $\TR{\myvec{w}_i} = (w_{i,k})$ be a
left eigenvector of $\mygmat{\Omega}$ corresponding to the eigenvalue
$\lambda_i$, so
\be
\label{leig}
\TR{\myvec{w}_i} \mygmat{\Omega} = \lambda_i \TR{\myvec{w}_i} \text{ .}
\ee
The Perron-Frobenius theory implies that $\myvec{w}_1$ can be chosen so that
all of its elements are positive and $\TR{\myvec{w}_1} \ones=1$.  The
left-eigenvector $\myvec{w}_1$ then yields a probability distribution which
gives the asymptotic relative frequencies of the different types, as $t \to
\infty$, when $\{\bZ(t):t \ge 0\}$ does not go extinct.\\

\noindent{}Expanding~\eqref{leig} in components yields
\be
\label{leigcomp}
\sum_{l=0}^{\kmax} w_{1,l}  \left(
\tau l \left(\dd_{k+1} + \delta_{l,k+1}\right)
-(\gamma + \tau l) \delta_{l,k} \right) = r w_{1,k} \quad (k\in\mathcal{K}) \text{ .}
\ee
Let $w(s)=\sum_{l=0}^{\kmax} s^l w_{1,l}$ $(s \ge 0)$ denote the
(probability-)generating function of $\myvec{w}_1$.
Multiplying~\eqref{leigcomp} by $s^k$ and summing over $k$ yields
\be
\label{diffw}
\tau f_{\tilde{D}-1}(s)\mu_W+\tau(1-s) w'(s)=(r+\gamma)w(s)\text{ ,}
\ee
where $f_{\tilde{D}-1}(s)=\sum_{k=1}^{\kmax} \dd_k s^{k-1}$ is the PGF of
$\tilde{D}-1$ and $\mu_W=\sum_{k=0}^{\kmax} k w_{1,k}$ is the mean of the
distribution $\myvec{w}_1$.  Setting $s=1$ in~\eqref{diffw} and
using~\eqref{growthrate} yields
\be
\label{meanW}
\mu_W=\frac{r+\gamma}{\tau}=\mu_{\tilde{D}-2}\text{ .}
\ee

Note that~\eqref{meanW} has a simple intuitive explanation.  For large $t$, a
typical individual gives birth at rate $\sum_{l=0}^{\kmax}w_{1,l} l \tau=\tau\mu_W$ and dies
completely (i.e.~without producing any offspring) at rate $\gamma$, so the
population growth rate $r=\tau\mu_W-\gamma$ and~\eqref{meanW} follows using~\eqref{growthrate} .

For $i,k \in \mathbb{Z}^+$, let $k_{[i]}=k(k-1)\ldots(k-i+1)$ denote a falling
factorial, with the convention $k_{[0]}=1$.  For $i=0,1,\ldots$, let
$\mu_W^{[i]}=\sum_{k=0}^{\kmax} k_{[i]} w_{1,k}$ be the $i$th factorial-moment
of the distribution $\myvec{w}_1$, so $\mu_W^{[0]}=1$ and $\mu_W^{[1]}=\mu_W$.
Note that $\mu_W^{[i]}=w^{(i)}(1)$ $(i=0,1,\ldots)$, where
$w^{(i)}(s)$ denotes the $i$-th derivative of $w(s)$.  Repeated differentiation
of~\eqref{diffw} yields
\be
\label{w1fac}
\mu_W^{[i]}=\frac{\mu_{\tilde{D}-2}}{\mu_{\tilde{D}-2+i}}\mu_{\tilde{D}-1}^{[i]}\text{ ,}
\ee
where $\mu_{\tilde{D}-1}^{[i]}$ is the $i$-th factorial-moment of $\tilde{D}-1$.
Note that $\mu_{\tilde{D}-1}^{[i]}=0$ for $i \ge \kmax$.  It then follows,
using the inversion formula which expresses the probability mass function of a
non-negative integer-valued random variable in terms of its factorial-moments
(see e.g.~Daley and Vere-Jones~\cite{Daley:1988}, page 117), that
\be
\label{w1vec}
w_{1,k}=\begin{cases}
	\sum_{i=k}^{\kmax -1} (-1)^{i-k} {i \choose k}\frac{\mu_{\tilde{D}-2} \mu_{\tilde{D}-1}^{[i]}}{i! \mu_{\tilde{D}-2+i}} & \text{ if } k=0,1,\ldots,\kmax -1 \text{ ,}\\
	0 & \text{ if }
	k = \kmax \text{ .}
\end{cases}
\ee
Observe that $w_{1,\kmax}=0$ since only initial infectives can have type $\kmax$.
Observe also that $\myvec{w}_1$ is determined just by the degree distribution of the
network and is invariant to the epidemic parameters $\gamma$ and $\tau$.\\

\noindent{}For $t\ge 0$, let $\myvec{m}(t)=(m^{(i)}(t))=\TR{(\mymat{M}(t)\ones)}$.
Thus the $k$-th element of $\myvec{m}(t)$ contains the mean total population
size at time $t$ given that the process starts with a single individual whose
type is $k$.  We derive a simple expression for $\myvec{m}(t)$.  The following
proposition is useful.
\begin{prop}
\label{prop:matid}
	For a matrix $\mymat{M}$ and vectors $\myvec{x}$, $\myvec{y}$ such that
	$\mymat{M} \myvec{x} = a \myvec{x} + b \myvec{y}$ and $\mymat{M} \myvec{y} = c \myvec{y}$,
	where $a,b$ and $c$ are scalars satisfying $a \ne c$,
	\be
	{\rm e}^{\mymat{M}t} \myvec{x} = {\rm e}^{at} \myvec{x} + \frac{b}{a-c}
	\left({\rm e}^{at} - {\rm e}^{ct}\right) \myvec{y}  \qquad \mbox{and}\qquad
	{\rm e}^{\mymat{M}t} \myvec{y} = {\rm e}^{c t} \myvec{y} \text{ .}
	\ee
\end{prop}
\begin{proof}
	The second identity follows straightforwardly from the definition of the matrix exponential and
	the fact that $\myvec{y}$ is a right eigenvector with eigenvalue $c$. For the first identity,
	\ba
	{\rm e}^{\mymat{M} t} \myvec{x} & = \sum_{i=0}^{\infty} \frac{1}{i!}(\mymat{M}t)^i \myvec{x} \\
 & = \sum_{i=0}^{\infty} \frac{t^i}{i!} \left( a^i \myvec{x} +
	\sum_{j=0}^{i-1} b c^{i-1-j}a^j \myvec{y} \right) \\
	& = {\rm e}^{at}\myvec{x} + \sum_{i=0}^{\infty} \frac{t^i}{i!}
b c^{i-1} \frac{\left(\frac{a}{c}\right)^i -1}{\left(\frac{a}{c}\right)-1} \myvec{y} \\
& = {\rm e}^{at} \myvec{x} + \frac{b}{a-c}
	\left({\rm e}^{at} - {\rm e}^{ct}\right) \myvec{y} \text{ .}
	\ea
\end{proof}

\noindent{}Let $\myvec{n}=\TR{(0,1,\ldots,\kmax)}$.  Observe that
\be
\mygmat{\Omega} \ones = \tau \myvec{n} - \gamma \ones  \qquad \mbox{and}\qquad
\mygmat{\Omega} \myvec{n} = r \myvec{n} \text{ ,} \label{idents}
\ee
so using Proposition~\ref{prop:matid} with $\mymat{M}=\mygmat{\Omega},
\myvec{x}=\ones, \myvec{y}=\myvec{n}, a=-\gamma, b=\tau$ and $c=r$, and
recalling from~\eqref{growthrate} that $r+\gamma=\mu_{\tilde{D}-2}\tau$, we
have
\be
{\rm e}^{\mygmat{\Omega} x} \ones =\mu_{\tilde{D}-2}^{-1}  \left({\rm e}^{rx} - {\rm e}^{-\gamma x}
\right) \myvec{n} + {\rm e}^{-\gamma x} \ones  \qquad \mbox{and} \qquad
{\rm e}^{\mygmat{\Omega} x} \myvec{n} = {\rm e}^{r x}\myvec{n} \text{ .}
\label{n1}
\ee
Thus,
\be
\label{meantot}
\myvec{m}(t)=\mu_{\tilde{D}-2}^{-1}\left({\rm e}^{rt} - {\rm e}^{-\gamma t}
\right) \myvec{n} + {\rm e}^{-\gamma t}\ones \text{ .}
\ee
and
\be
\label{asymmeantot}
\lim_{t \to \infty}{\rm e}^{-r t} \myvec{m}(t)=\mu_{\tilde{D}-2}^{-1} \myvec{n}.
\ee
While it is well known that asymptotically the mean prevalence grows
exponentially with rate constant $r$, i.e.\ that $\mathrm{prevalence}
\propto{\rm e}^{rt}$, these results allow us to see from~\eqref{meantot} that
the rate of convergence to this asymptotic behaviour is $r+\gamma$, and
from~\eqref{asymmeantot} that the constant of proportionality is the degree of
the initially infected individual divided by $\mu_{\tilde{D}-2}$.\\

\noindent{}We also consider the relationship between the equations above and
the diverse ODE approaches to the mean behaviour of the full network epidemic
model. Miller and Kiss~\cite{Miller:2014} consider several such approaches;
their notation can be related to ours by defining
\be
I(t) = \TR{\myvec{m}^{(k)}(t)} \myvec{1} \text{ ,} \qquad
\lambda(t) = \TR{\myvec{m}^{(k)}(t)} \myvec{n} \text{ .}
\label{iladef}
\ee
Substituting~\eqref{iladef} and~\eqref{idents} into~\eqref{Meq} gives
\be
\ddt{I} = \tau \lambda - \gamma I \text{ ,} \qquad
\ddt{\lambda} = r \lambda \text{ .}
\ee
This pair of equations can be derived from various models considered by Miller
and Kiss~\cite[c.f.~Section 3.4.1]{Miller:2014} assuming an initially small
infectious population and negligible susceptible depletion. Therefore, our
results suggest that the ODE approaches to mean behaviour do not require
correction as the infectious population becomes extremely small, and the
typical assumption that $1 \ll I(t=0) \ll N$ for the ODE system to hold may be
too conservative, with $N \gg 1$ being all that is required.

\section{Variance}
\label{sec:var}

The variance in infectious prevalence during the exponentially growing phase of an epidemic was considered in Graham and
House~\cite{Graham:2014}, but using the diffusion limit and an argument about
the neighbourhood around an infective. A branching process limit lets us be
more explicit. For $k\in\mathcal{K}$, let $\myvec{u}_k$ denote the
length-$(\kmax+1)$ column vector whose element corresponding to type $k$ is $1$ and all other
elements are $0$, so ${\myvec{u}_k}=(\delta_{i,k})$.  For $t \ge 0$ and $k \in \mathcal{K}$, let $\sigma_{ij}^{(k)}(t)={\rm cov}\left(Z^{(k)}_i(t),Z^{(k)}_j(t)\right)$ $(i,j \in \mathcal{K})$.

A matrix integrating
factor argument gives
\ba
\mymat{V}^{(k)}(t)=\left[\sigma_{ij}^{(k)}(t)\right] & = \EX{\myvec{Z}^{(k)}(t) \TR{\myvec{Z}^{(k)}(t)}}
-\EX{\myvec{Z}^{(k)}(t)}\EX{\TR{\myvec{Z}^{(k)}(t)}}
\\ & = \int_{0}^{t} {\rm e}^{\TR{\mygmat{\Omega}}(t-u)} \mymat{B}_k(u)
{\rm e}^{{\mygmat{\Omega}}(t-u)} \dint{u}
\text{ ,} \label{Vk}
\ea
where
\ba
\label{C_k}
\mymat{B}_k(t) & = \sum_{l=0}^{\kmax} \left({\rm e}^{\mygmat{\Omega}t}\right)_{k,l} \mymat{C}_l \text{ ,}
\\
\mymat{C}_k & = \omega_k \left(\mygmat{\partial}^2 P_{k}(\ones) +
	\mathrm{diag}(\myvec{f}_k) -
\myvec{u}_k \TR{\myvec{f}_k}-
\myvec{f}_k \TR{\myvec{u}_k}+
\myvec{u}_k \TR{\myvec{u}_k}
\right) \text{ ,} \\
\myvec{f}_k & = \mygvec{\partial}P_k(\ones) \text{ .}
\ea
See Dorman {\it et al.}~\cite{Dorman:2004}, Section 9, and also Athreya and
Ney~\cite{Athreya:1972}, page 203, for details\footnote{There is a
small error in the latter -- in the expression for $b_{jk}^{(i)}$ on page 203
of~\cite{Athreya:1972},
$\delta_{jk}-b_{ij}\delta_{ik}-b_{ik}\delta_{ij}+\delta_{ij}\delta_{ik}$ should
be replaced by $b_{ij}\delta_{jk}-b_{ik}\delta_{ij}-b_{ij}\delta_{ik}$.}.\\

\noindent{}For $t \ge 0$ and $k\in\mathcal{K}$, let $v^{(k)}(t)$ denote the
variance of the total population size at time $t$ given that the process starts
with a single individual, whose type is $k$. Then $v^{(k)}(t)= \TR{\ones}
\mymat{V}^{(k)}(t) \ones$ and it follows using~\eqref{Vk} that
\ba
v^{(k)}(t) & = \int_{0}^{t} \TR{({\rm e}^{\mygmat{\Omega}(t-u)} \ones)}
\mymat{B}_k(u) ({\rm e}^{{\mygmat{\Omega}}(t-u)} \ones) \dint{u}\\
& = \int_{0}^{t} \left(
\mu_{\tilde{D}-2}^{-1} \left({\rm e}^{r(t-u)} - {\rm e}^{-\gamma (t-u)}
\right)\right)^2
\TR{\myvec{n}}\mymat{B}_k(u) \myvec{n}  \; \dint{u} \\
 &\quad + 2 \int_{0}^{t}
\mu_{\tilde{D}-2}^{-1} \left({\rm e}^{r(t-u)} - {\rm e}^{-\gamma (t-u)}
\right){\rm e}^{-\gamma (t-u)}
\TR{\myvec{1}}\mymat{B}_k(u) \myvec{n}  \; \dint{u}\\
 &\quad + \int_{0}^{t}
{\rm e}^{-2 \gamma (t-u)}
\TR{\myvec{1}}\mymat{B}_k(u) \myvec{1} \; \dint{u}  \text{ ,}
\label{vk}
\ea
where we have used the first equation in~\eqref{n1} in deriving the last line.
This quantity has an exact but rather complex closed-form solution, which we give below
and derive in Appendix~\ref{app:variance}.

Let $\myvec{n}_2=\TR{(0^2,1^2,\ldots,{\kmax}^2)}$ and, for $t\ge 0$, let $\myvec{v}(t)=(v^{(i)}(t))$.  Then
\be
\label{vvec1}
\myvec{v}(t)=\alpha_0(t) \ones + \alpha_1(t) \myvec{n} +  \alpha_2(t) \myvec{n}_2\text{ ,}
\ee
where
\ba
\label{alphas}
\alpha_0(t) &= \gamma I_2(t)\text{ ,}\\
\alpha_1(t) &= \gamma \mu_{\tilde{D}-2}^{-1}\left[I_1(t)-I_2(t)+2I_3(t)+\mu_{\tilde{D}-2}^{-1}\mu_{\tilde{D}}^{-1}\mu_{(\tilde{D}-1)^2+1}(I_4(t)-I_5(t))\right]\\
&\quad+\tau\left[I_1(t)+2I_3(t)+\mu_{\tilde{D}-2}^{-2}\mu_{(\tilde{D}-2)^2}I_4(t)\right]\text{ ,}\\
\alpha_2(t) &= \gamma \mu_{\tilde{D}-2}^{-2} I_5(t)\text{ ,}
\ea
with
\ba
\label{integrals}
I_1(t)&=\frac{{\rm e}^{rt}-{\rm e}^{-2\gamma t}}{r+2\gamma}\text{ ,}\\
I_2(t)&=\frac{{\rm e}^{-\gamma t}\left(1-{\rm e}^{-\gamma t}\right)}{\gamma}\text{ ,}\\
I_3(t)&=\frac{{\rm e}^{rt}\left(1-{\rm e}^{-\gamma t}\right)}{\gamma}-I_1(t)\text{ ,}\\
I_4(t)&=\frac{{\rm e}^{rt}\left({\rm e}^{rt}-1\right)}{r}-2I_3(t)-I_1(t) \text{ ,}\\
I_5(t)&=\frac{{\rm e}^{2rt}-{\rm e}^{-(\gamma+2\tau)t}}{2r+2\tau+\gamma}
-2\frac{{\rm e}^{-\gamma t}\left({\rm e}^{rt}-{\rm e}^{-2\tau t}\right)}{r+2\tau}+\frac{{\rm e}^{-\gamma t}\left({\rm e}^{-2\tau t}-{\rm e}^{-\gamma t}\right)}{\gamma-2\tau}\text{ .}
\ea
(In~\eqref{integrals}, if a denominator is zero then the expression is given by the limit as that denominator tends to zero.  For example, if $\gamma=0$ then $I_2(t)=t{\rm e}^{-\gamma t}$, and if $\gamma=2\tau$, the final term in $I_5(t)$ is replaced by $t{\rm
e}^{-2 \gamma t}$.)

\noindent{}Equation~\eqref{vvec1} leads to a rather complex expression for
$\myvec{v}(t)$ in terms of elementary functions.  However, its asymptotic form
as $t \to \infty$ is much simpler.  Note that $\lim_{t \to \infty} {\rm
e}^{-2rt} I_k(t)=0$, for $k=1,2,3$, $\lim_{t \to \infty} {\rm e}^{-2rt}
I_4(t)=\frac{1}{r}$ and  $\lim_{t \to \infty} {\rm e}^{-2rt}
I_5(t)=\frac{1}{2r+2\tau+\gamma}=\frac{1}{2\mu_{\tilde{D}-1}\tau-\gamma}$.
Substituting these limits into~\eqref{vvec1} and~\eqref{alphas} yields
\be
\label{asymvar}
\lim_{t \to \infty} {\rm e}^{-2rt}\myvec{v}(t)=
\frac{1}{\mu_{\tilde{D}-2}^2\left(2\mu_{\tilde{D}-1}\tau-\gamma\right)}
\left[\frac{2\tau \mu_{\tilde{D}-1}\left(\mu_{(\tilde{D}-2)^2} \tau +\gamma\right)}{r}\myvec{n}
+\gamma \myvec{n}_2\right]\text{ .}
\ee
Note that both the asymptotic and exact expressions for $\myvec{v}(t)$ depend
on the degree distribution $D$ only through its first three moments.\\

\noindent{}It follows from~\eqref{asymmeantot} and~\eqref{asymvar} that, for
$k\in\mathcal{K}$,
\be
\label{asymvarmean}
\lim_{t \to \infty}\frac{{\rm var}\left(Z^{(k)}(t)\right)}{\EX{Z^{(k)}(t)}^2}=
\frac{1}{2\mu_{\tilde{D}-1}\tau-\gamma}
\left[\gamma+\frac{2\tau \mu_{\tilde{D}-1}\left(\mu_{(\tilde{D}-2)^2} \tau +\gamma\right)}{kr}\right]
\text{ .}
\ee
We note two features of these results. First, the equations~\eqref{vvec1}
and~\eqref{alphas} involve many rates that are linear combinations of $r$,
$\tau$ and $\gamma$, with the dominant being $2r$ and the subdominant being
$r$.  This leads to complex real-time behaviour as the system approaches its
asymptotic limit. In the diffusion limit, only the dominant and subdominant
rates are present, leading to the same overall rate of convergence
$r$, but other rates are not present~\cite{Graham:2014}. Secondly, the
dependence of the variance on initial conditions is not simple proportionality,
meaning that~\eqref{asymvar} contains terms proportional to both $\myvec{n}$
and $\myvec{n}_2$ (unless $\gamma=0$) and the right-hand side of~\eqref{asymvarmean} depends on
$k$.

\section{Covariance function}
\label{sec:cov}

For $t,s \ge 0$ and $k\in\mathcal{K}$, let $\sigma^{(k)}(t,s)={\rm
cov}\left(Z^{(k)}(t), Z^{(k)}(s)\right)$ denote the covariance of the total
population sizes at times $t$ and $s$ in the branching process which
approximates the early phase of an epidemic, given that the process starts with
a single individual, whose type is $k$.  We assume without loss of generality
that $t \le s$; although this choice does not respect alphabetical order, the
majority of results that follow take $t$ as an argument rather than $s$, and
are therefore more easily read as functions of time.  Then
\ba
\label{sigmast}
\sigma^{(k)}(t,s)&=\EX{{\rm cov}\left(Z^{(k)}(t), Z^{(k)}(s)|\bZ^{(k)}(t)\right)}\\
&\quad+{\rm cov}\left(\EX{Z^{(k)}(t)|\bZ^{(k)}(t)},\EX{Z^{(k)}(s)|\bZ^{(k)}(t)}\right)\text{ .}
\ea
The first term on the right hand side of~\eqref{sigmast} is zero, since
$Z^{(k)}(t)$ is non-random given $\bZ^{(k)}(t)$.
Now, $\EX{Z^{(k)}(t)|\bZ^{(k)}(t)}=\TR{\myvec{1}}\bZ^{(k)}(t)$ and $\EX{Z^{(k)}(s)|\bZ^{(k)}(t)}=\linebreak\TR{\bZ^{(k)}(t)}\mygmat{M}(s-t)\ones$, so
\ba
\label{sigmast1}
\sigma^{(k)}(t,s)&={\rm cov}\left(\TR{\myvec{1}}\bZ^{(k)}(t),\TR{\bZ^{(k)}(t)}\mygmat{M}(s-t)\ones\right)\\
&=\TR{\myvec{1}}\mymat{V}^{(k)}(t)\mygmat{M}(s-t)\ones\\
&=\mu_{\tilde{D}-2}^{-1}\left({\rm e}^{r(s-t)}-{\rm e}^{-\gamma(s-t)}\right)\TR{\myvec{1}}\mymat{V}^{(k)}(t)\myvec{n}+{\rm e}^{-\gamma(s-t)}v^{(k)}(t)
\text{ ,}
\ea
using~\eqref{meanM}, the first equation in~\eqref{n1} and noting that $\TR{\myvec{1}}\mymat{V}^{(k)}(t)\myvec{1}=v^{(k)}(t)$.  This leads to an exact, closed-form expression for the covariance function in terms of elementary functions, which we state below and derive in Appendix~\ref{app:covariance}.

For $t,s\ge 0$, let
$\myvec{\sigma}(t,s)=\TR{\left(\sigma^{(0)}(t,s),\sigma^{(1)}(t,s),\ldots,\sigma^{(\kmax)}(t,s)\right)}$.  Then
\be
\label{sigmavec}
\myvec{\sigma}(t,s)=\mu_{\tilde{D}-2}^{-1}\left({\rm e}^{r(s-t)}-{\rm e}^{-\gamma(s-t)}\right)\left(\beta_1(t) \myvec{n} +  \beta_2(t) \myvec{n}_2\right)+{\rm e}^{-\gamma(s-t)}\myvec{v}(t)
\text{ ,}
\ee
where
\ba
\label{betas}
\beta_1(t) &= \gamma \mu_{\tilde{D}-2}^{-1}\left[\mu_{\tilde{D}}^{-1}\mu_{(\tilde{D}-1)^2+1}(I_7(t)-I_8(t))+\mu_{\tilde{D}-2}I_6(t)\right]\\
&\quad+\tau\mu_{\tilde{D}-2}^{-1}\left[\mu_{(\tilde{D}-2)^2}I_7(t)+ \mu_{\tilde{D}-2}^2 I_6(t)\right]\text{ ,}\\
\beta_2(t) &= \gamma \mu_{\tilde{D}-2}^{-1} I_8(t)\text{ ,}
\ea
with
\ba
\label{integrals1}
I_6(t)&=\frac{{\rm e}^{rt}\left(1-{\rm e}^{-\gamma t}\right)}{\gamma} \text{ ,}\\
I_7(t)&=\frac{{\rm e}^{rt}\left({\rm e}^{rt}-1\right)}{r}-I_6(t) \text{ ,}\\
I_8(t)&=\frac{{\rm e}^{2rt}-{\rm e}^{-(\gamma+2\tau)t}}{2r+2\tau+\gamma}-\frac{{\rm e}^{-\gamma t}\left({\rm e}^{rt}-{\rm e}^{-2\tau t}\right)}{r+2\tau}
\text{ .}
\ea
(As at~\eqref{integrals}, an appropriate limit is taken if a denominator in~\eqref{integrals1} is zero.)

The covariance function takes a simple form in the limit as $t$ and $s \to \infty$.
Note that $\lim_{t \to \infty} {\rm e}^{-2rt} I_6(t)=0$, $\lim_{t \to \infty} {\rm
e}^{-2rt} I_7(t)=\frac{1}{r}$ and  $\lim_{t \to \infty} {\rm e}^{-2rt}
I_8(t)=\frac{1}{2\mu_{\tilde{D}-1}\tau-\gamma}$.  Substituting these limits
into~\eqref{sigmavec} yields that, for any  $s\ge0$,
\be
\label{asymcov}
\lim_{t \to \infty} {\rm e}^{-2rt}\myvec{\sigma}(t,t+s)={\rm e}^{rs} \lim_{t \to \infty} {\rm e}^{-2rt}\myvec{v}(t)\text{ .}
\ee
It follows that, for $k\in\mathcal{K}$ and $s>0$,
\be
\nonumber
\lim_{t \to \infty} {\rm corr}\left(Z^{(k)}(t),Z^{(k)}(t+s)\right)=1\text{ ,}
\ee
where ${\rm corr}$ denotes correlation.  This is not surprising since it is
well known that
\be
\nonumber
%\label{asconvZk}
{\rm e}^{-rt} Z^{(k)}(t) \xrightarrow{\text{a.s.}} W^{(k)} \quad\text{as } t \rightarrow \infty\text{ ,}
\ee
where $\xrightarrow{\text{a.s.}}$ denotes almost sure convergence (i.e.
convergence with probability $1$) and $W^{(k)}$ is a non-negative random which
satsifies $W^{(k)}=0$ if and only if the branching process goes extinct; see
e.g.~Athreya and Ney~\cite{Athreya:1972}, Theorem V.7.2.

\section{Unbounded degree distributions}
\label{sec:unbounded}

The above results have all assumed that there is a maximum degree $\kmax$.
Suppose that is not the case, so the branching process has countably many
types.  For $t\ge 0$, let $\bZ(t)=\TR{(Z_0(t), Z_1(t),\ldots)}$,
where $Z_i(t)$ denotes the number of individuals of type $i$ alive at time $t$,
and let $Z(t)=\sum_{i=0}^\infty Z_i(t)$ denote the total number individuals alive
at time $t$. (For ease of notation we drop explict reference to the type of the
initial individual.)  For $\kmax=1,2,\ldots$, let $\{\bZ(t,\kmax):t \ge 0\}$
denote the branching process derived from $\{\bZ(t):t \ge 0\}$ by ignoring all
individuals having type strictly greater than $\kmax$ and any offspring of such
individuals.  For $t \ge 0$, let $Z(t, \kmax)=\TR{\ones}\bZ(t,\kmax)$ be the total number of individuals alive in
$\{\bZ(t,\kmax):t \ge 0\}$ at time $t$.  Now, for any $t \ge 0$, $Z(t, \kmax)$
is monotonically increasing in $\kmax$ and converges almost surely to $Z(t)$ as
$\kmax \to \infty$.  Thus, by the monotone convergence theorem,
$\EX{Z(t)}=\lim_{\kmax \to \infty}\EX{Z(t,\kmax)}$.

The process $\{\bZ(t,\kmax):t \ge 0\}$ behaves like the branching process
described in Section~\ref{subsec:BP} but with infection rate $\tau$ replaced by
$\tau(\kmax)=\tau\mathbb{P}\left(\tilde{D} \le \kmax+1\right)$, and size-biased
degree distribution $\tilde{D}$ replaced by $\tilde{D}(\kmax)$, where
\be
\nonumber
\mathbb{P}\left(\tilde{D}(\kmax)=k\right)= \begin{cases}
	\frac{\dd_k}{\mathbb{P}\left(\tilde{D} \le \kmax+1\right)} & \text{ if } k=1,2,\ldots,\kmax+1 \text{ ,}\\
	0 & \text{ if } k= \kmax+2,\kmax+3,\ldots  \text{ .}
\end{cases}
\ee
The presence of $\kmax+1$ rather than $\kmax$ is because contacts with
individuals having degree strictly greater than $\kmax+1$ are ignored, as they
yield individuals with effective degree (and hence type) strictly greater than
$\kmax$.  Now $\tau(\kmax) \to \tau$ and $\EX{\tilde{D}(\kmax)} \to
\EX{\tilde{D}}$ as $\kmax \to \infty$, so the expression~\eqref{meantot} for
the mean total population size at time $t$ continues to hold in the unbouded
degree case, provided that $\EX{\tilde{D}}<\infty$, or equivalently that
$\EX{D^2}<\infty$.  A similar argument shows that the expressions for the
variance of $Z(t)$ and the covariance of $Z(t)$ and $Z(s)$, derived in
Sections~\ref{sec:var} and~\ref{sec:cov}, respectively, continue to hold
provided $\EX{\tilde{D}^2}<\infty$, or equivalently $\EX{D^3}<\infty$.
%In this context, we note that the weakest conditions obtained on the moments
%of the degree distribution for convergence of the mean of the full stochastic
%epidemic onto its deterministic limit are given by Janson {\it et
%al.}~\cite{Janson:2014}, who require uniform boundedness of the second moment
%of $D$.

\section{Probability of extinction}
\label{sec:ext}

For $t \ge 0$ and $k\in\mathcal{K}$, let
$\pi_k(t)=\mathbb{P}\left(Z^{(k)}(t)=0\right)$ be the probability that the
branching process is extinct at time $t$ given that it starts with one
individual of type $k$.  Then in general
\be
\ddt{}\pi_k(t) = -\omega_k \pi_k(t) + \omega_k P_k(\mygvec{\pi}(t)) \text{ ,}
\ee
where $\mygvec{\pi}(t)=(\pi_i(t))$. For our specific model, using~\eqref{omegak} and~\eqref{PGFPk}, we have
\be
\ddt{}\pi_k(t) = -(\gamma + \tau k) \pi_k(t)
+ \gamma + \tau k \pi_{k-1}(t) \sum_{l=0}^{\kmax-1} \dd_{l+1} \pi_l(t) \text{ .} \label{qqext}
\ee
These equations are not amenable to closed-form solution. Note, however, that
studies of time to extinction for network epidemics --
e.g.~Holme~\cite{Holme:2013} -- have tended to be based on Monte Carlo methods,
but~\eqref{qqext} could provide a complementary approach that is numerically
cheaper and more analytically tractable.

We will now consider three regimes in which asymptotic methods can be used to
bound the real-time behaviour of the probabilities of extinction. In
particular, we will see that early real-time behaviour is bounded by the death
rates $\omega_k$, while late-time behaviour is bounded by the asymptotic
real-time growth rate $r$ provided $r>-\gamma$.

\subsection{Late behaviour of the subcritical case}

Suppose that $r<0$, so the branching process is subcritical.  For $t \ge 0$ and
$k\in\mathbb{N}_0=\{0,1,\ldots\}$, we will work with the probability of
survival $q_k(t)=1-\pi_k(t)=\mathbb{P}\left(Z^{(k)}(t)>0\right)$. Now
$\mathbb{P}\left(Z^{(k)}(t)>0\right) \le \EX{Z^{(k)}(t)}$, so
using~\eqref{meantot} a simple upper bound for $q_k(t)$, valid also in the
unbounded degree setting using the results in Section~\ref{sec:unbounded}, is
\be
\label{qkupper}
q_k(t) \le k\mu_{\tilde{D}-2}^{-1}\left({\rm e}^{rt} - {\rm e}^{-\gamma t}
\right) + {\rm e}^{-\gamma t} \text{ .}
\ee
Note that $\mu_{\tilde{D}} < \infty$ is a necessary condition for $r<0$.
Under
the stronger condition that $\mu_{\tilde{D}^2}< \infty$,
Windridge~\cite{Windridge:2014} gives an exponential approximation, for large $t$, to
a quantity closely related to $q_k(t)$.  He assumes that what we call type-$0$
individuals are dead.  For $k=1,2,\ldots$, let $\hat{q}_k(t)=\mathbb{P}\left(\sum_{i=1}^\infty
Z_i^{(k)}(t)>0\right)$.  Then, Windridge shows that there exists a constant
$\hat{c} \in (0,1]$ such that, for any $a < \min\{\tau, -r\}$,
\be
\hat{q}_k(t) =\hat{c} k {\rm e}^{rt}\left(1+O(k{\rm e}^{-at})\right) \quad
\mbox{as } t \to \infty \text{ ,} \label{windq}
\ee
for any $k \ge 1$.  The constant $\hat{c}=\lim_{t\to \infty}{\rm e}^{-r t}
\hat{q}_1(t)$.  Note that for some practical purposes, $\hat{q}_k(t)$ may be of more
interest than $q_k(t)$, since type-$0$ individuals are unable to transmit
infection.  In particular, in Appendix~\ref{app:survival} we sketch the
argument that for the case where $r>-\gamma$ an analogous result
to~\eqref{windq} holds, i.e.~, for $k \ge 1$,
\be
{q}_k(t)\sim c k {\rm e}^{rt} \quad
\text{as}\quad t \to \infty \text{ ,}\quad
\text{where} \quad {c}=\lim_{t\to \infty}{\rm e}^{-r t}
q_1(t)>0 \text{ .} \label{ourq}
\ee
(For real-valued functions, $f$ and $g$ say, $f(t) \sim g(t)$ as $t \to \infty$ if
$\lim_{t \to \infty} f(t)/g(t)=1$.)

For the case where $r<-\gamma$ (so, from~\eqref{growthrate}, $\mu_{\tilde{D}}<2$) we show in Appendix~\ref{app:survival} that if $\mu_{\tilde{D}^2}< \infty$ then, for $k \ge 0$,
\be
{q}_k(t)\sim (1-k\mu_{\tilde{D}-2}^{-1}){\rm e}^{- \gamma t} \quad
\text{as}\quad t \to \infty \text{ .}\label{ourq1}
\ee
Note that in this case the asymptotic behaviour of the survival probability ${q}_k(t)$ is independent of the infection rate $\tau$.
The case $r<-\gamma$ could occur, for
example, at the end of an epidemic where $\tau \gg \gamma$. Such an epidemic
would consist primarily of transmission events at early times, with the late
behaviour dominated by recovery events.

\subsection{Late behaviour of the supercritical case}

An approximation to  $q_k(t)$ in the supercritical case ($r>0$) can
be obtained by exploiting the fact that a supercritical branching process
conditioned on extinction is probabilistically equivalent to a subcritical
branching process.  For $k\in\mathbb{N}_0$, let $T^{(k)}=\inf\{t \ge
0:Z^{(k)}(t)=0\}$ denote the extinction time of the branching process given
that it starts with one individual of type $k$, where $T^{(k)}=\infty$ if the
branching process survives forever, and let
$\pi_k=\mathbb{P}\left(T^{(k)}<\infty\right)=\pi_k(\infty)$ be the probability
that the branching process ultimately goes extinct.  Then,
\be
\label{qksuper}
q_k(t)=1-\pi_k+\pi_k \mathbb{P}\left(T^{(k)}>t|T^{(k)}<\infty\right)\text{ .}
\ee
Let $\{\tilde{\myvec{Z}}^{(k)}(t):t \ge 0 \}$ be distributed as
$\{\myvec{Z}^{(k)}(t):t \ge 0 |T^{(k)}<\infty \}$.  Then it follows from
Waugh~\cite{Waugh:1958}, Section 5, that $\{\tilde{\myvec{Z}}^{(k)}(t):t \ge 0
\}$ is also a continuous-time multitype Markov branching process, in which the
lifetime of a typical type-$k$ individual has an exponential distribution with
rate $\gamma+\tau k$, as at~\eqref{omegak}, but when it dies its offspring is
now distributed as follows:
\ba
\mathbb{P}\left(\mathrm{Offspring} = \varnothing | \mathrm{Parent\ type} = k \right) &
 = \frac{1}{\pi_k}\frac{\gamma}{\gamma + \tau k} \text{ ,} \\
 \mathbb{P}\left(\mathrm{Offspring} = \{k-1,l\} | \mathrm{Parent\ type} = k \right) &
 = \frac{\pi_{k-1} \pi_l}{\pi_k}\frac{\tau k \dd_{l+1}}{\gamma + \tau k}
 \quad(l\in\mathbb{N}_0) \text{ .}
\ea
Suppose now that there is a maximum degree size $\kmax$.  Let
$\tilde{\mygmat{\Omega}}=[\tilde{\Omega}_{l,k}]$ be the $(\kmax+1) \times
(\kmax+1)$ matrix with elements given by
\be
\nonumber
\tilde{\Omega}_{l,k} = \frac{\pi_{l-1}\pi_k}{\pi_l}\tau l \left(\dd_{k+1} + \delta_{l,k+1}\right)
-(\gamma + \tau l) \delta_{l,k} \qquad (l,k\in\mathcal{K}) \text{ .}
\ee
Then, recalling~\eqref{meanM},
\be
\EX{\tilde{Z}^{(k)}(t)}=\TR{\myvec{u}_k}{\rm e}^{\tilde{\mygmat{\Omega}}t} \ones \text{ ,}
\ee
where $\tilde{Z}^{(k)}(t)=\tilde{Z}^{(k)}_0(t)+\tilde{Z}^{(k)}_1(t)+\ldots+\tilde{Z}^{(k)}_{\kmax}(t)$.

Let $\tilde{r}$ denote the dominant eigenvalue of $\tilde{\mygmat{\Omega}}$ and
note that $\tilde{r}<0$.  For $t \ge 0$ and $k=0,1,\ldots, \kmax$, let
$\tilde{q}_k(t)=\mathbb{P}\left(\tilde{Z}^{(k)}(t)>0\right)$ be the probability
that the branching process $\{\tilde{\myvec{Z}}^{(k)}(t):t \ge 0 \}$ is not
extinct at time $t$ given that it starts with one individual of type $k$.  Then
we expect that arguments similar to those used in the proof of
Heinzmann~\cite{Heinzmann:2009}, Theorem 3.1, will show that there exists
constants $\tilde{c}_1, \tilde{c}_2,\ldots,\tilde{c}_{\kmax}$, satisfying
$0<\tilde{c}_k<\infty$, such for $k=1,2,\ldots,\kmax$,
\be
\label{qkcond}
\tilde{q}_k(t) =\tilde{c}_k {\rm e}^{\tilde{r}t}\left(1+o({\rm e}^{-\tilde{\gamma}t})\right) \quad \mbox{as } t \to \infty \text{ ,}
\ee
for any $\tilde{\gamma}>0$.  It then follows using~\eqref{qksuper} that
\be
\label{qksuperapprox}
q_k(t)=1-\pi_k+\pi_k \tilde{c}_k {\rm e}^{\tilde{r}t}\left(1+o({\rm
e}^{-\tilde{\gamma}t})\right) \quad \mbox{as } t \to \infty \text{ .}
\ee
Heinzmann~\cite{Heinzmann:2009}, Theorem 3.1, cannot be applied directly as it
assumes that the matrix $\tilde{\mygmat{\Omega}}$ is irreducible.  We do not
consider it here but we expect that Heinzmann's proof can be extended to our
situation.  If we assume that type-$0$ individuals are dead and only consider
initial individuals of types $1,2,\ldots,\kmax-1$  (recall that only initial
infectives can have type $\kmax$) then $\tilde{\mygmat{\Omega}}$ becomes a
$(\kmax-1) \times (\kmax-1)$ irreducible matrix.
Heinzmann~\cite{Heinzmann:2009}, Theorem 3.1, then yields~\eqref{qkcond}; note
that now $\pi_k$ is replaced by $\bar{\pi}_k=1-\lim_{t \to \infty}
\bar{q}_k(t)$ $(k=1,2,\kmax-1)$ in the definition of $\tilde{\mygmat{\Omega}}$
and
$\tilde{q}_k(t)=\tilde{q}_k(t)=\mathbb{P}\left(\sum_{i=1}^{\kmax-1}\tilde{Z}^{(k)}_i(t)>0\right)$.
The approximation~\eqref{qksuperapprox} then holds with $q_k(t)$ and $\pi_k$
replaced by $\bar{q}_k(t)$ and $\bar{\pi}_k$, repsectively.  Moreover, if we
then let $\TR{\tilde{\myvec{f}}_1}$ and $\tilde{\myvec{b}}_1$ be left and right
eigenvectors of $\tilde{\mygmat{\Omega}}$ corresponding to the eigenvalue
$\tilde{r}$, satisfying  $\TR{\tilde{\myvec{f}}_1}\tilde{\myvec{b}}_1=1$, then
$\tilde{c}_k=\left(\TR{\myvec{u}_k}\tilde{\myvec{b}}_1\right)h^*$, where
$h^*=\lim_{t \to \infty} {\rm
e}^{-\tilde{r}t}\TR{\tilde{\myvec{f}}_1}\tilde{\myvec{q}}(t)$ and
$\tilde{\myvec{q}}(t)=\TR{\left(\tilde{q}_i(t),\tilde{q}_2(t), \ldots,
\tilde{q}_{\kmax -1}(t)\right)}$.  Unfortunately, unlike with
$\mygmat{\Omega}$, there do not appear to be closed-form expressions for
$\tilde{r}$ and its associated eigenvectors.

\subsection{Early behaviour and matched asymptotics}

Matched asymptotics is a standard technique in mathematical biology for writing
down approximations to non-linear models that match known asymptotic behaviour
\cite{Murray:2002,Murray:2003}. While numerical solution of the
ODEs~\eqref{qqext} is efficient (as we have noted above) we now obtain a crude
approximation to the full system that takes a closed form in terms of
elementary functions.

First note that for $k\in\mathcal{K}$,  $\pi_k(0) = 0$ and $\pi_k(t)$ is
monotonically increasing with $t$. If we neglect the quadratic terms in $\pi$
in~\eqref{qqext} then, since these are only positive and increasing over time,
we get a lower bound for the extinction probabilities:
\be
\pi_k(t) \geq \pi_k^{(0)}(t)
 = \frac{\gamma}{\omega_k} \left( 1 - {\rm e}^{-\omega_k t} \right)
\text{ .}
\ee
Note that in standard matched asymptotics, we would identify a small parameter
from a ratio of rate constants as the basis for a systematic approximation
scheme~\cite{Murray:1984}; an alternative would be to approximate
systematically by, for example, letting $\pi_k^{(1)} = \pi_k - \pi_k^{(0)}$,
substituting into~\eqref{qqext} and neglecting quadratic terms to give a linear
set of equations for the next order of approximation. Here we consider only the
lowest order approximation, and hence define an `internal' solution for the
survival probability as:
\be
q_k^{(I)}(t) = 1-\pi_k^{(0)}(t) =
 \frac{1}{\omega_k} \left( \tau k + \gamma {\rm e}^{-\omega_k t} \right)
\text{ .}
\ee
Next, supposing we are in the subcritical case so that our
result~\eqref{ourq} holds. We will call this the `external' solution
\be
q_k^{(E)}(t) = ck{\rm e}^{r t} \text{ .}
\ee
To fix the constant $c$, we match the late behaviour of the internal solution
with the early behaviour of the external solution:
\be
\overline{q} = \lim_{t\rightarrow \infty} q_k^{(I)}(t) = \lim_{t\rightarrow 0} q_k^{(E)}(t)
\quad \Rightarrow \quad q_k^{(E)}(t) = \frac{\tau k}{\omega_k} {\rm e}^{r t} \text{ .}
\ee
Finally, the matched asymptotic solution is
\be
q_k^{(A)}(t) = q_k^{(I)}(t) + q_k^{(E)}(t) - \overline{q} =
\frac{1}{\omega_k} \left( \tau k {\rm e}^{r t} + \gamma {\rm e}^{-\omega_k t} \right)
\text{ .}
\ee
We compared this approximation as well as the internal and external solutions
to the exact solution $q_k(t)$, with results shown in Figure~\ref{fig:ma}. As
advertised, this is a relatively crude approximation, but is expressed in terms
of elementary functions and satisfies known asymptotic limits.

\section{Fluctuations in the emerging phase of a major outbreak}
\label{sec:fluc}

We now consider the early behaviour of supercritical epidemics that take off
(i.e.\ do not go extinct early on but ultimately end owing to long-term depletion of
susceptibles).  The early stages of such an epidemic are approximated by the
branching process defined in Section~\ref{subsec:BP} but conditioned on
non-extinction.  It is straightforward to adapt the results on means and
variances in Sections~\ref{sec:means} and~\ref{sec:var} to condition on
$Z^{(k)}(t)>0$.  Elementary calculation shows that, for $t \ge 0$ and
$k\in\mathbb{N}_0$,
\ba
\nonumber
\EX{Z^{(k)}(t)\Big|Z^{(k)}(t)>0}&=\frac{\EX{Z^{(k)}(t)}}{q_k(t)}
\text{ ,}\\
{\rm var}\left(Z^{(k)}(t)\Big|Z^{(k)}(t)>0\right)& =\frac{{\rm
var}\left(Z^{(k)}(t)\right)}{q_k(t)}-\pi_k(t)\left(\frac{\EX{Z^{(k)}(t)}}{q_k(t)}\right)^2
\text{ .}
\ea
Expressions for $\EX{Z^{(k)}(t)\Big|Z^{(k)}(t)>0}$ and ${\rm
var}\left(Z^{(k)}(t)\Big|Z^{(k)}(t)>0\right)$ then follow using~\eqref{meantot}
and~\eqref{vvec1}, respectively, though there is no closed-form formula for
$q_k(t)$ or $\pi_k(t)$.  Note that, assuming $r>0$ so $\pi_k<1$,
\be
\nonumber
\lim_{t \to \infty}\frac{{\rm var}\left(Z^{(k)}(t)\Big|Z^{(k)}(t)>0\right)}{\EX{Z^{(k)}(t)\Big|Z^{(k)}(t)>0}^2}
=\left(1-\pi_k\right)\lim_{t \to \infty}\frac{{\rm var}\left(Z^{(k)}(t)\right)}{\EX{Z^{(k)}(t)}^2}-\pi_k \text{ ,}
\ee
which depends on the degree $k$ of the initial infective. \\

\noindent{}The diffusion approximation studied in Graham and
House~\cite{Graham:2014} corresponds to the case where the number of infectives
at time $t=0$ is large.  Return to  the case where there is a maximum degree
$\kmax$ and suppose that the branching process does not go extinct.  Then it
follows from ~Athreya and Ney~\cite{Athreya:1972}, Theorem V.7.2, that, for any
$k\in\mathcal{K}$,
\be
\nonumber
\frac{\myvec{Z}^{(k)}(t)}{Z^{(k)}(t)} \xrightarrow{\text{a.s.}}
\myvec{w}_1 \quad\text{as } t \rightarrow \infty\text{ ,}
\ee
where $\myvec{w}_1$, given by~\eqref{w1vec}, is a left eigenvector of
$\mygmat{\Omega}$ corresponding to the dominant eigenvalue $\lambda_1=r$.  Thus
if an epidemic takes off and is still in its exponentiallly growing phase then
the relative frequencies of the different types of infectives will be close to
$\myvec{w}_1$.  Hence, we now assume that the initial number of individuals in
the branching process $Z(0)=K$, where $K$ is large, and that $Z_i(0) \approx
w_{1,i} K$ for $i \in \mathcal{K}$.  Label the initial individuals
$1,2,\ldots,K$.  Then, for $t \ge 0$, the total population size is
$Z(t)=\hat{Z}_1(t)+ \hat{Z}_2(t)+\ldots+\hat{Z}_K(t)$, where $\hat{Z}_i(t)$
denotes the total number of descentants of the initial individual $i$ that are
alive at time $t$, including $i$ itself if it is still alive.  Thus,
$\EX{Z(t)}=\sum_{i=1}^K \EX{\hat{Z}_i(t)}$, for all $t \ge 0$, and, since the
processes $\{\hat{Z}_i(t):t \ge 0\}$ $(i=1,2,\ldots,K)$ are mutually
independent, ${\rm var}\left(Z(t)\right)=\sum_{i=1}^K {\rm
var}\left(\hat{Z}_i(t)\right)$, for all $t \ge 0$, and ${\rm cov}\left(Z(t),
Z(s)\right)= \sum_{i=1}^K {\rm cov}\left(\hat{Z}_i(t), \hat{Z}_i(s)\right)$,
for all $t,s \ge 0$. \\

\noindent{}Note that~\eqref{w1fac} implies that
\be
\nonumber
\TR{\myvec{w}_1}\myvec{n}=\sum_{i=0}^{\kmax} i w_{1,i}=\mu_{\tilde{D}-2} \qquad\mbox{and}\qquad \TR{\myvec{w}_1}\myvec{n}_2=
\sum_{i=0}^{\kmax} i^2 w_{1,i}=\frac{\mu_{\tilde{D}-2} \mu_{(\tilde{D}-1)^2+1}}{\mu_{\tilde{D}}} \text{ .}
\ee
Assuming that the above approximation is exact, then, for $t \ge 0$, it follows from~\eqref{meantot} that
\be
\label{meanz}
\EX{Z(t)}=K\TR{\myvec{w}_1}\myvec{m}(t)=K {\rm e}^{rt}
\ee
and, after a little algebra, it follows from~\eqref{vvec1} that
\ba
{\rm var}\left(Z(t)\right) &= K\TR{\myvec{w}_1}\myvec{v}(t)\\
&=K\left(\gamma\left[I_9(t)+\mu_{\tilde{D}}^{-1}\mu_{\tilde{D}-2}^{-1}\left(\sigma_{\tilde{D}}^2+2\right) I_4(t)\right]+
\tau \mu_{\tilde{D}-2}^{-1}\left[\mu_{\tilde{D}-2}^2 I_9(t)+\sigma_{\tilde{D}}^2 I_4(t)\right]\right) \text{ ,} \label{varz}
\ea
where $\sigma_{\tilde{D}}^2={\rm var}(\tilde{D})$ and
\be
\nonumber
I_9(t)=\frac{{\rm e}^{rt}\left({\rm e}^{rt}-1\right)}{r}\text{ .}
\ee
Comparison of~\eqref{meanz} and~\eqref{varz} with the diffusion-based result of
Graham and House~\cite{Graham:2014} in the limit of large $t$ gives agreement
when $\gamma = 0$ (i.e.\ for the SI model) but not for $\gamma>0$. We believe
that this is due to the fact that the diffusion model was only four
dimensional, so a heuristic argument (given in Section 3.3 of Graham and
House~\cite{Graham:2014}, which gave results that were in good agreement with
simulation) about the neighbourhood of an infective node had to be made, in
contrast to the approach here that deals with each effective degree explicitly
and so has $\kmax+1$ dimensions. The argument about the neighbourhood around an
infective tries to account for correlations caused by variability in recovery
times, and so if $\gamma \rightarrow 0$ then these correlations do not exist.

Recent work by Constable and McKane~\cite{Constable:2014} considered the
reduction of high-dimensional stochastic models to low-dimensional diffusions
and this approach was shown to be asymptotically exact for some systems in the
small-noise limit by Parsons and Rogers~\cite{Parsons:2015}.  It is an open
question whether the argument in Section 3.3 of Graham and
House~\cite{Graham:2014} could be justified rigorously by a similar argument,
however we note that a branching-process approach makes fewer assumptions than
a low-dimensional diffusion limit and so will be more generally applicable.\\

\noindent{}Considering further results that can be obtained, it follows
from~\eqref{sigmavec} and a little algebra that, for $0 \le t \le s$,
\ba
\label{covz}
{\rm cov}\left(Z(t), Z(s)\right)&={\rm e}^{-\gamma(s-t)} {\rm var}\left(Z(t)\right)+K\mu_{\tilde{D}-2}^{-1}\left({\rm e}^{r(s-t)}-{\rm e}^{-\gamma(s-t)}\right)\\
&\quad \times
\left\{\gamma \mu_{\tilde{D}}^{-1}\left[\mu_{(\tilde{D}-1)^2+1}I_9(t)-
\left(\sigma_{\tilde{D}}^2+2\right) I_6(t)\right]\right.\\
&\quad +\left. \tau\left[\mu_{\tilde{D}-2}^2 I_9(t)-\sigma_{\tilde{D}}^2 I_6(t)\right]\right\}
\text{ .}
\ea
It seems plausible that these results extend to the case when there is no
maximal degree but that would involve results for countably infinite matrices
which we do not consider here.

Recall that the processes $\{\hat{Z}_i(t):t \ge 0\}$ $(i=1,2,\ldots,K)$ are
mutually independent.  It follows using the central limit theorem that, for
sufficiently large $K$, the process $\{Z(t):t \ge 0\}$, which approximates the
prevalence of infection during the early growth of an epidemic, is
approximately Gaussian with mean function given by~\eqref{meanz} and covariance
function given by~\eqref{covz}.

\section{Numerical examples}
\label{sec:numerical}

\subsection{Forward simulations}

We conducted a series of numerical experiments to provide specific examples of
the general results presented here. $M=10^4$ Monte Carlo simulations were
performed on three different configuration model networks, each of size
$N=10^4$, and with the degree distributions shown in Figure~\ref{fig:restart}
Row (i). (Note that each Monte Carlo simulation consisted of first simulating a network
and then simulating a single epidemic on it.)  Two different scenarios were considered. In the first -- most commonly
considered in the literature when simulations are compared to analytic
approaches -- $\mathrm{time}=0$ was defined as the first point when prevalence
is at a given level, $K$. In our simulations we took $K=100$, but in general
$K$ should take a value where the probability of extinction has become
negligible, but the depletion of the susceptible population has not had a
significant effect on the epidemic dynamics. In the second, each epidemic was
started from one node, picked uniformly at random, so the probability of
extinction played a major role. This scenario is less commonly considered when
comparing real-time simulated epidemics to differential equation models because
the latter are typically designed to hold when the epidemic is already
established.

Since analytic results for the probabilities of extinction $\pi_k(t)$ are not
available, the branching process results required numerical integration of
ordinary differential equations (in our case using Runge-Kutta methods). We
stress that the computational effort required to do this is much less than that
involved in performing Monte Carlo simulations, and has the benefit of not
depending on $N$.

The results for the first approach (restarting time at the first time
prevalence reaches 100) are given in Figure~\ref{fig:restart}. Row (ii) shows
sample trajectories (which all agree on prevalence at time 0).  Row (iii) shows
the simulated mean after time 0 on a logarithmic scale, which initially grows
at the constant rate predicted by the branching process model, and then reduces
as the susceptible population is depleted. Row (iv) shows the variance, which
has not converged to its asymptotic growth regime by the time prevalence is
equal to 100, an effect that is captured by the branching process model. The
variance does not take its largest value at the peak prevalence, but instead
has local maxima before and after the peak.

Figure~\ref{fig:sims} shows the results for the second approach in which there
is one randomly chosen initial infective at time 0. Row (i) shows some sample
trajectories.  Row (ii) shows the extinction probabilities, which are
accurately captured in the branching process model until very close to the end
of the epidemic when prevalence is low and extinction becomes likely again. Row
(iii) shows the mean, which does not start growing at a constant rate with the
convergence rate accurately captured in the branching process model. Row (iv)
shows the variance and convergence onto its asymptotic value; in this case
there is a single maximum just before the peak in prevalence.

Another important point is that while mean numbers infected are comparable
between Figures~\ref{fig:restart} and~\ref{fig:sims}, the variability in
the time for the epidemic to take off, as well as the contribution from
extinct epidemics, makes the real-time variance in Figure~\ref{fig:sims}
orders of magnitude larger than in Figure~\ref{fig:restart}.

\subsection{Statistical inference}

\label{sec:infer}

In order to demonstrate the potential use of the real-time effective degree
branching process model for statistical inference, we carried out a simulation
study. Here we simulated one epidemic that took off on a configuration model network of size
$N=10^6$ with degree distribution $D^{(3)}$ as in the right-hand column of
Figure~\ref{fig:restart} ($d^{(3)}_1 = 1/8$, $d^{(3)}_3 = 5/6$, $d^{(3)}_9 =
1/24$) and true rates ${\tau}_0 = 2$, ${\gamma}_0 = 1$. Letting $I(t)$ be
the prevalence of infection in the network model, we set time $t=0$ when
$I(t)=100$ for the first time and make 40 evenly-spaced observations (with gap
$\delta t = 0.05$ between each) of $I(t)$ up to $t_{\mathrm{end}}=2$.

We then define an approximate likelihood based on the methods of Ross {\it et al.}~\cite{Ross:2006}, in which a Gaussian process approximation based
on known first and second moments is used, which will be more accurate for
larger $N$, larger $I(0)$ and smaller $\delta t$. There should be, however, no
\textit{a priori} obstacle to fitting our model to data on smaller populations
even with incomplete data, for example by using Markov chain Monte Carlo methods to perform
multiple imputation as suggested by O'Neill and Roberts~\cite{ONeill:1999}.

Explicitly, we let the probability density function $f$ for sequential
observations be given by
\be
f(I(t+\delta t)|I(t)) = \mathcal{N}(\mathbb{E}[Z(t + \delta t) | Z(t)=I(t)] ,
\text{var}(Z(t + \delta t) | Z(t)=I(t)) \text{ ,}
\ee
where $\mathcal{N}(m,V)$ is the probability density function of a normal
distribution with mean $m$ and variance $V$, and the expectation and variance
of $Z(t + \delta t)$ are given by the results of Section~\ref{sec:fluc} above. The
likelihood is then
\be
L =
\prod_{t\in\{0,\delta t, \ldots, t_{\mathrm{end}}-\delta t\}}
f(I(t+\delta t)|I(t)) \text{ .}
\ee
We consider values of this likelihood across the range of rate constant
parameters $\tau$ and $\gamma$ under two different degree distributions: the
correct one, $D^{(3)}$, and a misspecified degree distribution $D^{(1)}$, which
is the one used in the left-hand column of Figure~\ref{fig:restart} ($d^{(1)}_3
= 1$).

Figure~\ref{fig:infer} shows the first quarter of the simulated epidemic
together with the Gaussian process approximation, as well as likelihood
surfaces for the correct and misspecified degree distributions.  Performing
maximum likelihood estimation using MATLAB's \texttt{mle()} function allows us
to obtain point estimates for parameters $\hat{\tau}$ and $\hat{\gamma}$, as
well as asymptotic 95\% confidence intervals and the parameter covariance matrix
$\hat{C}$ from the inverse Hessian. We quote results to 2 significant figures;
the asymptotic approximations also give very slightly negative lower confidence
intervals for $\hat{\gamma}$ which we round up to 0.  For the correct degree
distribution we obtain
\be
\begin{pmatrix} \hat{\tau}^{(3)}\\ \hat{\gamma}^{(3)} \end{pmatrix}
= \begin{pmatrix} 1.9 & [1.4,2.4]\\ 0.8  & [0,1.7] \end{pmatrix} \text{ ,}
	\qquad \hat{\boldsymbol{C}}^{(3)} = \begin{pmatrix}
	0.069 & 0.11\\ 0.11 & 0.19 \end{pmatrix} \text{ ,}
\ee
and for the misspecified degree distribution we obtain
\be
\begin{pmatrix} \hat{\tau}^{(1)}\\ \hat{\gamma}^{(1)} \end{pmatrix}
= \begin{pmatrix} 3.2 & [2.3,4.0]\\ 0.8  & [0,1.7] \end{pmatrix} \text{ ,}
	\qquad \hat{\boldsymbol{C}}^{(1)} = \begin{pmatrix}
	0.0045 & 0.011\\ 0.011 & 0.030 \end{pmatrix} \text{ .}
\ee
This shows that knowledge of the correct distribution allows both $\tau$ and
$\gamma$ to be estimated; although as would be expected the early asymptotic
growth rate $r$ is much more closely constrained by simulated data than other
directions in parameter space. It also shows that misspecification of the
degree distribution allows $r$ to be identified, but biases the estimate of, in
this case, $\tau$.

\section{Concluding comments}
\label{sec:conc}

\subsection{Summary of results}

In this paper, we have provided explicit closed-form expressions for the
real-time mean, variance and covariance function for disease prevalence during the early stages of the
Markovian SIR model on a configuration model network, as well as deriving
differential equations for the probabilities of extinction over time that are
relatively numerically cheap to solve. These allow for a more explicit
treatment of e.g.\ rates of convergence to asymptotic behaviour than has
previously been possible.

\subsection{Future directions}

We believe that the methods of real-time, multitype branching processes could
be more widely applied in infectious disease epidemiology, since they provide
results concerning extinction and variance that are not available using
deterministic differential equation models. For example, the effective-degree
based methodology presented here may be extended to include degree correlation
(e.g.\ in the sense of Newman~\cite{Newman:2002b}) by keeping track of the
actual, as well as effective, degrees of individuals, though the type space
becomes larger and explicit analytic results are unlikely to be available.  We
note that there is increasing interest in the eradication of infections
(e.g.~Klepac {\it et al.}~\cite{Klepac:2012}) and that arguably calculating
extinction probabilities and variability in outbreak sizes is of equal or
greater importance in this context than calculation of mean behaviour.

The explicit closed-form expressions derived have the potential to enhance
statistical work on epidemic prevalence curves. In particular, many empirically
observed epidemics of human pathogens exhibit more variability around the trend
than simple models would predict (see Black {\it et al.}~\cite{Black:2014},
particularly Section 1, for a discussion of this), which can bias parameter
estimation if an insufficiently variable model is used. Application of our
methods to real data would be an interesting extension of our work.

The possibility of a more general non-Markovian stochastic epidemic being
approximated by an appropriate real-time branching process is raised by the
results of Barbour and Reinert~\cite{Barbour:2013} and it would be interesting
to investigate whether our analysis could be adapted to this scenario.

Finally, there is the question of low-dimensional PGF-based modelling of the
whole network epidemic that incorporates stochasticity accurately. For example,
the work of Miller~\cite{Miller:2014a} considered accounting for early
fluctuations and Graham and House~\cite{Graham:2014} considered a diffusion
approximation once early fluctuations were negligible, but the results
presented here as well as those of Barbour and Reinert~\cite{Barbour:2013}
suggest that a more unified low-dimensional stochastic approach that explicitly
models early fluctuations may be possible.

\appendix
\section{Convergence of moments}
\label{app:moments}

We determine sufficient conditions for the first two moments of the number of infectives in
the epidemic $E_N$ among a population of $N$ individuals to converge to the corresponding moments of the limiting branching process $\mathcal{B}$.
For ease of exposition, in $E_N$ we assume that at time $t=0$ there is one infective and the remaining $N-1$ individuals are susceptible.  The initial infective is chosen by sampling a stub uniformly at random from all stubs used to form the network, with the individual attached to that stub being the initial infective.  The arguments are easily extended to other choices of initial infective(s).

In the independent and identically distributed (i.i.d.) degree case, we assume that a single sequence $D_1,D_2,\dots$ of i.i.d.~copies of $D$ is used to construct a sequence of epidemics $(E_N)$, where, for $N=1,2,\dots$, the epidemic $E_N$ is constructed using $D_1,D_2,\dots, D_N$.  In the prescribed degree case (see Section~\ref{subsec:model}), recall that
$p_k^{(N)}$ $(k=0,1,\dots)$ denotes the empirical degree distribution in the epidemic
$E_N$ and, for $f:\mathbb{R} \to \mathbb{R}$, let $\mu_{f\left(D^{(N)}\right)}=\sum_{k=0}^{\infty} p_k^{(N)} f(k)$.

For $N=1,2,\dots$ and $t \ge 0$, let $Y_N(t)$ be the number of infectives in
$E_N$ at time $t$ and let $Z(t)$ be the number of individuals in the limiting
branching process $\mathcal{B}$.  Then arguing as in the proof of Ball and
Neal~\cite{Ball:2008}, Theorem A.1, shows that in the i.i.d.~degree case, if
$\mu_D <\infty$ then the sequence of epidemics $(E_N)$ and the limiting
branching process $\mathcal{B}$ can be constructed on a common probability
space so that , with probability one, for any $t > 0$, $Y_N(u)$ and $Z(u)$
coincide for all $u \in [0,t]$ for all sufficiently large $N$.  The same
conclusion holds in the prescribed degree case provided $p^{(N)}_k \to p_k$
$(k=0,1,\ldots)$ and $\mu_D^{(N)} \to \mu_D$ as $N \to \infty$, where
$\sum_{k=0}^{\infty} p_k=1$  and $\mu_D < \infty$.  Thus, under these
conditions, in both cases, for any $t \ge 0$, $Y_N(t)$ converges almost surely
to $Z(t)$ as $N \to \infty$.   We obtain further conditions, under which, for
fixed $t \ge 0$, the sequence $\left(Y_N(t)^2\right)$ is uniformly integrable,
which then (e.g.~Grimmett and Stirzaker~\cite{Grimmett:2001}, Chapter 7,
Section 10) implies immediately that $\lim_{N \to \infty}\EX{Y_N(t)} =
\EX{Z(t)},\lim_{N \to \infty}\EX{Y_N(t)^2} = \EX{Z(t)^2}$ and, for any $0 \le s
\le t$,  $\lim_{N \to \infty}{\rm cov}\left(Y_N(s),Y_N(t)\right) = {\rm
cov}\left(Z(s), Z(t)\right)$. To show that $\left(Y_N(t)^2\right)$ is uniformly
integrable it is sufficient to show that the sequence
$\left(\EX{Y_N(t)^{2+\delta}}\right)$ is bounded above for some $\delta>0$.

For ease of exposition, we use notation from the i.i.d.~degree case.  The construction in Ball and Neal~\cite{Ball:2008} involves
for each $N$ constructing a realisation of a branching process, $\mathcal{B}_N$ say, which is defined analagously to $\mathcal{B}$ but using the empirical distribution of $D_1,D_2,\dots, D_N$ rather than the distribution of $D$.  In $\mathcal{B}_N$, for each birth a stub is chosen independently and uniformly from all the $D_1+D_2+\dots+D_N$ stubs and the degree of the individual that the chosen stub belongs to gives the degree of the individual born at that birth. The process of infectives in $E_N$ follows $\mathcal{B}_N$ except when (i) a sampled stub has previously been chosen, in which case stubs are resampled until one that has not been chosen previously is obtained, or (ii) a sampled  stub has not been chosen previously but is attached to an individual that has already been infected, in which case the corresponding birth and all descendants of that individual in $\mathcal{B}_N$ are ignored in $E_N$.  Note that (i) implies that $\mathcal{B}_N$ need not be an almost sure upper bound for the process of infectives in $E_N$.  The branching processes $\mathcal{B}_N$ $(N=1,2,\dots)$ and $\mathcal{B}$ are coupled so that with probability one, for any fixed $t>0$, $\mathcal{B}_N$ and $\mathcal{B}$ coincide over $[0,t]$ for all sufficiently large $N$.

Let $D_{(1)},D_{(2)},\dots,D_{(N)}$ be the order statistics of $D_1,D_2,\dots, D_N$, i.e.~$D_1,D_2,\dots, D_N$ arranged in increasing order.   For $\epsilon \in (0,1)$, let $\mathcal{B}_{N,\epsilon}$ be the branching process that is defined analagously to
$\mathcal{B}$, but using the empirical distribution of $D_{([N\epsilon]+1)},D_{([N\epsilon]+2)},\dots, D_{(N)}$.  (For $x \in \mathbb{R}$, $[x]$ denotes the greatest integer $\le x$.)  For $t \ge 0$, let
$Z_ {N,\epsilon}(t)$ be the number of individuals alive in $\mathcal{B}_{N,\epsilon}$ at time $t$ and let $T_{N,\epsilon}(t)$ denote the total progeny of $\mathcal{B}_{N,\epsilon}$ by time $t$, including the initial ancestor.  Then $Y_N(t)\overset{st}{\le}Z_ {N,\epsilon}(t)$, provided that $T_{N,\epsilon}(t) \le N\epsilon$, where $\overset{st}{\le}$ denotes stochastically smaller than. (Up until $[N \epsilon]$ infections have occurred in $E_N$, the empirical distribution of the degrees of unsampled stubs, where the degree of a stub is the degree of the individual to which it is attached, is stochastically smaller
than that of the stubs belonging to the top $N-[N \epsilon]$ individuals when ordered by degree.)  As $Y_N(t)$ is at most $N$,
it follows that
\begin{equation*}
\EX{Y_N(t)^{2+\delta}} \le \EX{Z_ {N,\epsilon}(t)^{2+\delta}}+N^{2+\delta}\mathbb{P}\left(T_{N,\epsilon}(t)>N\epsilon\right).
\end{equation*}
By Markov's inequality,
\begin{equation*}
\mathbb{P}\left(T_{N,\epsilon}(t)>N\epsilon\right)\le\frac{1}{(N\epsilon)^{2+\delta}}\EX{T_{N,\epsilon}(t)^{2+\delta}}.
\end{equation*}
Also, $Z_ {N,\epsilon}(t) \overset{st}{\le} T_ {N,\epsilon}(t)$, so
\begin{equation}
\label{YNUniintbound}
\EX{Y_N(t)^{2+\delta}} \le \left(1+\epsilon^{-(2+\delta)}\right)\EX{T_{N,\epsilon}(t)^{2+\delta}}.
\end{equation}

We now bound $\EX{T_{N,\epsilon}(t)^{2+\delta}}$, for $\delta \in (0,1)$.  Note
that this moment is smaller than the corresponding moment for the branching
process in which $\gamma=0$ and individuals retain their original effective
degree throughout their
lifetime.  Moreover, by rescaling the time axis, we can assume without loss of generality that $\tau=1$.  Thus consider
a multitype Markov birth process, with types $0,1,\dots,J$, in which an individual of type $i$ gives birth at rate $i$ and
the types of successive births are i.i.d.~with probability mass function
$\hat{p}_j$ $(j=0,1,\dots,J)$.  For $i=0,1,\dots,J$, let $\hat{T}_i(t)$ denote the total number of individuals alive in this process at time $t$ given that at time $0$ there is one indivdiual, whose type is $i$.  For $\alpha>0$ and $t \ge 0$, let
$\hat{\mu}_i^{(\alpha)}(t)=\EX{\hat{T}_i(t)^{\alpha}}$ $(i=0,1,\dots,J)$ and let $\hat{\mu}^{(\alpha)}(t)=\sum_{i=0}^J \hat{p}_i
\hat{\mu}_i^{(\alpha)}(t)$.  It is possible using a backward argument to derive explicit expressions for $\hat{\mu}_i^{(k)}(t)$
for $k=1,2,\dots$, though the algebra soon becomes very tedious.  As our aim is to bound $\hat{\mu}^{(2+\delta)}(t)$, we simply derive bounds for $\hat{\mu}_i^{(1)}(t),\hat{\mu}_i^{(2)}(t)$ $(i=0,1,\dots,J)$ and finally $\hat{\mu}^{(2+\delta)}(t)$.  Moreover, our bounds are deliberately coarse to facilitate easy application to $\EX{T_{N,\epsilon}(t)^{2+\delta}}$.  For $f:\mathbb{R} \to \mathbb{R}$, let $\hat{\mu}_{f(D)}=\sum_{i=0}^J f(i) \hat{p}_i$.

For $i=0,1,\dots,J$, the backward equation for $\hat{\mu}_i^{(1)}(t)$ is
\begin{equation}
\label{backmuhat1i}
\ddt{\hat{\mu}_i^{(1)}}=i\hat{\mu}^{(1)}, \qquad \hat{\mu}_i^{(1)}(0)=1.
\end{equation}
Multiplying~\eqref{backmuhat1i} by $\hat{p}_i$ and summing over $i=0,1,\dots,J$, yields
\begin{equation*}
\ddt{\hat{\mu}^{(1)}}=\hat{\mu}_D\hat{\mu}^{(1)}, \qquad \hat{\mu}^{(1)}(0)=1.
\end{equation*}
Thus
\be
\label{muhat1}
\hat{\mu}^{(1)}(t)={\rm e}^{\hat{\mu}_D t},
\ee
which on substituting into~\eqref{backmuhat1i} yields $\hat{\mu}_i^{(1)}=1+i\int_0^t {\rm e}^{\hat{\mu}_D s} \dint{s}$, so
\be
\label{muhat1i}
\hat{\mu}_i^{(1)}\le i(1+t){\rm e}^{\hat{\mu}_D t}\qquad (i=1,2,\dots,J).
\ee

For $i=0,1,\dots,J$, the backward equation for $\hat{\mu}_i^{(2)}(t)$ is
\be
\label{backmuhat2i}
\ddt{\hat{\mu}_i^{(2)}}=2i\hat{\mu}_i^{(1)}\hat{\mu}^{(1)}+i\hat{\mu}^{(2)}, \qquad \hat{\mu}_i^{(2)}(0)=1.
\ee
Substituting from~\eqref{muhat1} and~\eqref{muhat1i}, and then multiplying~\eqref{backmuhat2i} by $\hat{p}_i$ and summing over $i=0,1,\dots,J$, yields
\begin{equation*}
\ddt{\hat{\mu}^{(2)}}\le 2\hat{\mu}_{D^2}(1+t){\rm e}^{2\hat{\mu}_D t} +\hat{\mu}_D \hat{\mu}^{(2)}, \qquad \hat{\mu}^{(2)}(0)=1,
\end{equation*}
whence
\begin{eqnarray*}
\label{muhat2}
\hat{\mu}^{(2)}(t)&\le&{\rm e}^{\hat{\mu}_D t}+\int_0^t {\rm e}^{\hat{\mu}_D (t-s)}2\hat{\mu}_{D^2}(1+s){\rm e}^{2\hat{\mu}_D s}\dint{s}\\
&\le& \left(1+2\hat{\mu}_{D^2} t(1+t)\right){\rm e}^{2\hat{\mu}_D t}.
\end{eqnarray*}
Substituting this bound into~\eqref{backmuhat2i} and noting that the right-hand side of~\eqref{backmuhat2i} is increasing in $t$ yields
\be
\label{muhat2i}
\hat{\mu}_i^{(2)}(t) \le i^2 g(t,\hat{\mu}_{D^2}){\rm e}^{2\hat{\mu}_D t} \qquad (i=1,2,\dots,J),
\ee
where
\be
\label{gtD2}
g(t,\hat{\mu}_{D^2})=1+t\left[1+2(1+t)(1+\hat{\mu}_{D^2} t)\right].
\ee

Let $\delta \in (0,1)$.  For $i=0,1,\dots,J$, the backward equation for $\hat{\mu}_i^{(2+\delta)}(t)$ is
\be
\label{backmuhatdeltai}
\ddt{\hat{\mu}_i^{(2+\delta)}}=-i \hat{\mu}_i^{(2+\delta)}(t)+i\EX{\left(\hat{T}_i(t)+\hat{T}(t)\right)^{2+\delta}}, \qquad \hat{\mu}_i^{(2+\delta)}(0)=1,
\ee
where $\hat{T}_i(t)$ and $\hat{T}(t)$ are independent and $\hat{T}(t)$ is distributed as a mixture of $\hat{T}_0(t), \hat{T}_1(t),\dots,\hat{T}_J(t)$ with mixing probabilities $\hat{p}_0, \hat{p}_1,\dots, \hat{p}_J$.

Let $a$ and $b$ be nonnegative real numbers.  Application of the mean
value theorem yields that
\begin{equation*}
(a+b)^{2+\delta} \le a^{2+\delta}+(2+\delta)b(a+b)^{1+\delta}.
\end{equation*}
Further, $(a+b)^{1+\delta} \le
2^{\delta}\left(a^{1+\delta}+b^{1+\delta}\right)$, so
\begin{equation*}
(a+b)^{2+\delta} \le a^{2+\delta}+2^{\delta}(2+\delta)\left(b
a^{1+\delta}+b^{2+\delta}\right).
\end{equation*}
Setting $a=\hat{T}_i(t)$ and $b=\hat{T}(t)$ in this inequality, taking
expectations exploiting the independence of $\hat{T}_i(t)$ and
$\hat{T}(t)$, substituting into~\eqref{backmuhatdeltai} and noting that
$2^{\delta}(2+\delta) <6$ gives
\be
\label{backmuhatdeltai1}
\ddt{\hat{\mu}_i^{(2+\delta)}}\le 6i\left[\hat{\mu}^{(1)}(t)\hat{\mu}_i^{(1+\delta)}(t) +\hat{\mu}^{(2+\delta)}(t)\right], \qquad \hat{\mu}_i^{(2+\delta)}(0)=1.
\ee

Now
\begin{equation*}
\hat{\mu}_i^{(2)}(t)=\EX{\left(\hat{T}_i(t)^{1+\delta}\right)^{\frac{2}{1+\delta}}}
\ge \left(\hat{\mu}_i^{(1+\delta)}(t)\right)^{\frac{2}{1+\delta}},
\end{equation*}
by Jensen's inequality, so
\be
\label{mui1delta}
\hat{\mu}_i^{(1+\delta)}(t)\le \left(\hat{\mu}_i^{(2)}(t)\right)^{\frac{1+\delta}{2}}
\le i^{1+\delta}g(t,\hat{\mu}_{D^2}){\rm e}^{(1+\delta)\hat{\mu}_D t},
\ee
using~\eqref{muhat2i} and noting that $g(t,\hat{\mu}_{D^2})\ge 1$.  Substituting~\eqref{mui1delta} into~\eqref{backmuhatdeltai1}, multiplying by $\hat{p}_i$ and summing over $i=0,1,\dots,J$, yields
\begin{equation*}
\ddt{\hat{\mu}^{(2+\delta)}}\le 6 \hat{\mu}_{D^{2+\delta}}g(t,\hat{\mu}_{D^2}){\rm e}^{(2+\delta)\hat{\mu}_D t}+6 \hat{\mu}_D \hat{\mu}^{(2+\delta)}.
\end{equation*}
Hence,
\begin{eqnarray}
\label{hatmud2delta}
\hat{\mu}^{(2+\delta)} &\le& {\rm e}^{6\hat{\mu}_D t}+\int_0^t {\rm e}^{6\hat{\mu}_D (t-s)} 6 \hat{\mu}_{D^{2+\delta}}g(s,\hat{\mu}_{D^2}){\rm e}^{(2+\delta)\hat{\mu}_D s}
\dint{s}\nonumber\\
&\le& {\rm e}^{6\hat{\mu}_D t}\left[1+6 \hat{\mu}_{D^{2+\delta}}tg(t,\hat{\mu}_{D^2})\right].
\end{eqnarray}

We return to epidemics on networks and introduce some more notation.  For $\epsilon \in (0,1)$, let $k_0(\epsilon)=\min\{k:p_0+p_1+\dots,p_k > \epsilon\}$ and, for $k=0,1,\dots$, let
\begin{equation*}
p_k(\epsilon)= \begin{cases}
	0 & \text{ if } k<k_0(\epsilon) \text{ ,}\\
	1-\frac{1}{1-\epsilon}\sum_{k=k_0(\epsilon)}^{\infty} p_k & \text{ if } k= k_0(\epsilon)  \text{ ,}\\
	\frac{p_k}{1-\epsilon} & \text{ if } k> k_0(\epsilon)  \text{ .}
\end{cases}
\end{equation*}
For $f:\mathbb{R} \to \mathbb{R}$, let $\mu_{f(D(\epsilon))}=\sum_{k=0}^\infty p_k(\epsilon) f(k)$.  As above, we assume without loss of generality that $\tau=1$.

Consider first the model with prescribed degrees.  Let $D_{(1)}^{(N)},D_{(2)}^{(N)},\dots,D_{(N)}^{(N)}$ be the order statistics of $D_1^{(N)},D_2^{(N)},\dots, D_N^{(N)}$.  For $\epsilon \in (0,1)$, let $p_k^{(N)}(\epsilon)=(N-[N\epsilon])^{-1}\sum_{i=N-[N\epsilon]}^N \delta_{k,D_{(i)}^{(N)}}$ $(k=0,1,\dots)$ and,\newline for $f:\mathbb{R} \to \mathbb{R}$, let
\begin{equation*}
\mu_{f(D^{(N)}(\epsilon))}=\sum_{k=0}^\infty p_k^{(N)}(\epsilon) f(k)=\frac{1}{N-[N\epsilon]}\sum_{k=[N\epsilon]+1}^N f\left(D_{(k)}^{(N)}\right)^\alpha.
\end{equation*}
Let $\tilde{p}^{(N)}_k(\epsilon)=kp_k^{(N)}(\epsilon)/\mu^{(N)}_{D(\epsilon)}$ $(k=1,2,\dots)$. For $f:\mathbb{R} \to \mathbb{R}$, let
\begin{equation*}
\mu_{f\left(\tilde{D}^{(N)}(\epsilon)\right)}=\sum_{k=1}^\infty \tilde{p}_k^{(N)}(\epsilon) f(k)=\frac{\mu_{D^{(N)}(\epsilon)f\left(D^{(N)}(\epsilon)\right)}}{\mu_{D^{(N)}(\epsilon)}}.
\end{equation*}

Fix the population size $N$ and $\epsilon \in (0,1)$.  In the birth process used to bound the right-hand side of~\eqref{YNUniintbound}, the types of individuals born are distributed according to
$\hat{p}_i=\tilde{p}^{(N)}_{i+1}(\epsilon)$ $(i=0,1,\dots,J)$, where $J=\max\left\{D_{(k)}^{(N)}:k=1,2,\dots,N\right\}$. Thus, for $\delta \in (0,1)$, it follows using~\eqref{YNUniintbound} and~\eqref{hatmud2delta} that
\begin{equation*}
\EX{Y_N(t)^{2+\delta}} \le h_1(N,\epsilon,t),
\end{equation*}
where
\begin{equation*}
h_1(N,\epsilon,t)=\left(1+\epsilon^{-(2+\delta)}\right)\left[1+6\mu_{(\tilde{D}^{(N)}(\epsilon)-1)^{2+\delta}}t g\left(t,\mu_{(\tilde{D}^{(N)}(\epsilon)-1)^2}\right)\exp\left(6 \mu_{\tilde{D}^{(N)}(\epsilon)-1}t\right)\right].
\end{equation*}

Suppose that there exists $\delta \in (0,1)$ such that $\mu_{D^{3+\delta}}<\infty$ and $\mu_{(D^{(N)})^{3+\delta}} \to \mu_{D^{3+\delta}}$ as $N \to \infty$.  It is easily verified that these conditions imply that, for each $\epsilon \in (0,1)$
and any $\alpha \in [0,3+\delta]$, $\mu_{D(\epsilon)^{\alpha}}<\infty$ and $\mu_{(D^{(N)}(\epsilon))^\alpha} \to \mu_{D(\epsilon)^{\alpha}}$ as $N \to \infty$.  Hence, for any $\alpha \in [0,2+\delta]$, $\mu_{(\tilde{D}^{(N)}(\epsilon)-1)^\alpha} \to \mu_{(\tilde{D}(\epsilon)-1)^\alpha}$ as $N \to \infty$, where $\mu_{(\tilde{D}(\epsilon)-1)^\alpha}< \infty$. It follows that $h_1(N,\epsilon,t) \to h_1(\epsilon,t)$ as $N \to \infty$, where
\begin{equation*}
h_1(\epsilon,t)=\left(1+\epsilon^{-(2+\delta)}\right)\left[1+6\mu_{(\tilde{D}(\epsilon)-1)^{2+\delta}}t g\left(t,\mu_{(\tilde{D}(\epsilon)-1)^2}\right)\exp \left(6 \mu_{\tilde{D}(\epsilon)-1}t\right)\right] <\infty.
\end{equation*}
Thus the sequence $\left(\EX{Y_N(t)^{2+\delta}}\right)$ is bounded and, for any $t \ge 0$ and any $\alpha \in [0,2]$,  $\EX{Y_N(t)^\alpha} \to \EX{Z(t)^\alpha}$ as $N \to \infty$.

Turn now to the model with i.i.d.~degrees. Recall that we construct a sequence of epidemics $(E_N)$ from a single sequence
$D_1,D_2,\dots$ of i.i.d.~copies of $D$.  For $N=1,2,\dots$, let $D^{(N)}_k=D_k$ ($k=1,2,\dots,N$). Using the formulae derived previously for the prescribed degree case, but noting that now the degrees are random, by conditioning on the degree sequence
$D_1,D_2,\dots$ we obtain that
\be
\label{EYN2bound}
\EX{Y_N(t)^{2+\delta}} \le h_2(N,\epsilon,t),
\ee
where
\begin{equation*}
h_2(N,\epsilon,t)=\EX{\left(1+\epsilon^{-(2+\delta)}\right)\left\{1+6\mu_{(\tilde{D}^{(N)}(\epsilon)-1)^{2+\delta}}t g\left(t,\mu_{(\tilde{D}^{(N)}(\epsilon)-1)^2}\right)\exp \left(6 \mu_{\tilde{D}^{(N)}(\epsilon)-1}t\right)\right\}}.
\end{equation*}

Fix  $\epsilon \in (0,1)$.  Recalling the definition~\eqref{gtD2} of the function $g$, to obtain an upper bound for 
$h(N,\epsilon,t)$, it is sufficient to obtain upper bounds for
\begin{equation*}
\EX{\mu_{(\tilde{D}^{(N)}(\epsilon)-1)^{2+\delta}}\exp \left(6 \mu_{\tilde{D}^{(N)}(\epsilon)-1}t\right)} \qquad\mbox{and}\qquad \EX{\mu_{(\tilde{D}^{(N)}(\epsilon)-1)^{2+\delta}}\mu_{(\tilde{D}^{(N)}(\epsilon)-1)^{2}}\exp \left(6 \mu_{\tilde{D}^{(N)}(\epsilon)-1}t\right)}.
\end{equation*}
Now $\mu_{(\tilde{D}^{(N)}(\epsilon)-1)^{2}} \le \mu_{(\tilde{D}^{(N)}(\epsilon)-1)^{2+\delta}}$ and, by Jensen's inequality, $\left(\mu_{(\tilde{D}^{(N)}(\epsilon)-1)^{2+\delta}}\right)^2 \le \mu_{(\tilde{D}^{(N)}(\epsilon)-1)^{2(2+\delta)}}$, so, since $\delta \in (0,1)$, it is sufficient to obtain an upper bound for $\EX{\mu_{(\tilde{D}^{(N)}(\epsilon)-1)^6}\exp \left(6 \mu_{\tilde{D}^{(N)}(\epsilon)-1}t\right)}$.  Further, using the Cauchy-Schwarz inequality,
\begin{equation*}
\EX{\mu_{(\tilde{D}^{(N)}(\epsilon)-1)^6}\exp \left(6 \mu_{\tilde{D}^{(N)}(\epsilon)-1}t\right)} \le \sqrt{
\EX{\mu_{(\tilde{D}^{(N)}(\epsilon)-1)^6}^2} \EX{\exp \left(12 \mu_{\tilde{D}^{(N)}(\epsilon)-1}t\right)}}.
\end{equation*}

Let $M_{D^2}(\theta)=\EX{\exp \left(\theta D^2\right)}$ ($\theta \in \mathbb{R}$) be the moment-generating function of $D^2$ and suppose that there exists $\theta_0>0$ such that $M_D(\theta_0)<\infty$.  Note that this implies that $\EX{D^\alpha}<\infty$ for all $\alpha \ge 0$.

Assume first that $p_0=0$, so $\mu_{D^{(N)}(\epsilon)}\ge 1$ almost surely. Now
\begin{eqnarray}
\label{mudtile6bd}
\mu_{(\tilde{D}^{(N)}(\epsilon)-1)^6}&=&\frac{\mu_{D^{(N)}(\epsilon)\left(D^{(N)}(\epsilon)-1)^6\right)}}{\mu_{D^{(N)}(\epsilon)}}\nonumber\\
&\le& \mu_{D^{(N)}(\epsilon)\left(D^{(N)}(\epsilon)-1)^6\right)}\qquad\mbox{almost surely}\nonumber\\
&=&\frac{1}{N-[N\epsilon]}\sum_{k=[N\epsilon]+1}^N D^{(N)}_{(k)}\left(D^{(N)}_{(k)}-1\right)^6\nonumber\\
&\le& \frac{1}{N(1-\epsilon)}\sum_{k=1}^N D_k^7,
\end{eqnarray}
since $D^{(N)}_k=D_k$ ($k=1,2,\dots,N$).  Thus, since $D_1,D_2,\dots,D_N$ are i.i.d.,
\begin{eqnarray*}
\EX{\mu_{(\tilde{D}^{(N)}(\epsilon)-1)^6}^2}&\le&\frac{1}{N^2(1-\epsilon)^2}\EX{\left(\sum_{k=1}^N D_k^7\right)^2}\\
&=&\frac{1}{N^2(1-\epsilon)^2}\left[N\mu_{D^{14}}+N(N-1)\mu_{D^7}^2\right]\\
&\le&\frac{1}{(1-\epsilon)^2}\left(\mu_{D^{14}}+\mu_{D^7}^2\right).
\end{eqnarray*}

A similar argument to~\eqref{mudtile6bd} yields $\mu_{(\tilde{D}^{(N)}(\epsilon)-1)}\le\frac{1}{N(1-\epsilon)}\sum_{k=1}^N D_k^2$, so
\begin{equation*}
\EX{\exp \left(12 \mu_{\tilde{D}^{(N)}(\epsilon)-1}t\right)} \le \EX{\exp \left(\frac{12t}{N(1-\epsilon)}\sum_{k=1}^N D_k^2 \right)}=
\left[M_{D^2}\left(\frac{12t}{N(1-\epsilon)}\right)\right]^N.
\end{equation*}
Fix $t \ge 0$.  Then $\EX{\exp \left(12 \mu_{\tilde{D}^{(N)}(\epsilon)-1}t\right)}<\infty$ for $N \ge N(\epsilon,t)$, where $N(\epsilon,t)=\frac{12t}{(1-\epsilon)\theta_0}$.  Now \newline $M_{D^2}(\theta)=1+\mu_{D^2}\theta+o(\theta)$ as $\theta \to 0$, so
\begin{eqnarray*}
\lim_{N \to \infty} \EX{\exp \left(12 \mu_{\tilde{D}^{(N)}(\epsilon)-1}t\right)}&\le&\lim_{N \to \infty}\left[1+\frac{12t}{N(1-\epsilon)}\mu_{D^2}+o\left(\frac{1}{N}\right)\right]^N\\
&=&\exp\left(\frac{12t\mu_{D^2}}{1-\epsilon}\right)<\infty.
\end{eqnarray*}

The above arguments show that there exists $h_2(\epsilon,t)<\infty$ such that $h_2(N,\epsilon,t)<h_2(\epsilon,t)$ for all $N \ge N(\epsilon,t)$. Now $Y_N(t) \le N$ for all $N$, so $\EX{Y_N(t)^{2+\delta}}\le N(\epsilon,t)^{2+\delta}$ for $N<N(\epsilon,t)$.
Thus, recalling~\eqref{EYN2bound}, the sequence $\left(\EX{Y_N(t)^{2+\delta}}\right)$ is bounded and, for any $\alpha \in [0,2]$,  $\EX{Y_N(t)^\alpha} \to \EX{Z(t)^\alpha}$ as $N \to \infty$.

Suppose now that $p_0>0$. Then $D \overset{st}{\le} D'$, where $D'$ has distribution given by $\mathbb{P}(D'=1)=p_0+p_1$ and $\mathbb{P}(D'=k)=p_k$ ($k=2,3,\dots$).  It follows that for fixed population size $N$, fixed $\epsilon \in (0,1)$ and any $t \ge 0$, $T_{N,\epsilon}(t) \overset{st}{\le}  T_{N,\epsilon}'(t)$, where $T_{N,\epsilon}'(t)$ is the total  progeny at time $t$ of the branching process defined analagously to  $\mathcal{B}_{N, \epsilon}$ but using the empirical distribution of $D_1',D_2',\dots,D_N'$, where $D_1',D_2',\dots,D_N'$ are
i.i.d.~copies of $D'$.  Further, $M_{D'^2}(\theta_0)<\infty$ if $M_{D^2}(\theta_0)<\infty$ and the above argument can be used to show that the sequence $\left(\EX{Y_N(t)^{2+\delta}}\right)$ is bounded.

The above argument is easily adapted to show that in the prescribed degree case $\lim_{N \to \infty}\EX{Y_N(t)} = \EX{Z(t)}$
under the weaker condition that there exists $\delta>0$ such that $\mu_{D^{2+\delta}}<\infty$ and $\mu_{(D^{(N)})^{2+\delta}} \to \mu_{D^{2+\delta}}$ as $N \to \infty$.  Moreover, although we have not worked through all of the details, it seems likely that the argument can also be adapted to prove that, for any $\alpha>1$, if there exists  $\delta>1$ such that $\mu_{D^{\alpha+\delta}}<\infty$ and $\mu_{(D^{(N)})^{\alpha+\delta}} \to \mu_{D^{\alpha+\delta}}$ as $N \to \infty$
then $\lim_{N \to \infty}\EX{Y_N(t)^\alpha} = \EX{Z(t)^\alpha}$ for all $t \ge 0$.  Again we have not worked through all of the details but in the i.i.d.~degree case it seems likely that the above condition on $M_{D^2}(\theta)$ quarantees that
$\lim_{N \to \infty}\EX{Y_N(t)^\alpha} = \EX{Z(t)^\alpha}$ for all $\alpha,t \ge 0$. Finally, in the i.i.d.~degree case it seems likely that weaker conditions will suffice when the limiting branching process $\mathcal{B}$ is subcritical, i.e. when $r<0$, as in that case the exponential functions appearing in $\EX{T_{N,\epsilon}(t)^{2+\delta}}$, prior to taking expectations, will all have negative arguments.

\section{Eigenvalues of $\mygmat{\Omega}$}
\label{app:eigenvalues}

Let $\mygmat{A}=[a_{l,k}]=\mygmat{\Omega} - \lambda \mygmat{I}$.  Observe that
$a_{0,0}=-(\lambda + \gamma)$, $a_{0,k}=0$ for $k=1,2,\ldots,\kmax$,
$a_{l,\kmax}=0$ for $l=1,2,\ldots,\kmax-1$ and $a_{\kmax,\kmax}=-(\kmax \tau +
\lambda + \gamma)$.  Thus, expanding the determinant $|\mygmat{A}|$ along the $0$-th row and
then the cofactor $A_{0,0}$ down the last column yields
\be
\label{detA}
|\mygmat{A}|=(\lambda + \gamma)(\kmax \tau + \lambda + \gamma)|\mygmat{B}|  \text{ ,}
\ee
where
\ba
\nonumber
\mygmat{B} &=
\begin{bmatrix}
	\tau(\dd_2 -1) -\lambda - \gamma & \tau \dd_3 & \cdots & \tau \dd_{\kmax} \\
	2\tau(\dd_2 + 1) & 2\tau(\dd_3 -1) -\lambda - \gamma & \cdots & 2 \tau \dd_{\kmax} \\
	\vdots & \vdots & \ddots & \vdots \\
	(\kmax-1)\tau \dd_2 & (\kmax-1) \tau \dd_3 & \cdots & (\kmax-1) \tau (\dd_{\kmax}-1)
	- \lambda - \gamma
\end{bmatrix}\text{ .}\\
\ea
More precisely, $\mygmat{B}$ is the $(\kmax-1) \times (\kmax-1)$ matrix with elements
\be
\nonumber
b_{l,k}= \tau l \left(\dd_{k+1} + \delta_{l,k+1}\right)
 -(\gamma + \tau l + \lambda) \delta_{l,k} \qquad (l,k=1,2,\ldots,\kmax-1) \text{ .}
\ee
Subtracting $l \times$ the first row of $\mygmat{B}$ from the $l$-th row of
$\mygmat{B}$, for $l=2,3,\ldots,\kmax-1$, now gives
$|\mygmat{B}|=|\mygmat{C}|$, where
\ba
\nonumber
\mygmat{C} &=
\begin{bmatrix}
	\tau(\dd_2 -1) -\lambda - \gamma & \tau \dd_3 & \cdots & \tau \dd_{\kmax} \\
	2(2\tau + \lambda + \gamma) & -2\tau -\lambda - \gamma & \cdots & 0 \\
	3(\tau + \lambda + \gamma) & 3\tau & \cdots & 0 \\
	\vdots & \vdots & \ddots & \vdots \\
	(\kmax-1) (\tau + \lambda + \gamma) & 0 & \cdots & -(\kmax-1) \tau - \lambda - \gamma
\end{bmatrix}
\ea
has elements given by
\ba
\nonumber
c_{1,k}&=\tau \dd_{k+1} -(\gamma + \tau + \lambda) \delta_{1,k} & (k=1,2,\ldots,\kmax-1)\text{ ,} \\
c_{l,k}&= \tau l \delta_{l,k+1}
-(\gamma + \tau l + \lambda) \delta_{l,k} +l(\gamma + \tau + \lambda)\delta_{1,k} \qquad & (l=2,3,\ldots,\kmax-1; k=1,2,\ldots,\kmax-1) \text{ .}
\ea
In particular, $c_{1,1}= \tau(\dd_2-1)-\gamma-\lambda$, $c_{1,k}=\tau
\dd_{k+1}$ for $k=2,3,\ldots,\kmax-1$, and $c_{l,k}=0$ for $2 \le l<k\le
\kmax-1$.  Thus, adding $k \times$ the $k$-th column of $\mygmat{C}$ to the
first column of $\mygmat{C}$, for $k=2,3,\ldots,\kmax-1$, yields
$|\mygmat{C}|=|\mygmat{D}|$, where
\ba
\nonumber
\mygmat{D} & =
\begin{bmatrix}
	\tau(({\textstyle \sum_{l=0}^{\kmax}} l \dd_{l+1}) -1)
	  -\lambda - \gamma & \tau \dd_3 & \cdots & \tau \dd_{\kmax} \\
	0 & -2\tau -\lambda - \gamma & \cdots & 0 \\
	0 & 3\tau & \cdots & 0 \\
	\vdots & \vdots & \ddots & \vdots \\
	0 & 0 & \cdots & -(\kmax-1) \tau - \lambda - \gamma
\end{bmatrix} \text{ .}
\ea
Note that $d_{1,1}=\tau\left[\left(\sum_{l=0}^{\kmax} l \dd_{l+1}\right)
-1\right]-\gamma-\lambda$ and, for $l \ge 2$, that $d_{l,1}=0, d_{l,l}=-(\gamma
+ \tau + \lambda)$ and $d_{l,k}=0$ for $k>l$.  Thus expanding $|\mygmat{D}|$
down the first column gives
\be
\nonumber
|\mygmat{D}|=\left\{\tau\left[\left(\sum_{l=0}^{\kmax} l \dd_{l+1}\right) -1\right]-\gamma-\lambda\right\}\left(-1\right)^{\kmax-2}\prod_{l=2}^{\kmax-1}
\left(l\tau+\gamma+\lambda\right)\text{ .}
\ee
Recalling~\eqref{detA} and $|\mygmat{B}|=|\mygmat{C}|=|\mygmat{D}|$, it follows that the eigenvalues of
$\mygmat{\Omega}$ are given by~\eqref{eigomega}.

\section{Derivation of variance} \label{app:variance} Recall the definition of
$\mymat{C}_k$ at~\eqref{C_k}. Note, using~\eqref{fac2off}, that for
$k\in\mathcal{K}$, \ba \label{onesCketc} \TR{\myvec{1}}\mymat{C}_k \myvec{1} &
= \tau k + \gamma \text{ ,}\\ \TR{\myvec{1}}\mymat{C}_k \myvec{n} & =
\TR{\myvec{n}}\mymat{C}_k \myvec{1} = (r + 2\gamma) k \text{ ,}\\
\TR{\myvec{n}}\mymat{C}_k \myvec{n} & = \tau \mu_{(\tilde{D}-2)^2} k + \gamma
k^2 \text{ .} \ea Now let $\myvec{c}_{11}$ be a column vector whose $k$-th
element is $\TR{\myvec{1}}\mymat{C}_k \myvec{1}$, and define $\myvec{c}_{1n}$
and $\myvec{c}_{nn}$ similarly, using $\TR{\myvec{1}}\mymat{C}_k \myvec{n}$ and
$\TR{\myvec{n}}\mymat{C}_k \myvec{n}$, respectively. Noting that $r +
2\gamma=\gamma+\tau\mu_{\tilde{D}-2}$, in this more compact
notation~\eqref{onesCketc} becomes \ba \label{vecc11etc} \myvec{c}_{11} & =
\tau \myvec{n} + \gamma \ones \text{ ,}\\ \myvec{c}_{1n} & = \left(\gamma+\tau
\mu_{\tilde{D}-2}\right) \myvec{n} \text{ ,}\\ \myvec{c}_{nn} & = \tau
\mu_{(\tilde{D}-2)^2} \myvec{n} + \gamma \myvec{n}_2 \text{ ,} \ea
and~\eqref{vk} yields \ba \myvec{v}(t) & = \int_{0}^{t} \left(
\mu_{\tilde{D}-2}^{-1} \left({\rm e}^{r(t-u)} - {\rm e}^{-\gamma (t-u)}
\right)\right)^2 {\rm e}^{\mygvec{\Omega} u}\myvec{c}_{nn}  \; \dint{u} \\
&\quad + 2 \int_{0}^{t} \mu_{\tilde{D}-2}^{-1}  \left({\rm e}^{r(t-u)} - {\rm
e}^{-\gamma (t-u)} \right){\rm e}^{-\gamma (t-u)} {\rm e}^{\mygvec{\Omega}
u}\myvec{c}_{1n}  \; \dint{u}\\ &\quad + \int_{0}^{t} {\rm e}^{-2 \gamma (t-u)}
	{\rm e}^{\mygvec{\Omega} u}\myvec{c}_{11} \; \dint{u} \text{ .} \label{vvec}
	\ea Now \be \mygmat{\Omega} \myvec{n}_2  = \tau
	\mu_{(\tilde{D}-1)^2+1}\myvec{n} - (2\tau + \gamma) \myvec{n}_2\text{ ,} \ee
	so, using  Proposition~\ref{prop:matid}, with $\mymat{M}=\mygmat{\Omega},
	\myvec{x}=\myvec{n}_2 \, \myvec{y}=\myvec{n}, a=-(2\tau+\gamma),
	b=\tau\mu_{(\tilde{D}-1)^2+1}$ and $c=r$, and recalling
	from~\eqref{growthrate} that $r+\gamma=\mu_{\tilde{D}-2}\tau$ so
	$a-c=-\tau(2+\mu_{\tilde{D}-2})=-\tau \mu_{\tilde{D}}$, we have \be {\rm
	e}^{\mygvec{\Omega} u}\myvec{n}_2  =
	\mu_{\tilde{D}}^{-1}\mu_{(\tilde{D}-1)^2+1} \left( {\rm e}^{ru} - {\rm
	e}^{-(2\tau+\gamma)u} \right) \myvec{n} + {\rm e}^{-(2\tau+\gamma)u}
	\myvec{n}_2 \text{ .} \ee Hence, using also~\eqref{n1}, \ba \label{eOmegac}
	{\rm e}^{\mygvec{\Omega} u}\myvec{c}_{11} & =
	\mu_{\tilde{D}-2}^{-1}\left(\gamma+\tau \mu_{\tilde{D}-2}\right){\rm
	e}^{ru}\myvec{n} +\gamma{\rm e}^{-\gamma
	u}\left(\ones-\mu_{\tilde{D}-2}^{-1}\myvec{n}\right)\text{ ,}\\ {\rm
	e}^{\mygvec{\Omega} u}\myvec{c}_{1n} & = \left(\gamma+\tau
	\mu_{\tilde{D}-2}\right){\rm e}^{ru}\myvec{n}\text{ ,}\\ {\rm
	e}^{\mygvec{\Omega} u}\myvec{c}_{nn} & = \left(\tau
	\mu_{(\tilde{D}-2)^2}+\gamma \mu_{\tilde{D}}^{-1}
\mu_{(\tilde{D}-1)^2+1}\right){\rm e}^{ru}\myvec{n}+\gamma{\rm
e}^{-(2\tau+\gamma)u}\left(\myvec{n}_2-\mu_{\tilde{D}}^{-1}
\mu_{(\tilde{D}-1)^2+1}\myvec{n}\right) \text{ .} \ea

\noindent{}Let $I_1(t)=\int_0^t {\rm e}^{-2 \gamma (t-u)} {\rm e}^{ru} \; \dint{u}, I_2(t)=\int_0^t {\rm e}^{-2 \gamma (t-u)} {\rm e}^{-\gamma u} \; \dint{u}, I_3(t)=\int_0^t\left({\rm e}^{r(t-u)}-{\rm e}^{-\gamma(t-u)}\right){\rm e}^{-\gamma(t-u)} {\rm e}^{ru} \; \dint{u}$, $I_4(t)=\int_0^t \left({\rm e}^{r(t-u)}-{\rm e}^{-\gamma(t-u)}\right)^2 {\rm e}^{ru} \; \dint{u}$ and
$I_5(t)=\int_0^t \left({\rm e}^{r(t-u)}-{\rm e}^{-\gamma(t-u)}\right)^2 {\rm e}^{-(2\tau+\gamma)u}\; \dint{u}$.  It is easily verified that these integrals are given by~\eqref{integrals}.
Substituting~\eqref{eOmegac} into~\eqref{vvec} and using~\eqref{integrals} yields~\eqref{vvec1}.

\section{Derivation of covariance function}
\label{app:covariance}
For $t \ge 0$, let $\myvec{u}(t)$ be the column vector whose $k$-th element is
$\TR{\myvec{1}}\mymat{V}^{(k)}(t)\myvec{n}$. Arguing as in the derivation of~\eqref{vvec} yields
\ba
\label{uvec}
\myvec{u}(t) &= \int_{0}^{t}
\mu_{\tilde{D}-2}^{-1} \left({\rm e}^{r(t-u)} - {\rm e}^{-\gamma (t-u)}
\right){\rm e}^{r(t-u)} {\rm e}^{\mygvec{\Omega} u} \myvec{c}_{nn}  \; \dint{u}\\
&\quad + \int_{0}^{t} {\rm e}^{-\gamma (t-u)} {\rm e}^{r(t-u)} {\rm e}^{\mygvec{\Omega} u}\myvec{c}_{1n}  \; \dint{u}\text{ .}
\ea
Using~\eqref{eOmegac} now gives
\be
\label{uvec1}
\myvec{u}(t)= \beta_1(t) \myvec{n} +  \beta_2(t) \myvec{n}_2\text{ ,}
\ee
where
\ba
\label{betas}
\beta_1(t) &= \gamma \mu_{\tilde{D}-2}^{-1}\left[\mu_{\tilde{D}}^{-1}\mu_{(\tilde{D}-1)^2+1}(I_7(t)-I_8(t))+\mu_{\tilde{D}-2}I_6(t)\right]\\
&\quad+\tau\mu_{\tilde{D}-2}^{-1}\left[\mu_{(\tilde{D}-2)^2}I_7(t)+ \mu_{\tilde{D}-2}^2 I_6(t)\right]\text{ ,}\\
\beta_2(t) &= \gamma \mu_{\tilde{D}-2}^{-1} I_8(t)\text{ ,}
\ea
with $I_6(t)=\int_0^t {\rm e}^{- \gamma (t-u)} {\rm e}^{r (t-u)}{\rm e}^{ru} \; \dint{u}, I_7(t)=\int_0^t \left({\rm e}^{r(t-u)} - {\rm e}^{-\gamma (t-u)}\right){\rm e}^{r(t-u)}{\rm e}^{ru} \; \dint{u}$ and \newline $I_8(t)=\int_0^t \left({\rm e}^{r(t-u)} - {\rm e}^{-\gamma (t-u)}\right){\rm e}^{r(t-u)} {\rm e}^{-(2\tau+\gamma)u}\; \dint{u}$.  It is easily verified that these integrals are given by~\eqref{integrals1}.  The expression~\eqref{sigmavec} for the covariance follows using~\eqref{sigmast1},~\eqref{uvec},~\eqref{uvec1} and~\eqref{vk}.

\section{Late behaviour of subcritical survival probabilities}
\label{app:survival}

To bound the late probabilities of survival in the subcritical case, first note
that due to the inability of type-$0$ individuals to transmit we have
\be
q_0(t) = {\rm e}^{-\gamma t} \text{ .} \label{q0sol}
\ee
Recalling that $q_k(t)=1-\pi_k(t)$, it follows from~\eqref{qqext}, with
$\kmax=\infty$, that for $k>0$,
\be
\ddt{q_k} = -(\gamma + \tau k)q_k + \tau k q_{k-1} + \tau k (1-q_{k-1})\sum_{l=0}^{\infty}
\tilde{p}_{l+1}q_l \text{ .} \label{qkext}
\ee
This leads to the equation for $k=1$ in
the form
\be
\ddt{q_1}  = r q_1 + h(t) \text{ ,}
\label{q1ext}
\ee
where
\be
h(t)  = \underbrace{\tau {\rm e}^{-\gamma t} \left(1 + \tilde{p}_{1} - \sum_{k=0}^{\infty}
\tilde{p}_{k+1}q_k(t)\right)}_{h_1(t)} + \underbrace{\tau \sum_{k=1}^{\infty} \tilde{p}_{k+1}
\left(q_k(t) - k q_1(t)\right)}_{h_2(t)} \text{ .}
\label{heq}
\ee
We assume first that $r>-\gamma$. Integrating~\eqref{q1ext} gives
\be
{\rm e}^{-r t}q_1(t) = 1 - \int_{0}^{t} {\rm e}^{-rt} h(u)\dint{u} \text{ .}
\nonumber
\ee
The limit $\lim_{t\rightarrow\infty}{\rm e}^{-rt}q_1(t)$ therefore exists if
$r$ is within the region of convergence of the Laplace transform of $h$.
Considering $h_1(t)$ in~\eqref{heq}, since $q_k(t)\in[0,1]$, for all $k$, we
have that $\tilde{p}_{1} \le h_1(t) \leq \tau (1+\tilde{p}_{1}) {\rm e}^{-\gamma t}$, so the
Laplace transform of $h_1$ converges by the assumption that $r>-\gamma$.

Now considering $h_2$, we follow Windridge~\cite{Windridge:2014} and consider
an initial individual with effective degree $k$, and stubs labelled by integer
$i=1,2,\ldots,k$. Let $T_i$ be the time that the individual and its progeny
through stub $i$ exist, and let $T$ be the lifetime of the branching process.  Then, using the Bonferroni
inequalities as in~\cite{Windridge:2014}, for $k \ge 1$,
\be
\{T>t\} = \bigcup_{i=1}^{k} \{T_i>t\} \text{ ,}\qquad
q_k(t) \leq k q_1(t) \qquad \mbox{and} \qquad
q_k(t) \geq k q_1(t) - k^2 \mathbb{P}(T_1>t, T_2>t) \text{ .}
\label{qkbounds}
\ee
Therefore, we have that
\be
0 \leq h_2(t) \leq \tau \mu_{(\tilde{D}-1)^2} \mathbb{P}(T_1>t, T_2>t) \text{ .}
\nonumber
\ee
Now, from~\eqref{qkupper} we have that for some constant $\kappa$,
\be
q_k(t) \leq k \kappa {\rm e}^{rt} \text{ .}
\ee
Writing $R$ for the lifetime of the initial infective individual we have that
\be
\mathbb{P}(T_1>t, T_2>t) = \mathbb{P}(R>t) + \mathbb{P}(T_1>t, T_2>t, R\leq t)
\text{ .}
\nonumber
\ee
We then recall that $\mathbb{P}(R>t) = q_0(t) = {\rm e}^{-\gamma t}$, after
which the argument follows closely that of~\cite{Windridge:2014} (as in the
derivation of~\eqref{pRT1T2bound} below) and we find
that
\be
\lim_{t \to \infty} \int_0^\infty {\rm e}^{-rt}\mathbb{P}(T_1>t, T_2>t) \dint{t} < \infty \text{ .}
\label{P12int}
\ee
Thus the Laplace transform of $h_2$ converges at $r$, whence
$\lim_{t\rightarrow\infty}{\rm e}^{-rt}q_1(t)$ is equal to a finite constant,
$c$ say.  Note that $c$ is strictly positive since $c \ge \lim_{t \to
\infty}{\rm e}^{-rt}\hat{q}_1(t)=\hat{c}>0$.  Further,~\eqref{P12int} implies
that $\lim_{t \to \infty}{\rm e}^{-rt}\mathbb{P}(T_1>t, T_2>t)=0$, and it
follows from the two inequalities in~\eqref{qkbounds} that, $\lim_{t \to
\infty} {\rm e}^{-rt}q_k(t)=kc$, for $k \ge 0$, proving~\eqref{ourq}.\\

\noindent{}We consider now the case when $r<-\gamma$.  For $t\ge 0$ and $k=0,1,\ldots$,
write
\be
\label{qktilde}
q_k(t)={\rm e}^{-\gamma t}+\tilde{q}_k(t) \qquad \mbox{and} \qquad \tilde{u}_k(t)={\rm e}^{\gamma t}\tilde{q}_k(t)\text{ ,}
\ee
so $\tilde{q}_k(t)$ is the probability that the branching process has survived to time $t$ but the initial individual has not,
and $\tilde{u}_0(t)=0$ for all $t \ge 0$.  It follows using~\eqref{qkext} that
\be
\ddt{\tilde{u}_1} = \tau\left[-\tilde{u}_1+(1-{\rm e}^{-\gamma t})\left(1+\sum_{l=1}^{\infty}
\tilde{p}_{l+1}\tilde{u}_l\right)\right] \text{ .}
\label{du1tilde}
\ee
The Bonferroni inequalities yield that, for $k \ge 1$,
\be
k \tilde{q}_1(t)- k^2 \mathbb{P}(R<t, T_1>t, T_2>t) \le \tilde{q}_k(t) \leq k \tilde{q}_1(t) \text{ ,}
\label{bonf}
\ee
so
\be
0 \le \tilde{u}_k(t) \le k \tilde{u}_1(t)\text{ ,}
\label{uktildebound}
\ee
and~\eqref{du1tilde} implies that
\be
\ddt{\tilde{u}_1} \le \tau\left(1+\mu_{\tilde{D}-2}\tilde{u}_1\right)\text{ ,}
\nonumber
\ee
whence, for all $t \ge 0$, recalling that $\mu_{\tilde{D}-2}<0$,
\be
0 \le  \tilde{u}_1(t) \le -\frac{1}{\mu_{\tilde{D}-2}}\left(1-{\rm e}^{\tau \mu_{\tilde{D}-2} t}\right) \le  -\frac{1}{\mu_{\tilde{D}-2}}\text{ .}
\label{u1tildebound}
\ee
Conditioning on the lifetime of the initial individual in the branching process,
\ba
\mathbb{P}(R \le t, T_1>t, T_2>t)&=\int_{u=0}^t \gamma {\rm e}^{-\gamma u} \left[\int_{v=0}^u \tau {\rm e}^{-\tau v}q_1(t-v)\dint{v}\right]^2 \dint{u}\\
&\le \int_{u=0}^t \gamma {\rm e}^{-\gamma u} \left[\int_{v=0}^u \tau {\rm e}^{-\tau v}\left(-\frac{\mu_{\tilde{D}-1}}{\mu_{\tilde{D}-2}}\right) {\rm e}^{-\gamma (t-v)}\dint{v}\right]^2 \dint{u}\text{ ,}
\nonumber
\ea
since~\eqref{qktilde} and~\eqref{u1tildebound} imply that $q_1(t) \le -\frac{\mu_{\tilde{D}-1}}{\mu_{\tilde{D}-2}} {\rm e}^{-\gamma t}$.
Elementary integration then shows that there exists $c'=c'(\mu, \tau,\mu_{\tilde{D}})<\infty$ such that, for all $t \ge 0$,
\be
{\rm e}^{\gamma t}\mathbb{P}(R \le t, T_1>t, T_2>t) \le
\begin{cases}
c'{\rm e}^{-\min\{\gamma,2\tau\}t}  & \text{ if } \gamma \ne 2 \tau \text{ ,}\\
c' t {\rm e}^{-\gamma t} & \text{ if } \gamma = 2 \tau \text{ .}
\end{cases}
\label{pRT1T2bound}
\ee
The differential equation~\eqref{du1tilde} may be written in the form
\be
\ddt{\tilde{u}_1} = \tau(1+\tilde{u}_1)-\tau \underbrace{{\rm e}^{-\gamma t}\left(1+\sum_{l=1}^{\infty}
\tilde{p}_{l+1}\tilde{u}_l\right)}_{h_3(t)}-\tau \underbrace{\sum_{l=2}^\infty \tilde{p}_{l+1}(l \tilde{u}_1-\tilde{u}_l)}_{h_4(t)}
\text{ ,}
\nonumber
\ee
whence
\be
\tilde{u}_1(t)=-\frac{1}{\mu_{\tilde{D}-2}}\left(1-{\rm e}^{\tau \mu_{\tilde{D}-2} t}\right)-\tau{\rm e}^{\tau \mu_{\tilde{D}-2} t}
\int_0^t {\rm e}^{-\tau \mu_{\tilde{D}-2} u}\left(h_3(u)+h_4(u)\right)\dint{u}\text{ .}
\label{tilde1ut}
\ee
Now~\eqref{uktildebound} and~\eqref{u1tildebound} imply that $0 \le h_3(t) \le
-\frac{1}{\mu_{\tilde{D}-2}}{\rm e}^{- \gamma t}$, whence
\be
\lim_{t \to \infty} {\rm e}^{\tau \mu_{\tilde{D}-2} t}\int_0^t {\rm e}^{-\tau \mu_{\tilde{D}-2} u}h_3(u)\dint{u}=0\text{ .}
\label{inth3}
\ee
Further, it follows using~\eqref{bonf} and~\eqref{pRT1T2bound} that $0 \le
h_4(t) \le \mu_{(\tilde{D}-1)^2} c'{\rm e}^{-\min\{\gamma,2\tau\}t}$, where
${\rm e}^{-\min\{\gamma,2\tau\}t}$ is replaced by $t {\rm e}^{-\gamma t}$ if
$\gamma = 2 \tau$, whence~\eqref{inth3} also holds when $h_3(u)$ is replaced by
$h_4(u)$.  Letting $t \to \infty$ in~\eqref{tilde1ut} yields~\eqref{ourq1}.

\section*{Acknowledgements}

We gratefully acknowledge support from the Isaac Newton Institute for
Mathematical Sciences, Cambridge, where we held Visiting Fellowships under the
Infectious Disease Dynamics programme and its follow-up meeting, during which
this work was initiated.  TH is supported by the Engineering and Physical
Sciences Research Council (Grant number EP/N033701/1). We would like to thank
Josh Ross for helpful comments on this manuscript.  We would also like to thank
the referees and associate editor for their constructive comments which have
improved the presentation of the paper.

\clearpage

\begin{table}
\begin{tabular}{p{.2\textwidth}p{.475\textwidth}p{.225\textwidth}}
\hline\noalign{\smallskip}
\textbf{Primary Notation} & \textbf{Meaning} & \textbf{Equivalent notation} \\
\noalign{\smallskip}\hline\noalign{\smallskip}
\multicolumn{3}{l}{\textit{Network properties:}}\\
\noalign{\smallskip}\hline\noalign{\smallskip}
$N$ & Size of the population & \\
$D$ & A random variable for an individual's degree & $D_1, D_2, \ldots$ \\
$p_k$ & Probability mass function for $D$ evaluated at $k$ & \\
$\kmax$ & The maximum degree & \\
$\mathcal{K}$ & The set of possible degrees & $\{0,1,\ldots \kmax\}$ \\
$\tilde{D}$ & A random variable for an individual's size-biased degree & $\tilde{D}_1, \tilde{D}_2, \ldots$ \\
$\tilde{p}_k$ & Probability mass function for $\tilde{D}$ evaluated at $k$ & $\mu_D^{-1}kp_k$ \\
\noalign{\smallskip}\hline\noalign{\smallskip}
\multicolumn{3}{l}{\textit{Vectors and matrices:}}\\
\noalign{\smallskip}\hline\noalign{\smallskip}
$\myvec{v}$ & A column vector whose $k$-th entry is $v_k$ & $(v_k)$ \\
$\TR{\myvec{v}}$ & A row vector (transpose of a column vector)& \\
$\mymat{M}$ & A matrix with $(k,l)$-th entry $M_{kl}$ or $m_{kl}$ & $[M_{kl}]$, $[m_{kl}]$ \\
$|\mymat{M}|$ & Determinant of matrix $\mymat{M}$ & \\
$\myvec{1}$ & A column vector whose entries are all equal to 1 & \\
$\myvec{n}$ & A column vector whose $i$th entry is $i$ & \\
$\myvec{n}_2$ & A column vector whose $i$th entry is $i^2$ & \\
$\mymat{I}$ & The identity matrix & $[\delta_{k,l}]$ \\
\noalign{\smallskip}\hline\noalign{\smallskip}
\multicolumn{3}{l}{\textit{Probability:}}\\
\noalign{\smallskip}\hline\noalign{\smallskip}
$\mathbb{P}(e)$ & Probability of event $e$ & \\
$\mu_{f(X)}$ & Expected value of a function $f$ of a random variable $X$ & $\EX{f(X)}$ \\
$M_X(\theta) $ & Moment generating function for random variable $X$ & $\EX{\mathrm{exp}(\theta X)}$ \\
$\mathrm{var}(X) $ & Variance of random variable $X$ & $\EX{X^2} - \EX{X}^2$ \\
$\mathrm{cov}(X,Y) $ & Covariance of random variables $X$ and $Y$ & $\EX{X Y} - \EX{X}\EX{Y}$ \\
\noalign{\smallskip}\hline\noalign{\smallskip}
\multicolumn{3}{l}{\textit{Epidemic and branching process dynamics:}}\\
\noalign{\smallskip}\hline\noalign{\smallskip}
$\tau$ & Rate of transmission across a network link & \\
$\gamma$ & Rate of recovery from infection & \\
$\omega_k$ & Death rate for individual of type $k$ & \\
$t$ & Real time & $s$ \\
$\mathcal{B}$ & The limiting branching process & \\
$E_N$ & The epidemic process in a population of size $N$ & \\
$K$ & A large value of infectious population size & \\
$Z^{(k)}_i(t)$ & Random number of individuals of type $i$ in the branching process at time $t$
given initial type $k$ & \\
$\pi_k(t)$ & Probability that the branching process is extinct at time $t$ given initial type $k$ &
$1-q_k(t)$%, $\mathbb{P}(\myvec{Z}^{(k)}=\mathbf{0})$
\\
\noalign{\smallskip}\hline
\end{tabular}
\caption{Here we define notation that is used in multiple sections of the paper.} \label{tab:notation}  
\end{table}

\clearpage

\begin{figure}
	\centering
	\subfloat{
	{\resizebox{.45\textwidth}{!}{\includegraphics{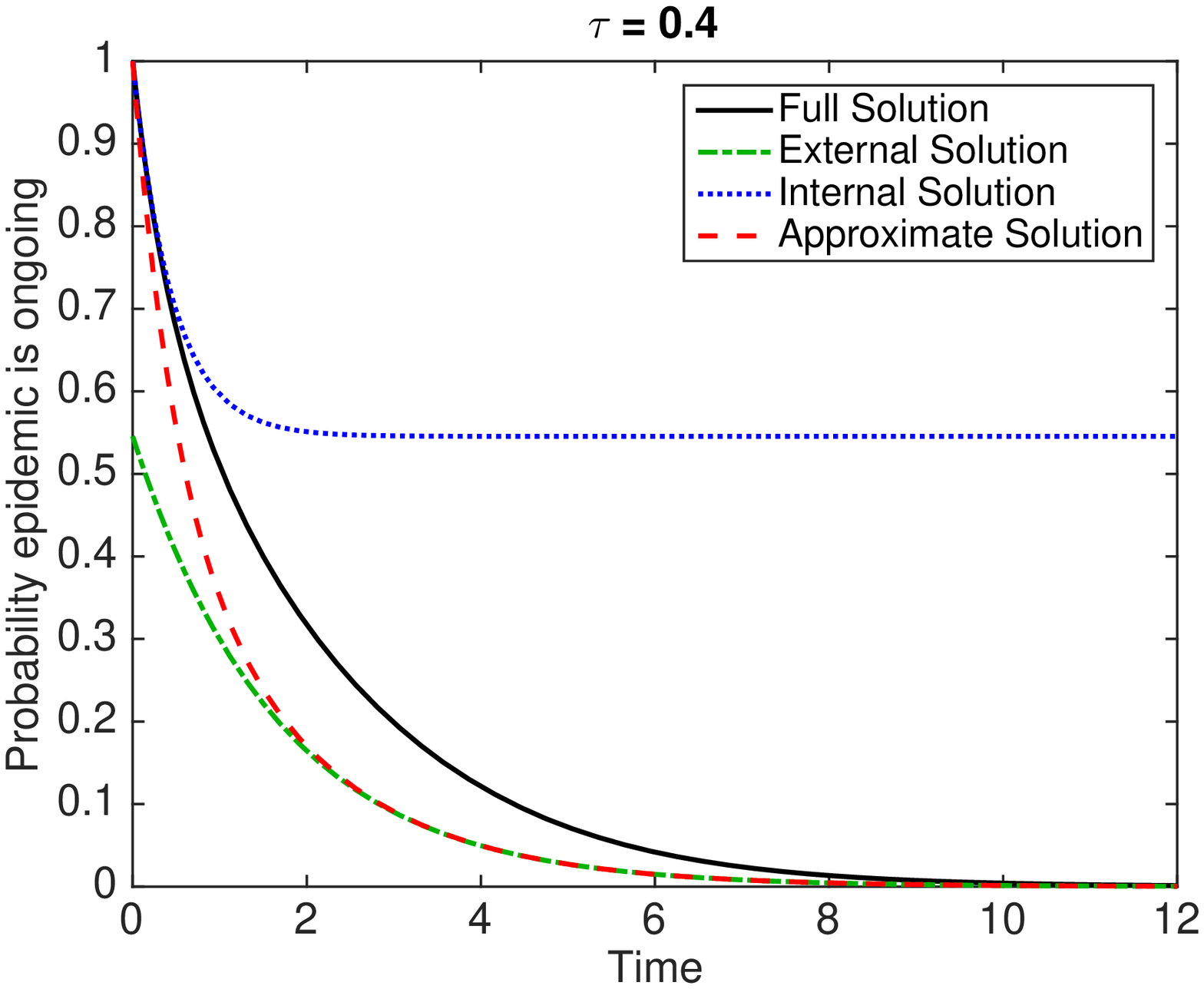}}}
	{\resizebox{.45\textwidth}{!}{\includegraphics{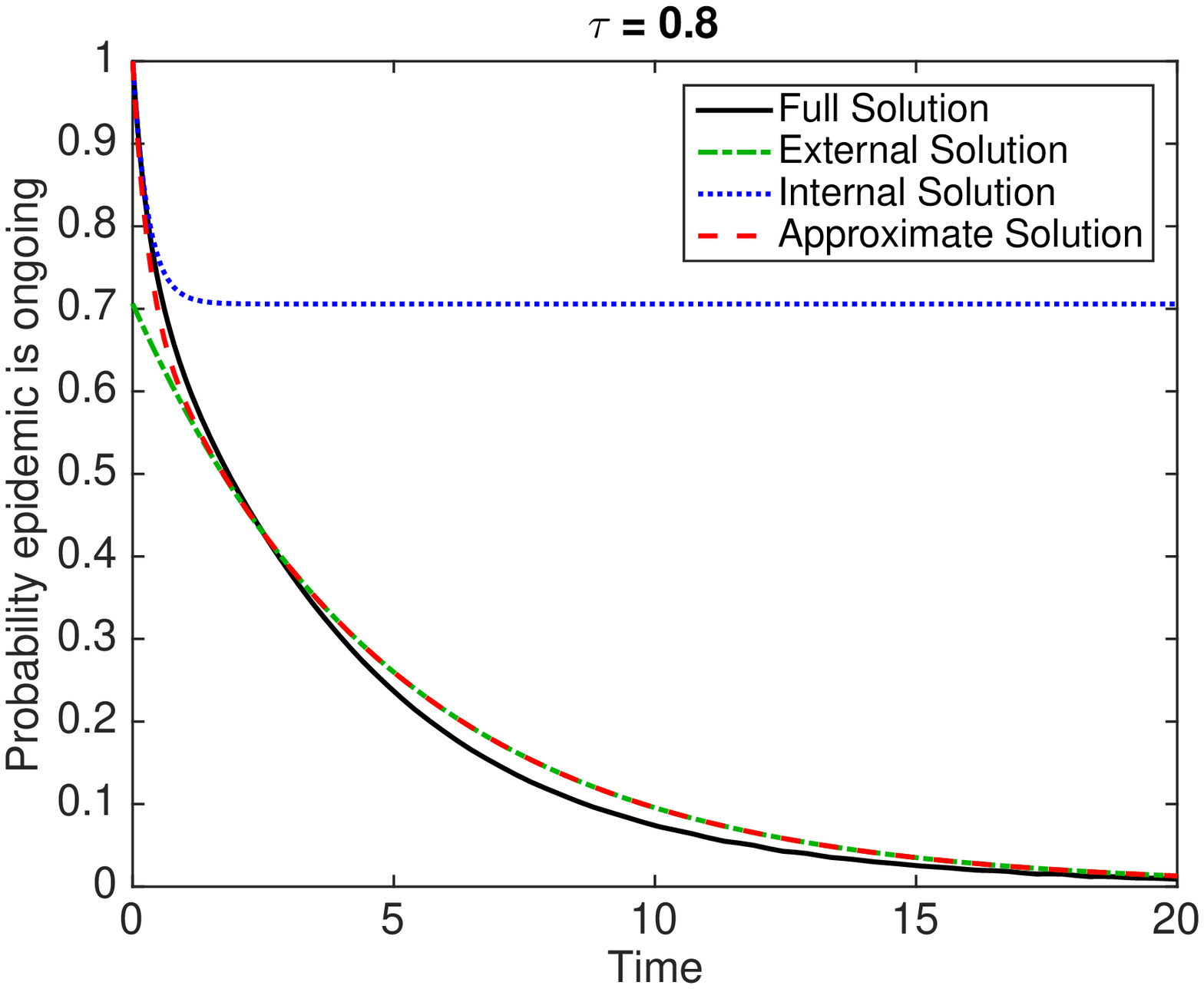}}}
}
\caption{Extinction probability for a subcritical epidemic compared to
	approximations. Results are for a 3-regular graph with $\gamma = 1$ and
	values of $\tau$ indicated in the figure titles. The internal and external
	solutions each fail severely at certain points, but the approximate solution
 crudely captures the overall behaviour.}
	\label{fig:ma}
\end{figure}

\newlength{\fs}
\setlength{\fs}{.3\textwidth}

\begin{figure}
	\centering
	\subfloat[Degree distribution histograms.]{%
	{\resizebox{\fs}{!}{\includegraphics{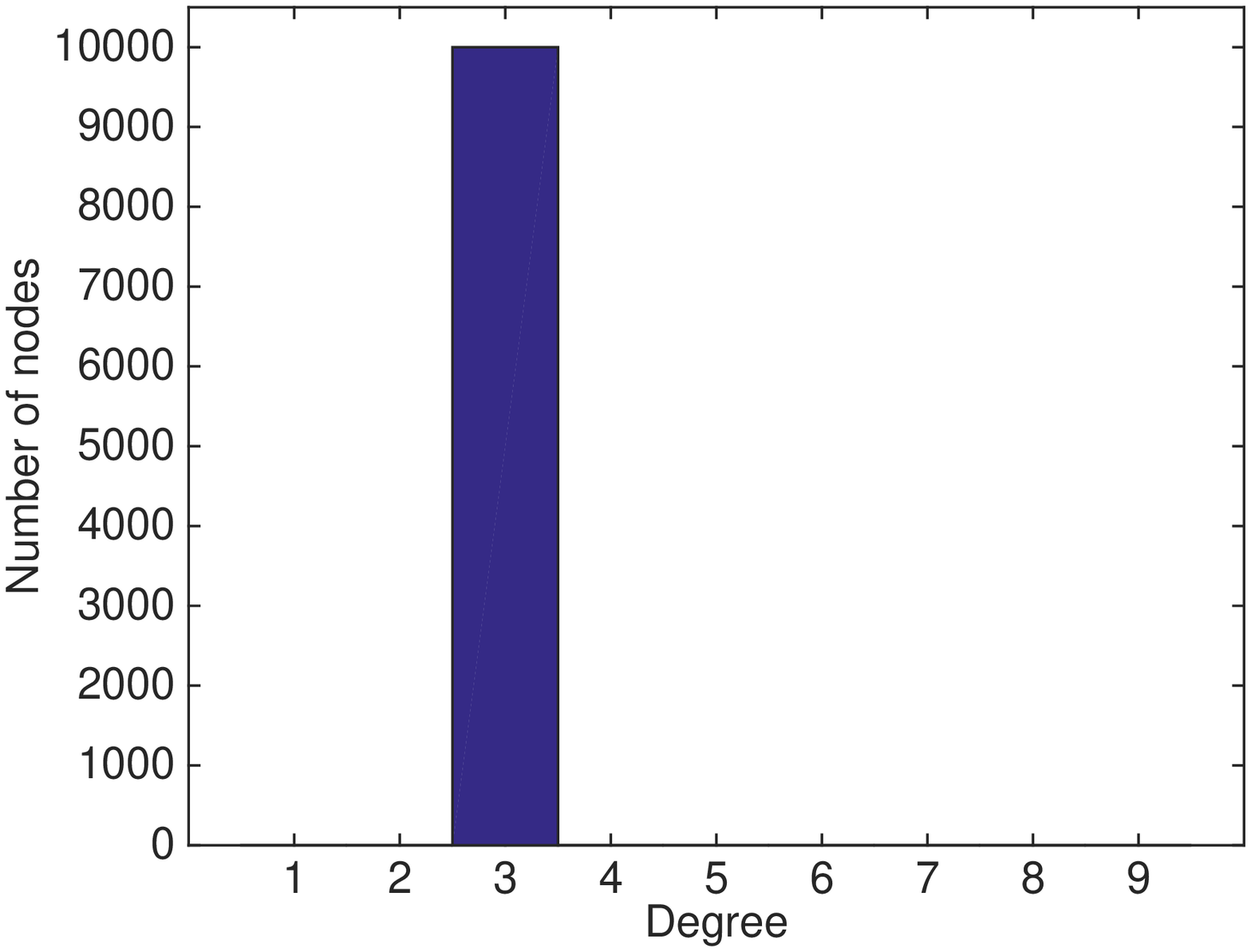} }} 
	{\resizebox{\fs}{!}{\includegraphics{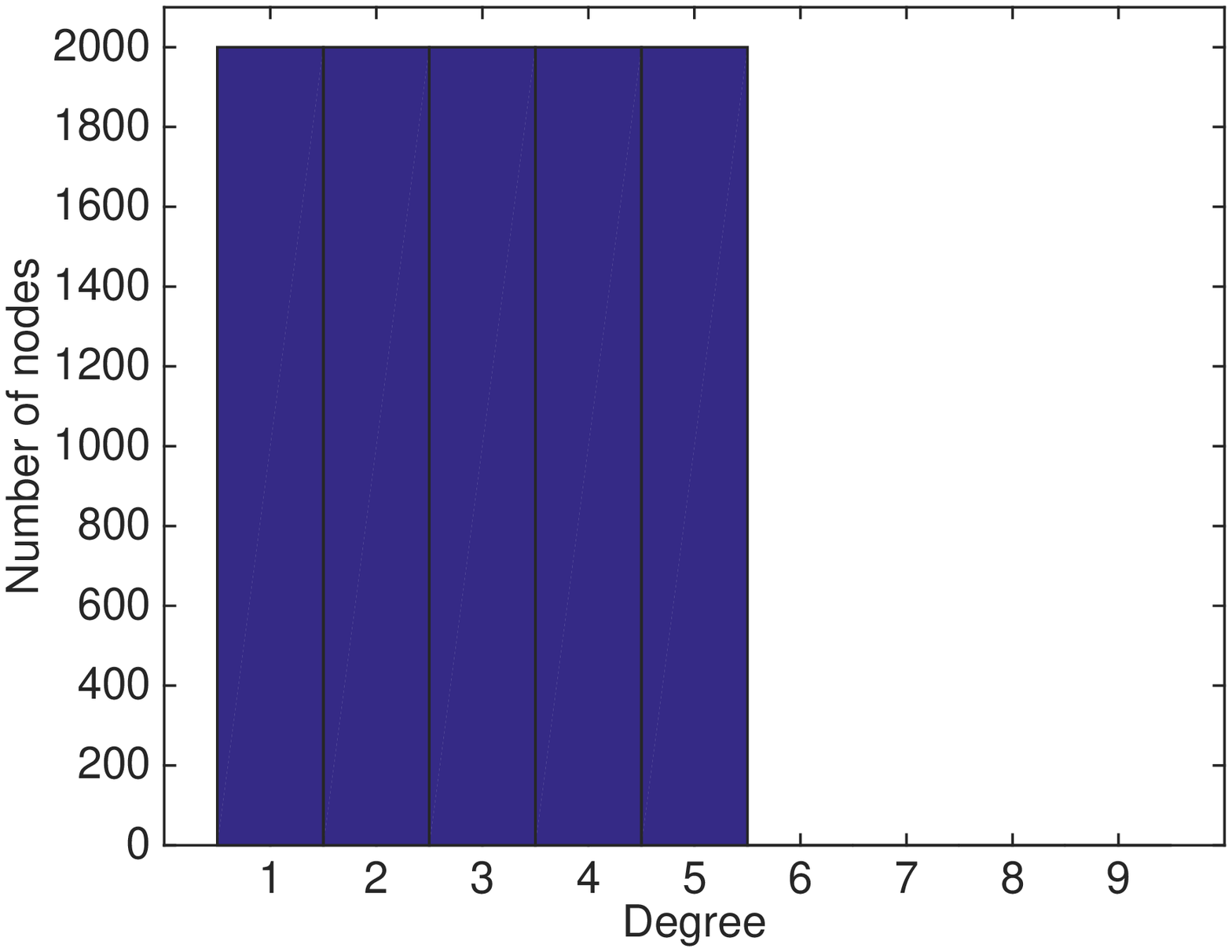} }} 
	{\resizebox{\fs}{!}{\includegraphics{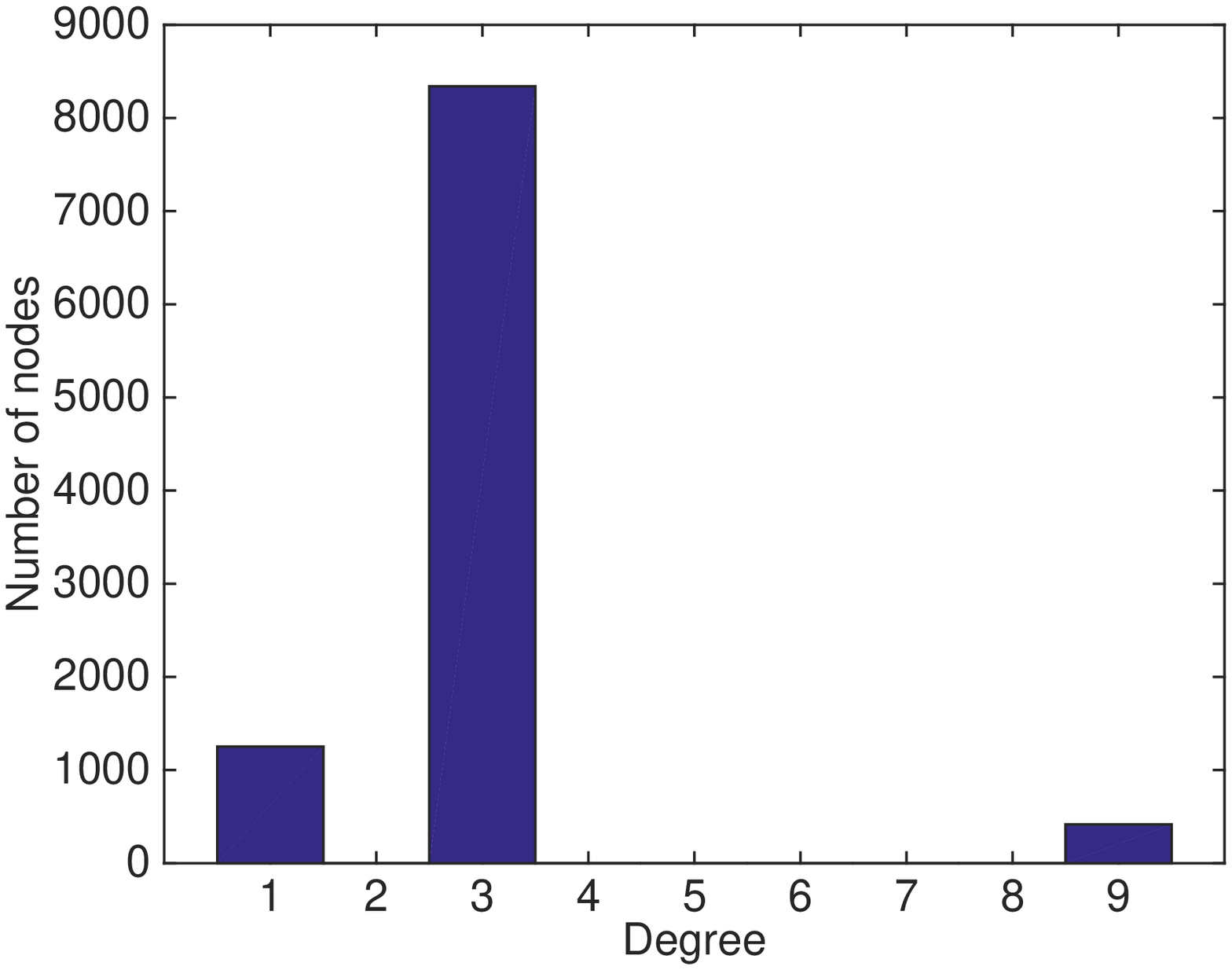} }}} \\
	\subfloat[100 sample trajectories.]{%
		{\resizebox{\fs}{!}{\includegraphics{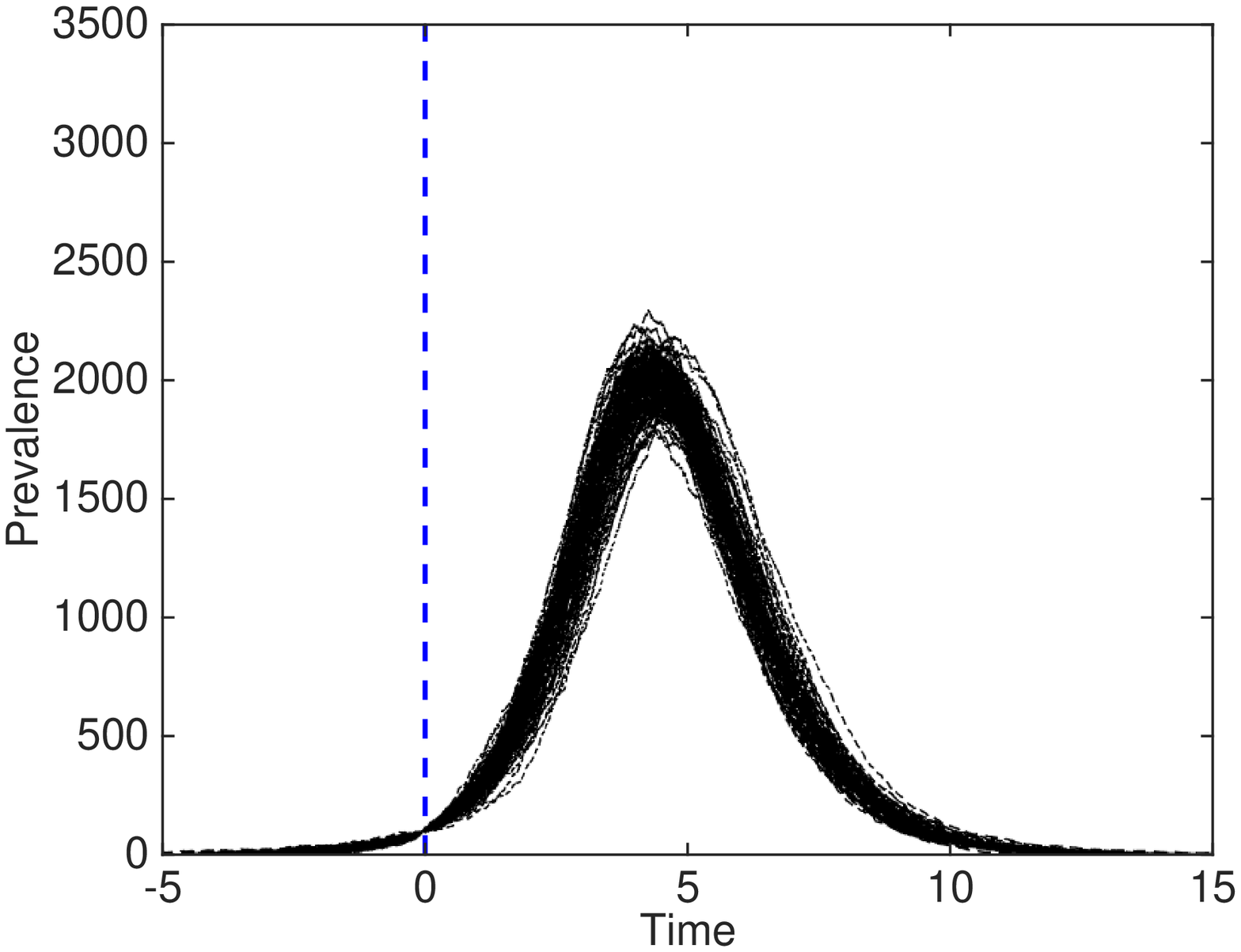} }} 
	{\resizebox{\fs}{!}{\includegraphics{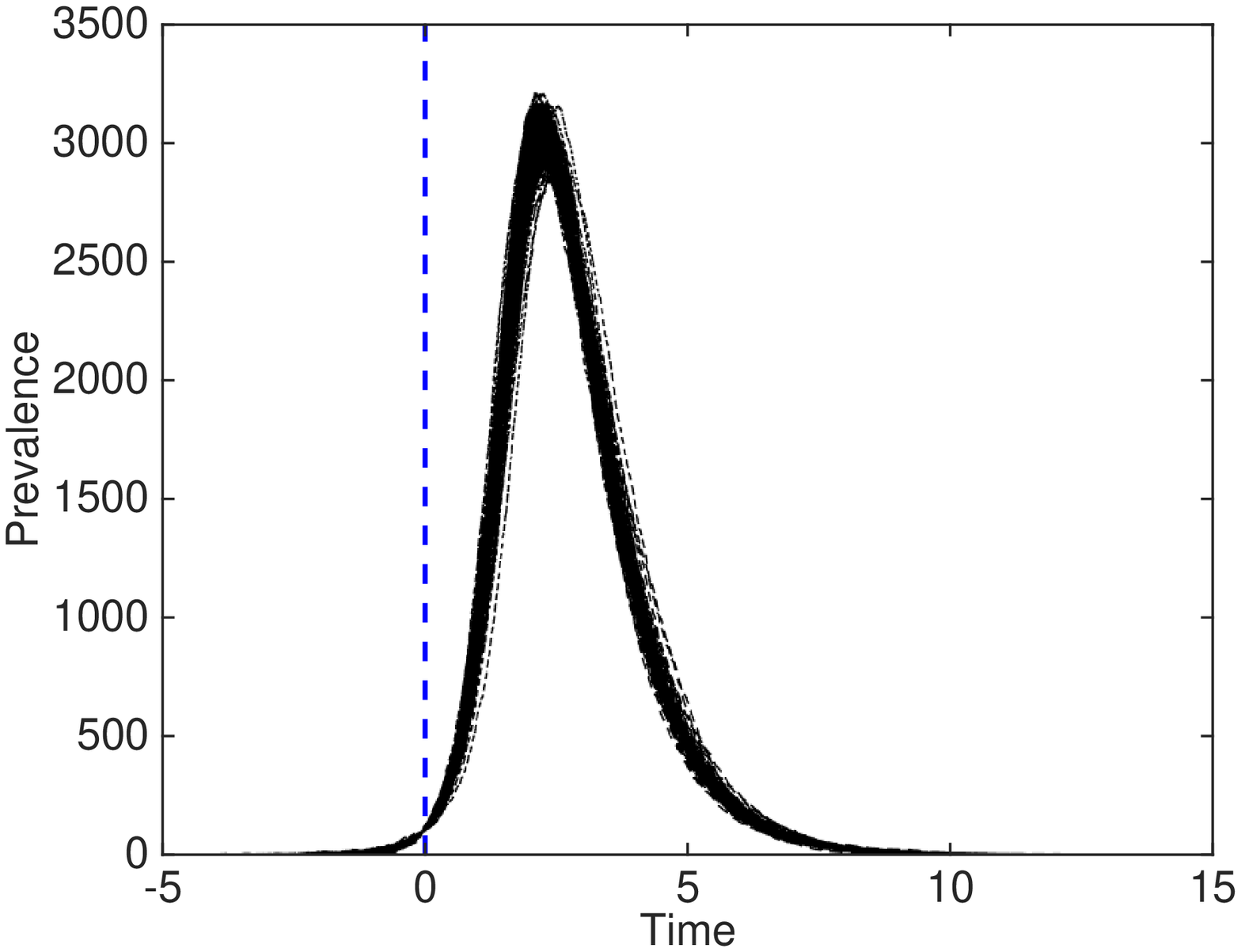} }} 
	{\resizebox{\fs}{!}{\includegraphics{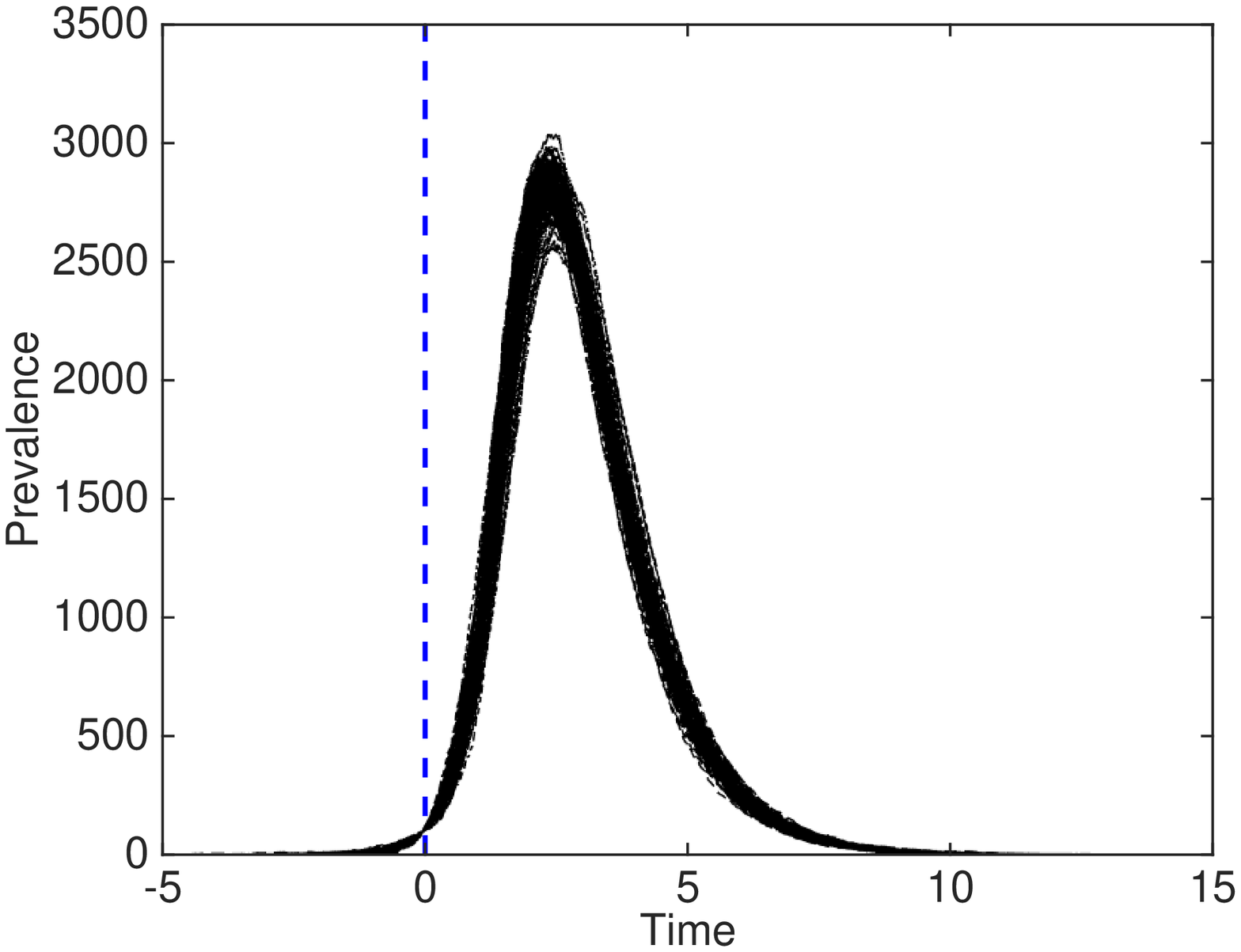} }}} \\
	\subfloat[Mean prevalence. Black solid: simulations; Red dashed: branching process.]{%
		{\resizebox{\fs}{!}{\includegraphics{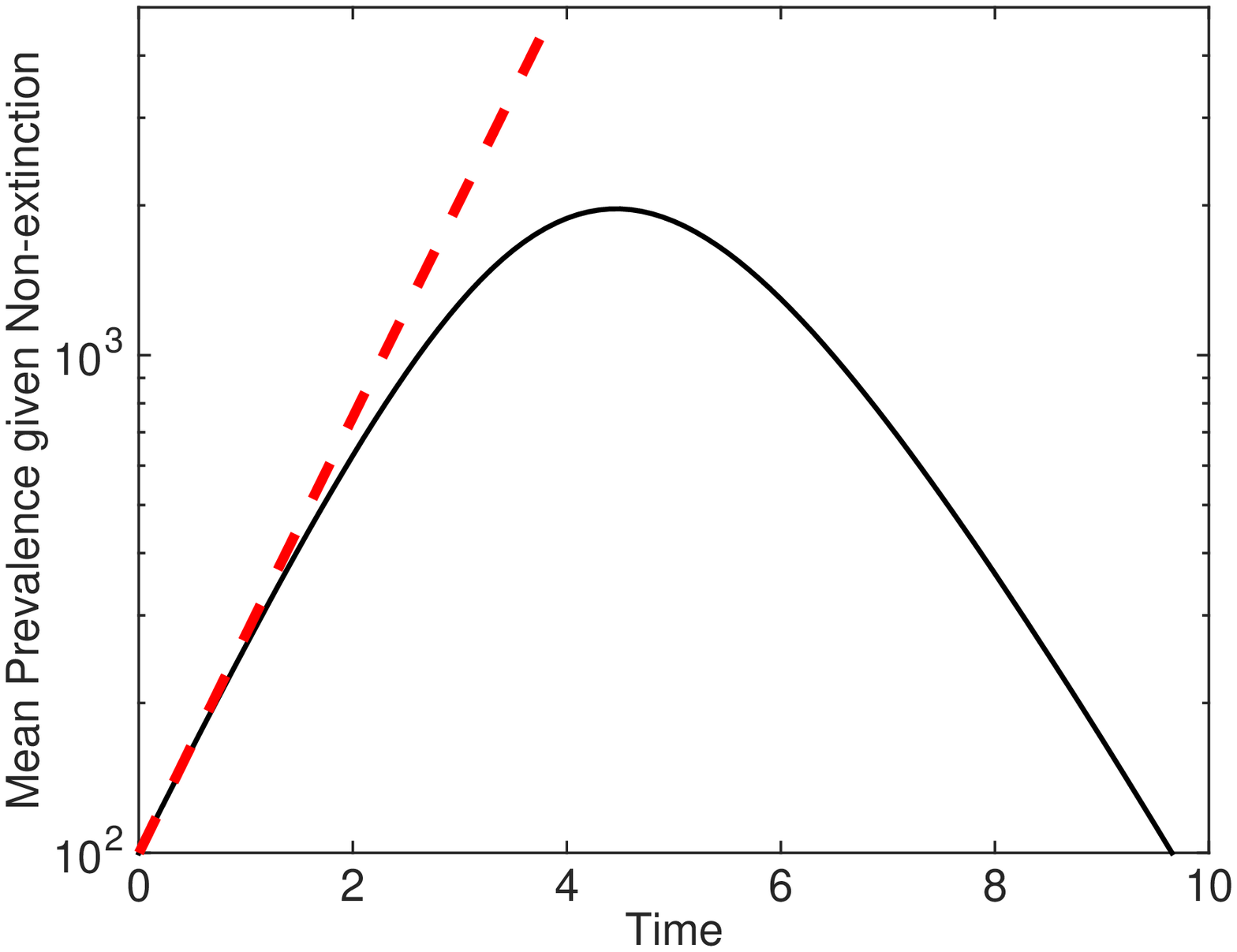} }} 
	{\resizebox{\fs}{!}{\includegraphics{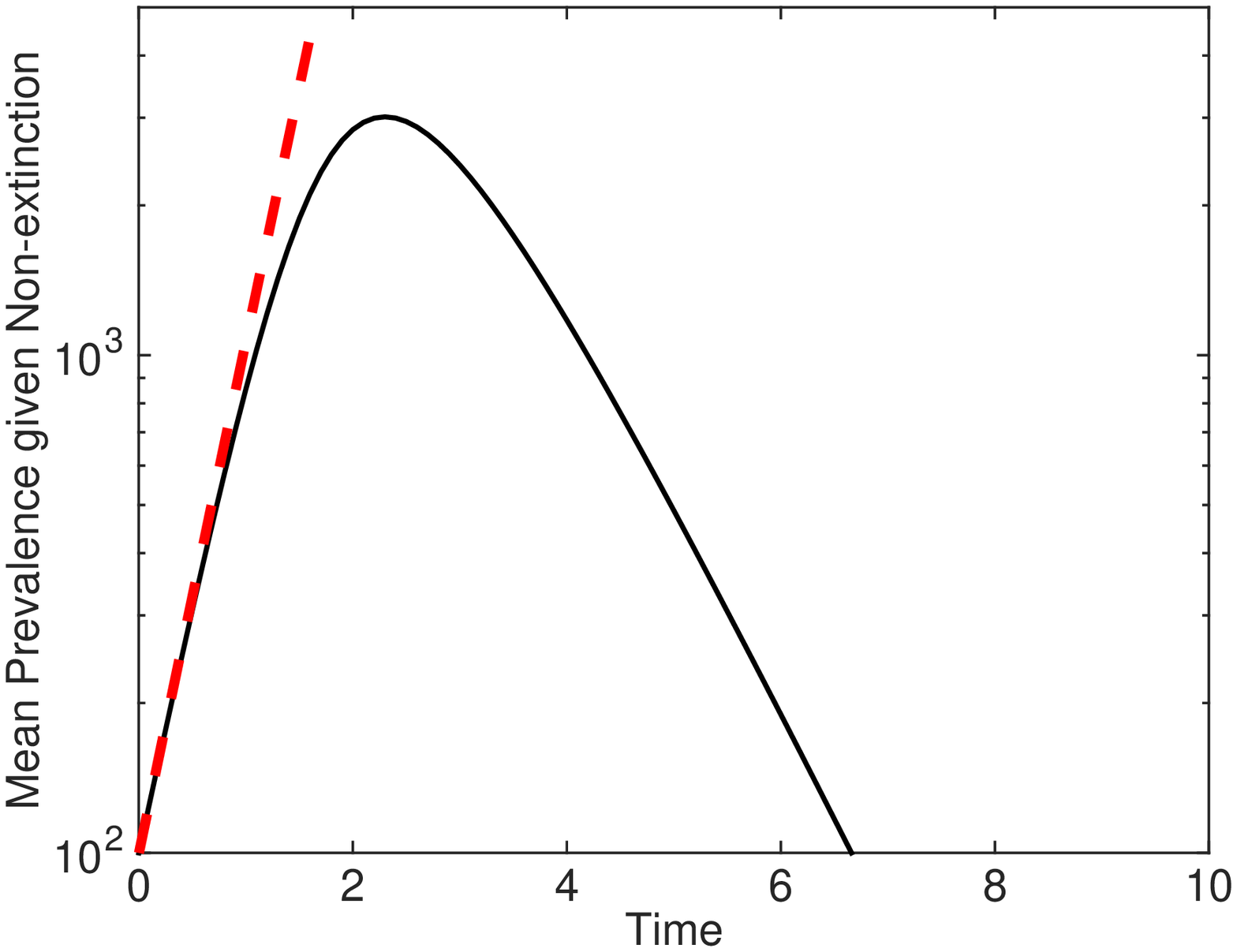} }} 
	{\resizebox{\fs}{!}{\includegraphics{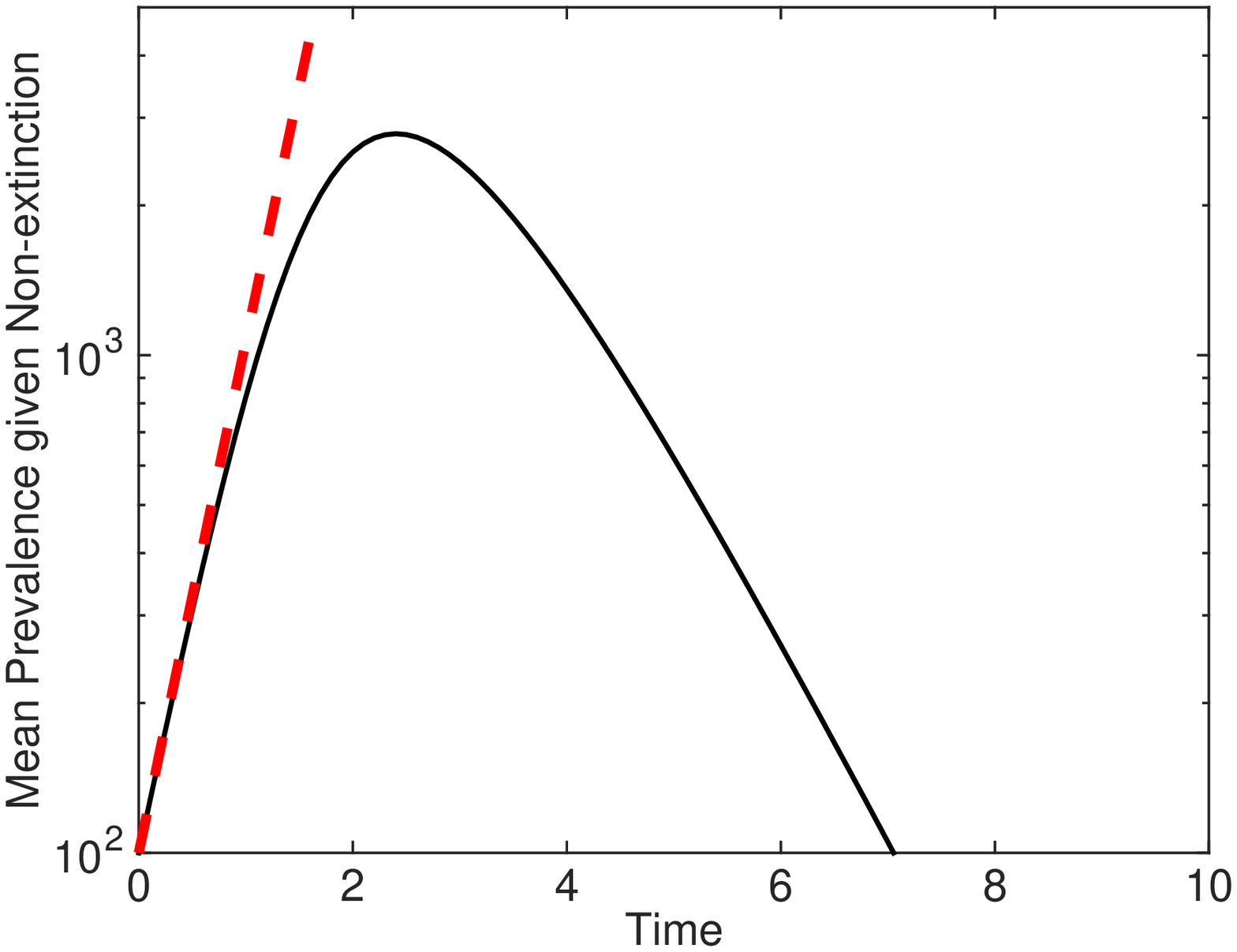} }}} \\
	\subfloat[Variance in prevalence. Black solid: simulations; Red dashed: branching process.]{%
		{\resizebox{\fs}{!}{\includegraphics{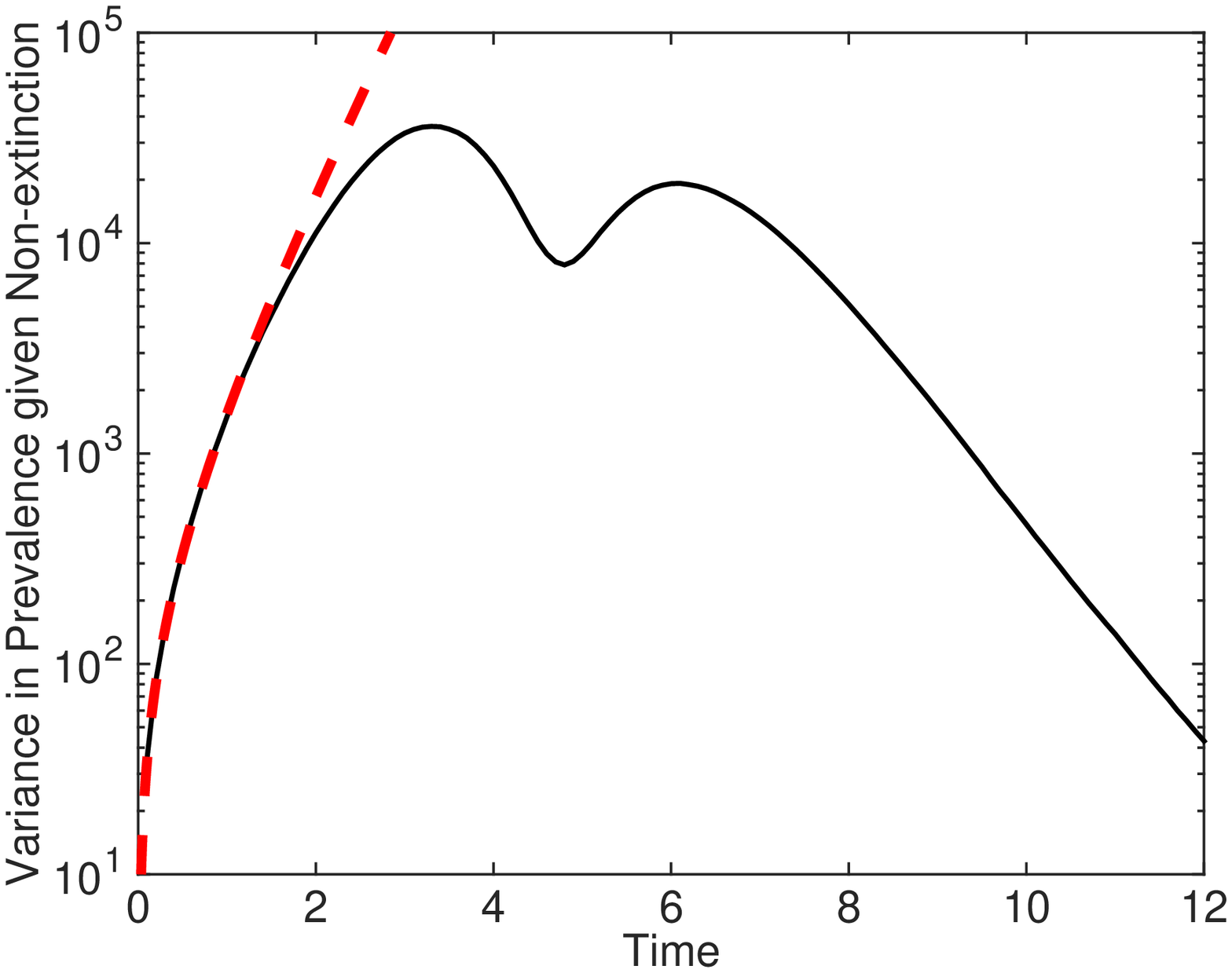} }} 
	{\resizebox{\fs}{!}{\includegraphics{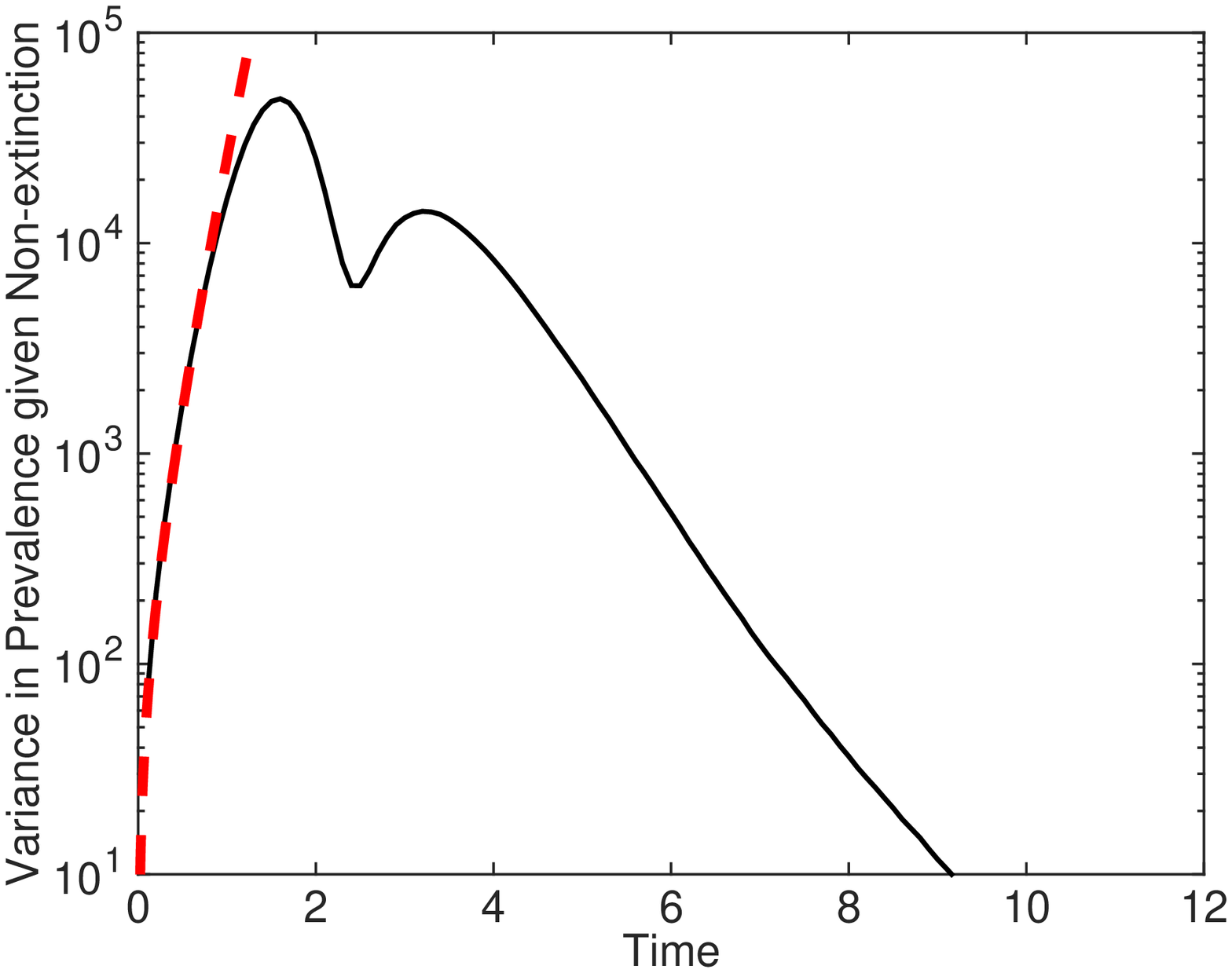} }} 
	{\resizebox{\fs}{!}{\includegraphics{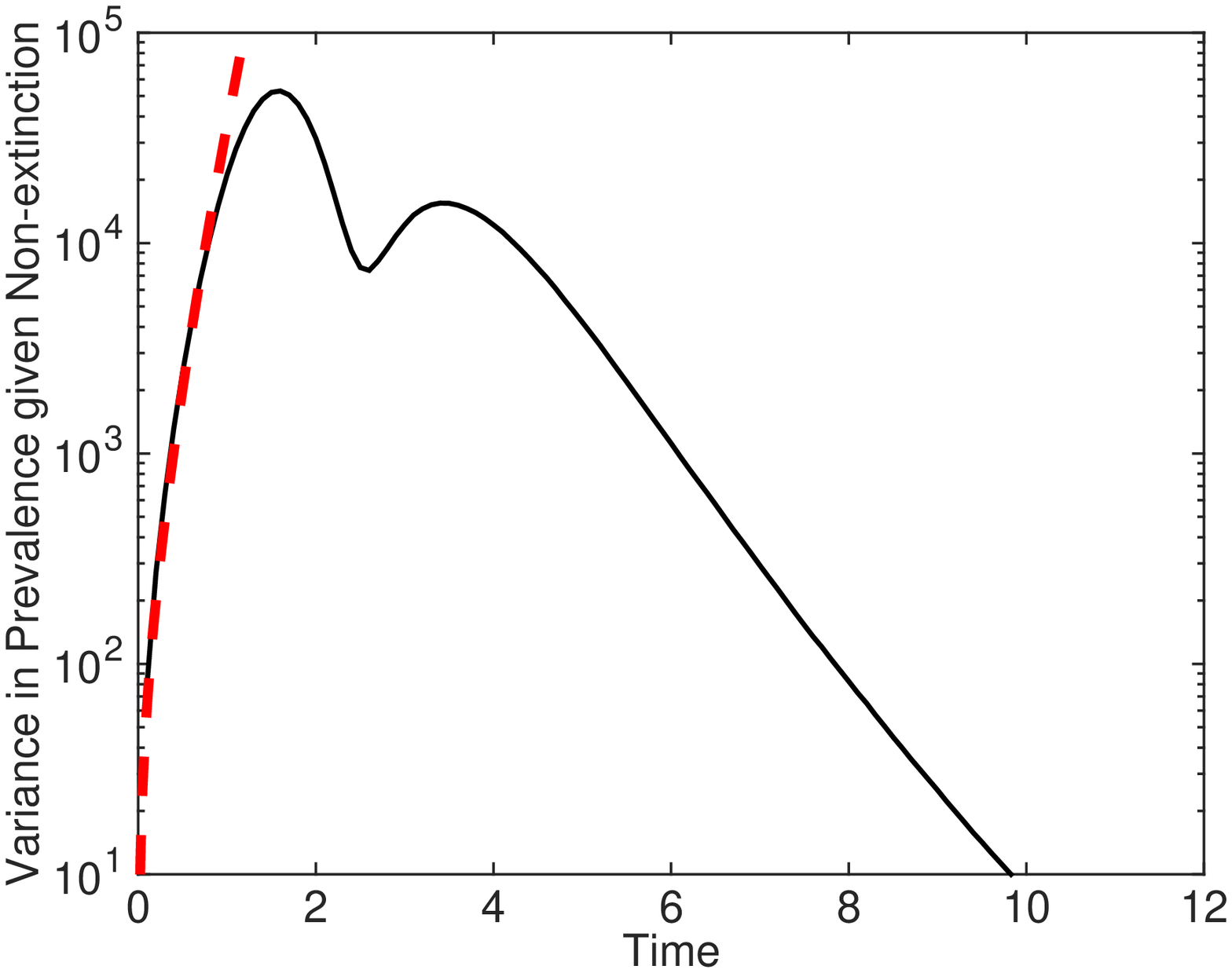} }}} \\
	\caption{Epidemic simulations that set $\mathrm{time}=0$ when prevalence is
	equal to 100. Parameters are $\tau=2$, $\gamma=1$ throughout.}
	\label{fig:restart}
\end{figure}

\begin{figure}
	\centering
	\subfloat[100 sample trajectories.]{%
	{\resizebox{\fs}{!}{\includegraphics{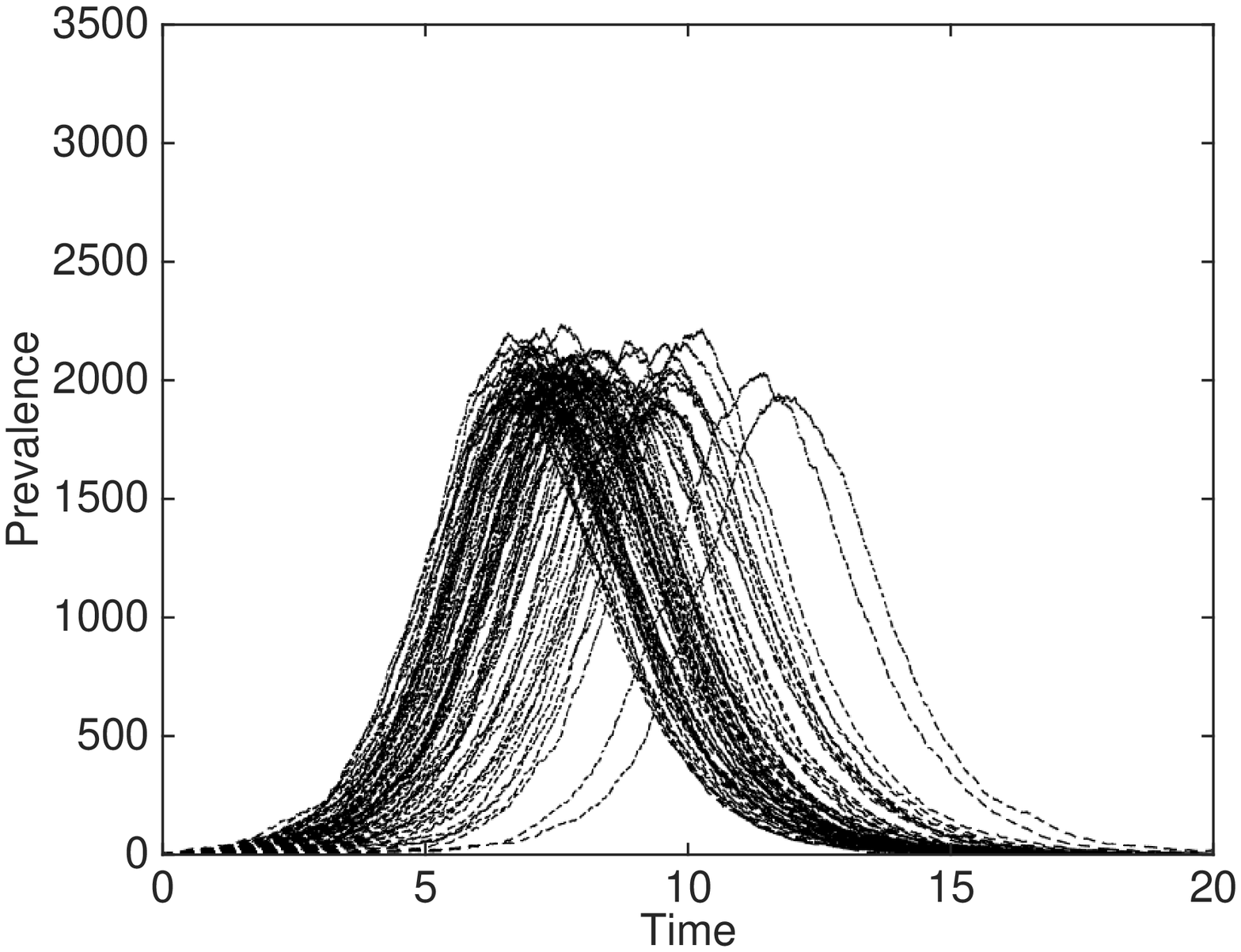} }} 
	{\resizebox{\fs}{!}{\includegraphics{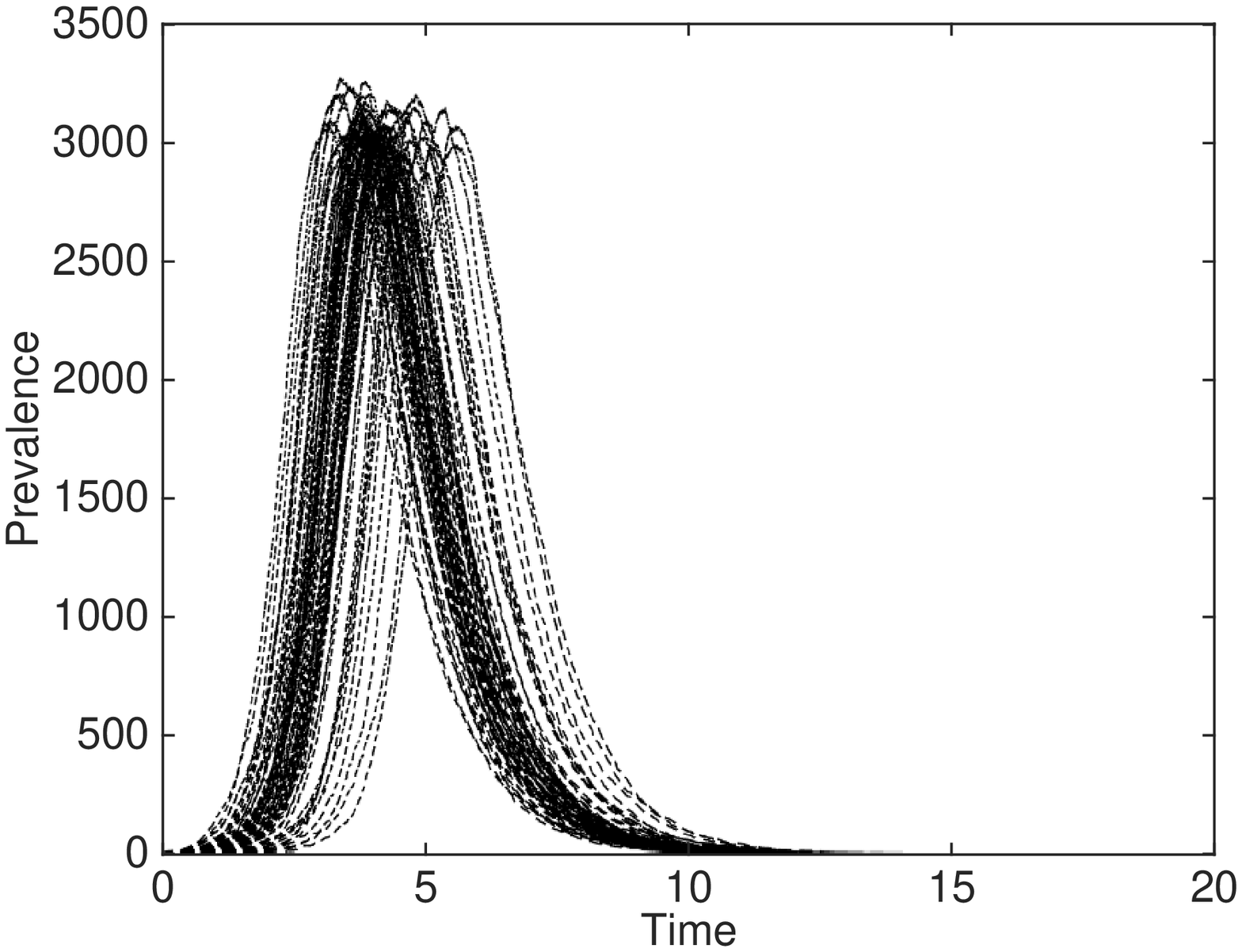} }} 
	{\resizebox{\fs}{!}{\includegraphics{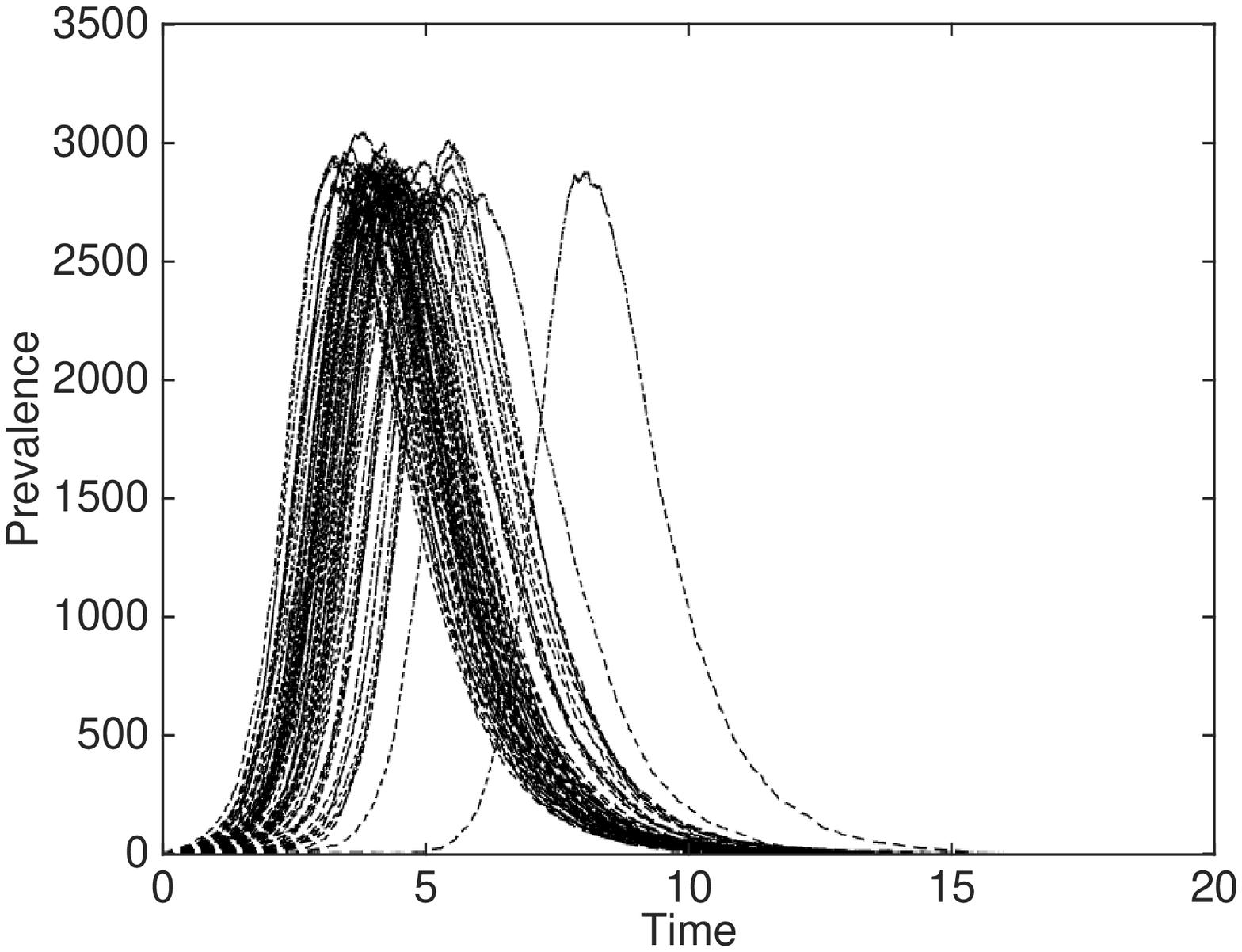}}}} \\
	\subfloat[Extinction probabilities. Black solid: simulations; Red dashed: branching process.]{
	{\resizebox{\fs}{!}{\includegraphics{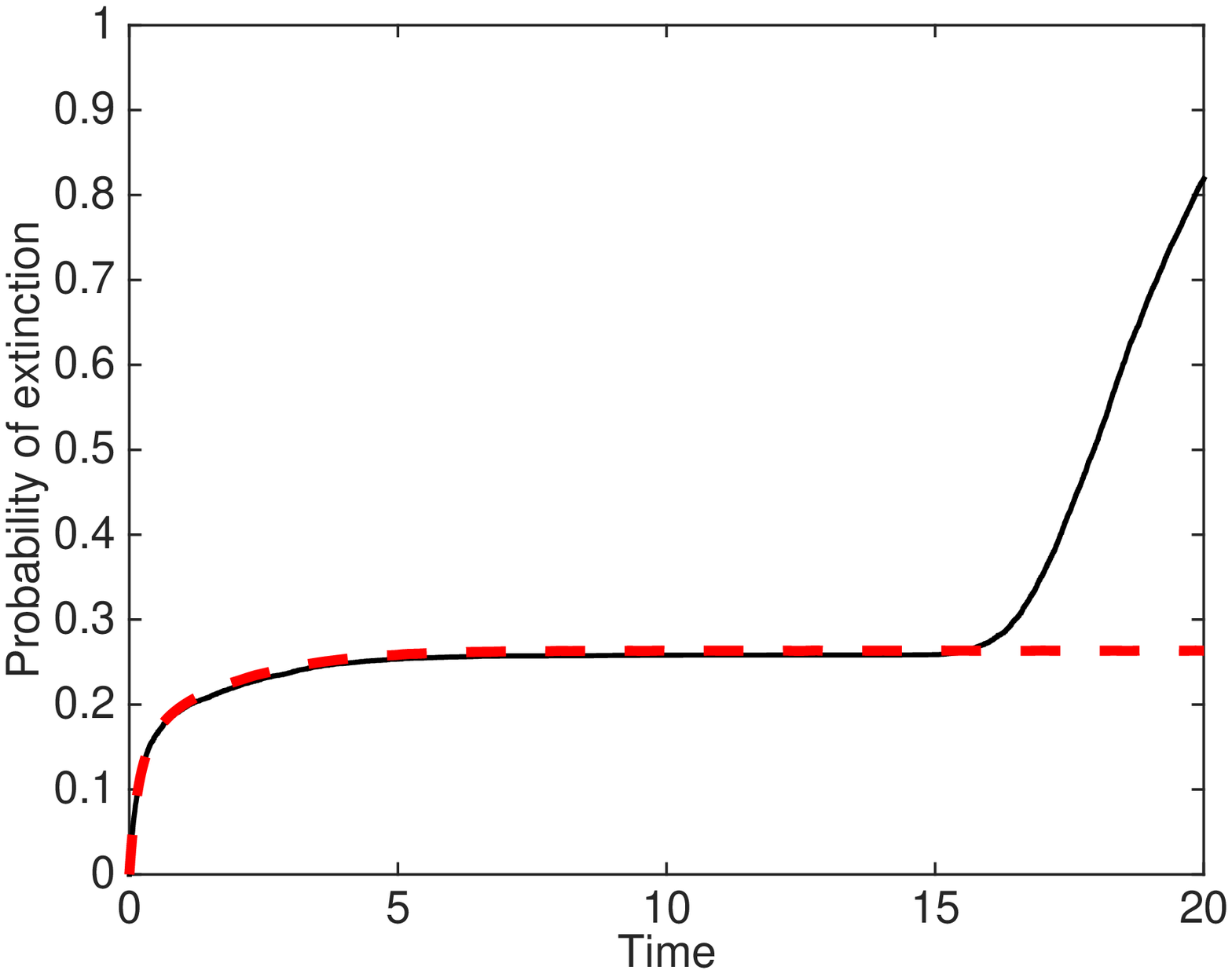} }} 
	{\resizebox{\fs}{!}{\includegraphics{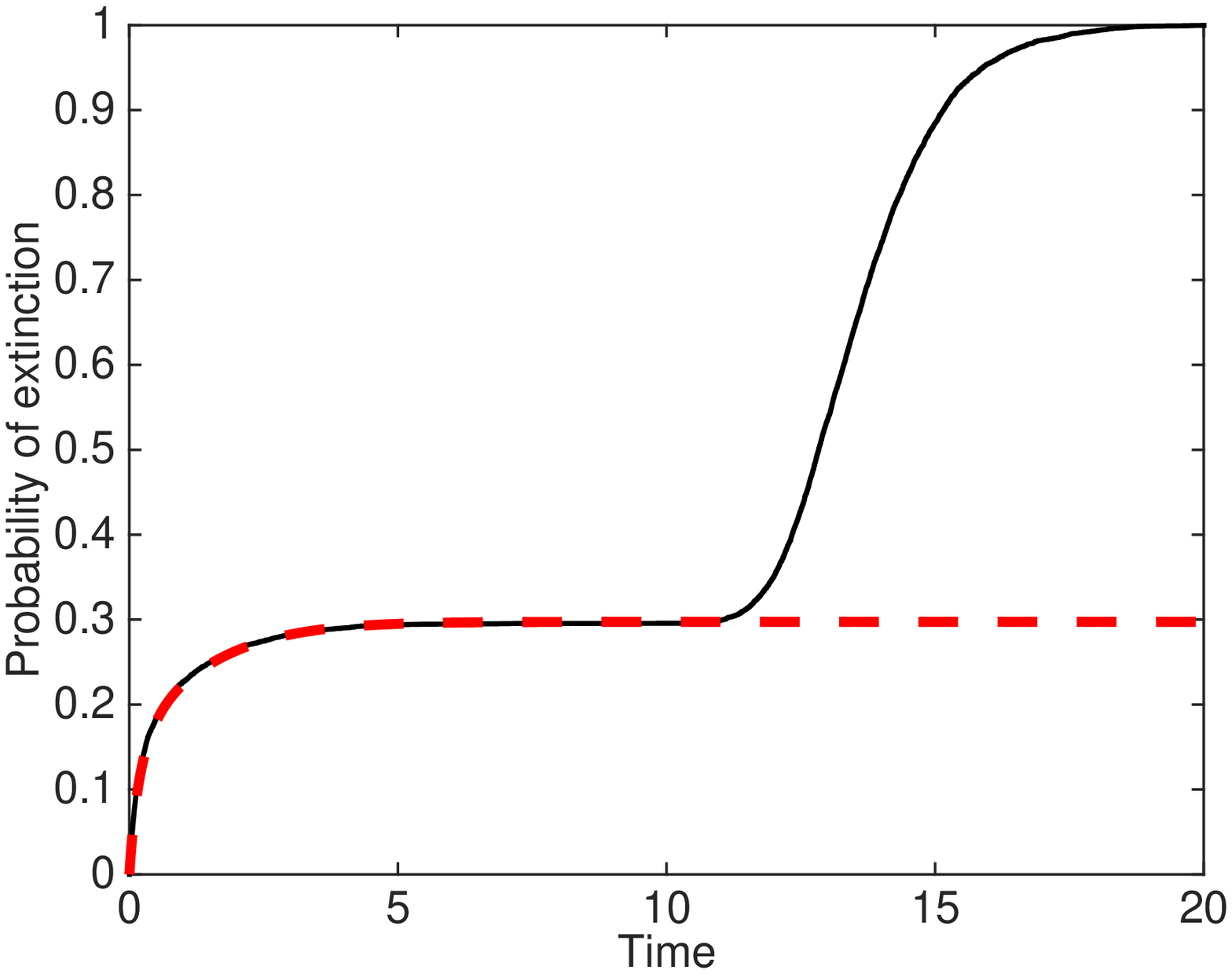} }} 
	{\resizebox{\fs}{!}{\includegraphics{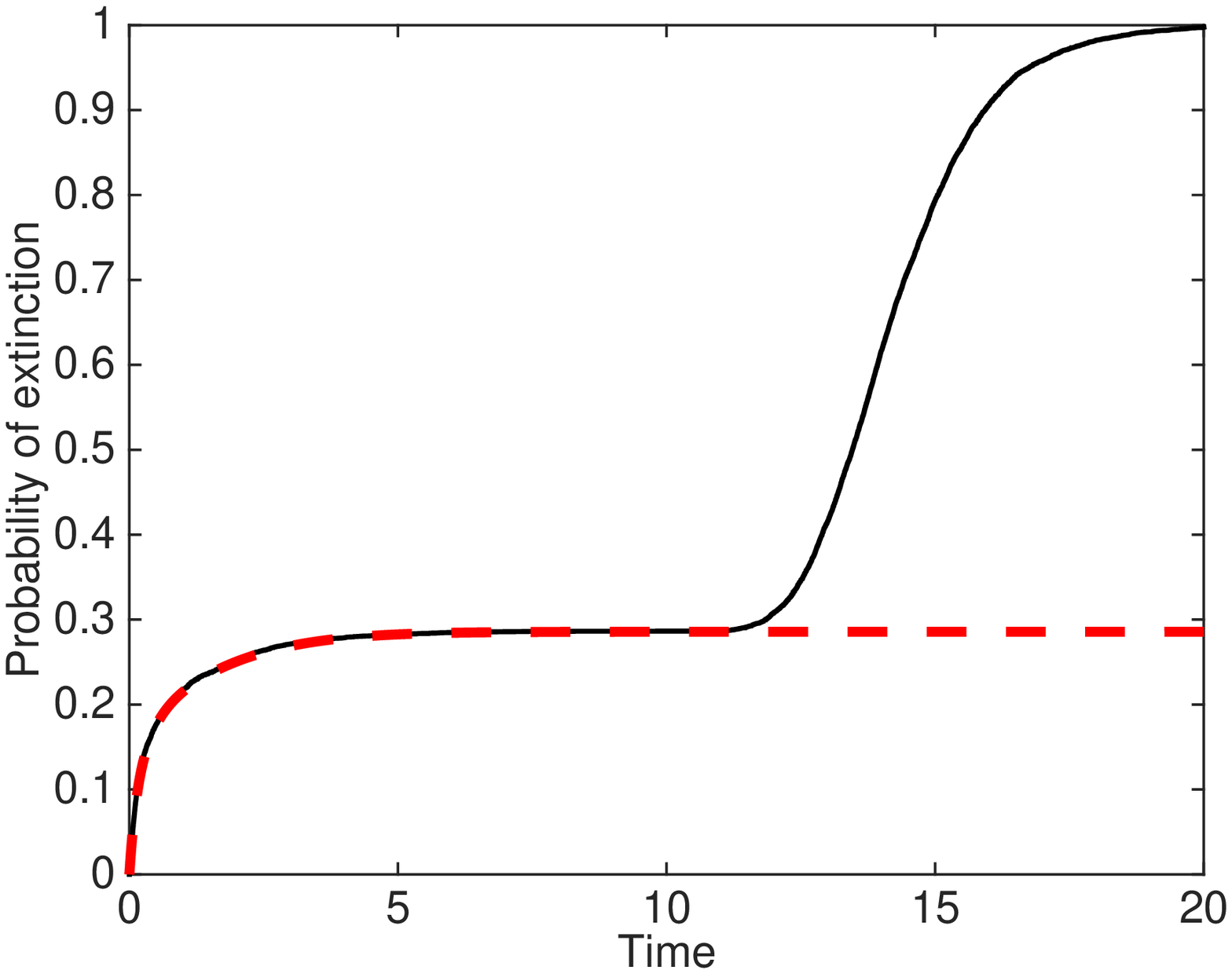} }}} \\
	\subfloat[Mean prevalence. Black solid: simulations; Red dashed: branching process.]{
	{\resizebox{\fs}{!}{\includegraphics{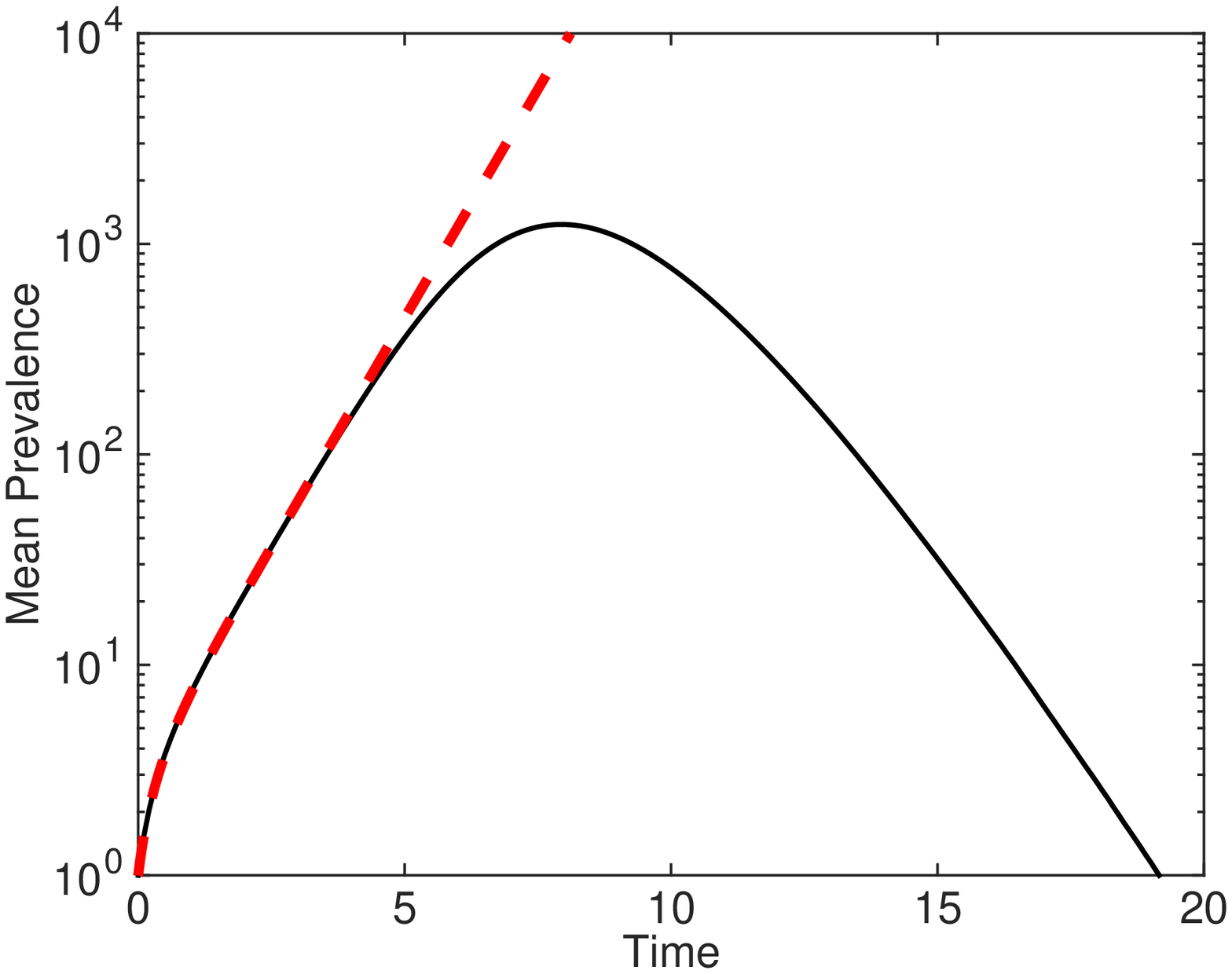} }} 
	{\resizebox{\fs}{!}{\includegraphics{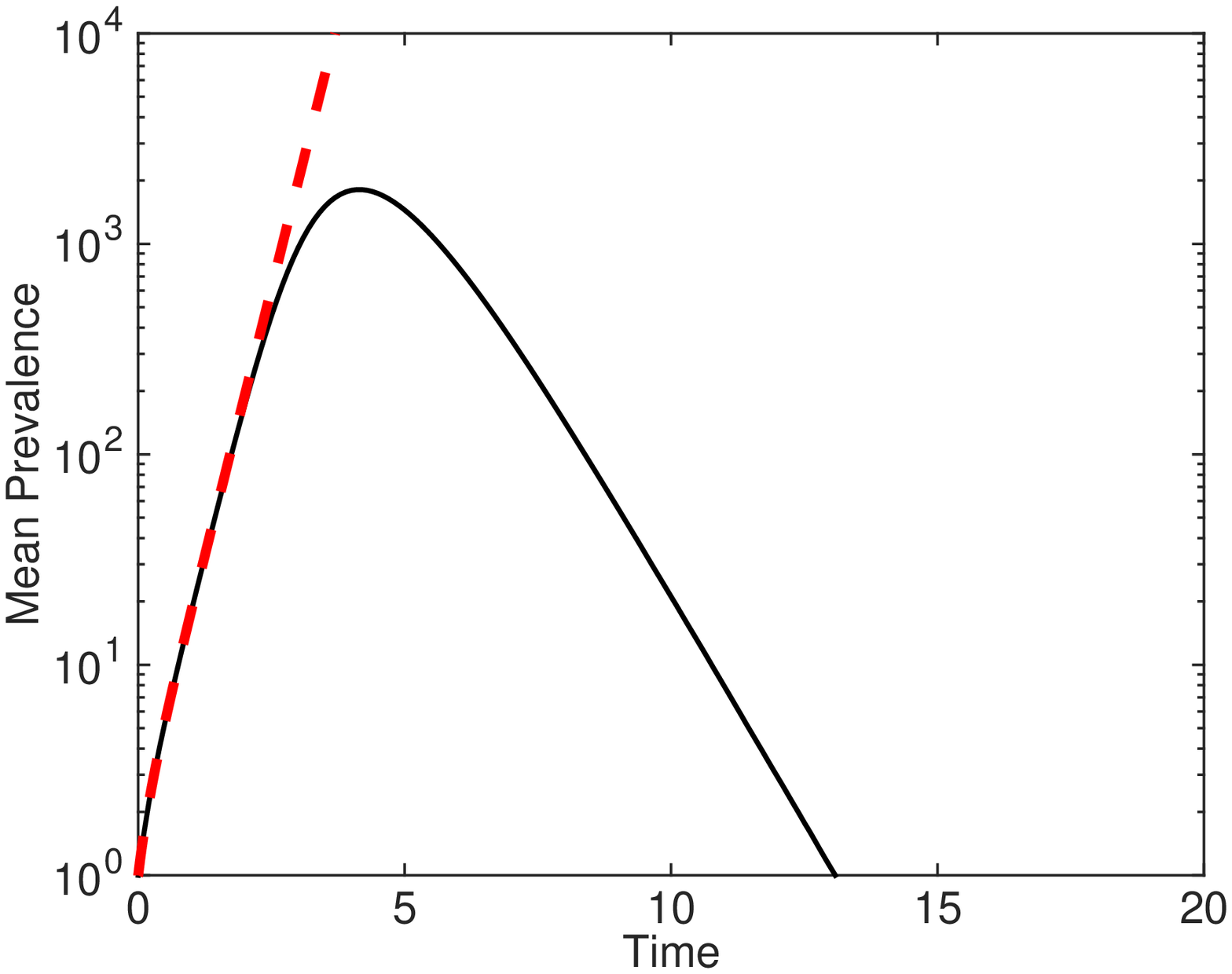} }} 
	{\resizebox{\fs}{!}{\includegraphics{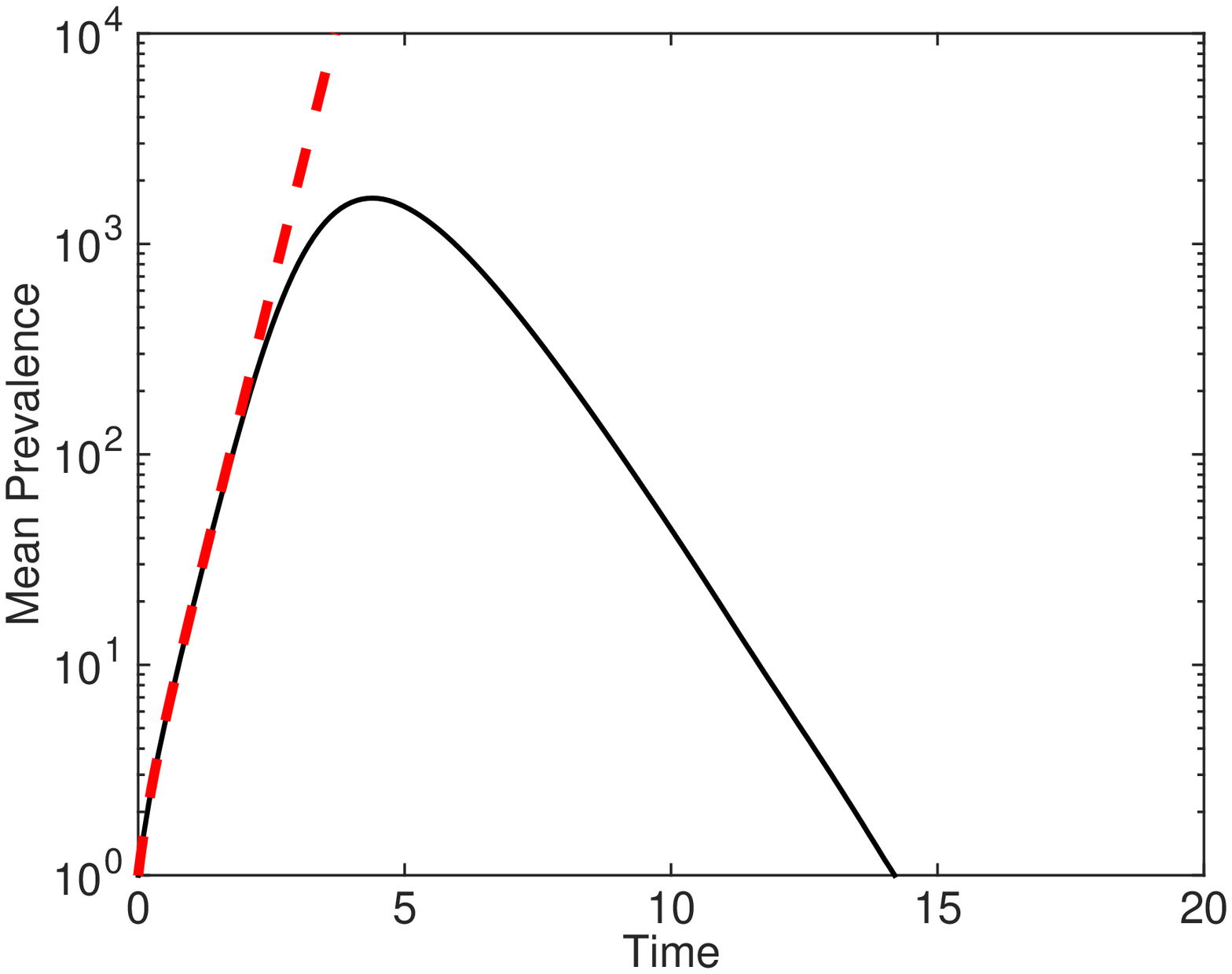} }}} \\
	\subfloat[Variance in prevalence. Black solid: simulations; Red dashed: branching process.]{
	{\resizebox{\fs}{!}{\includegraphics{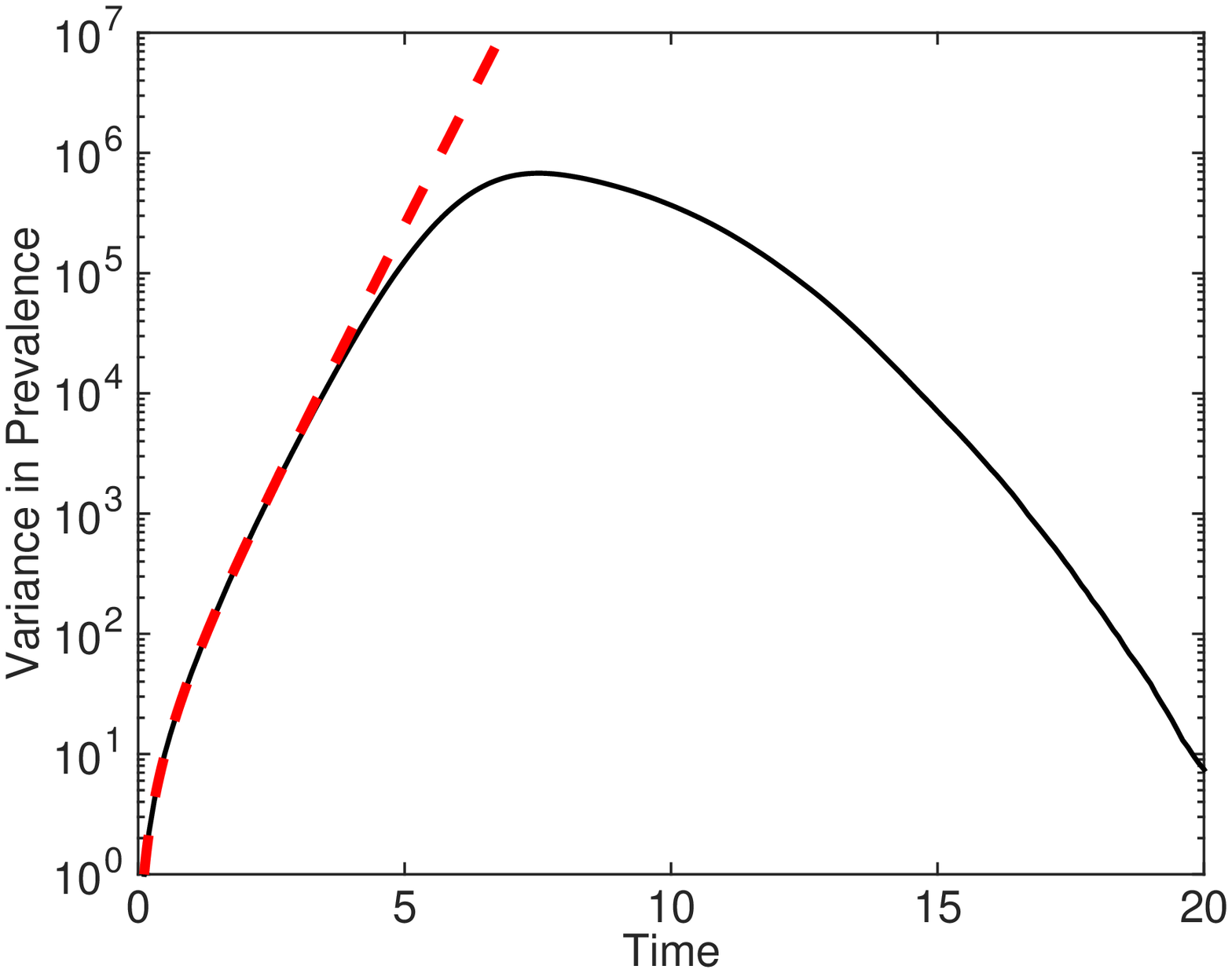} }} 
	{\resizebox{\fs}{!}{\includegraphics{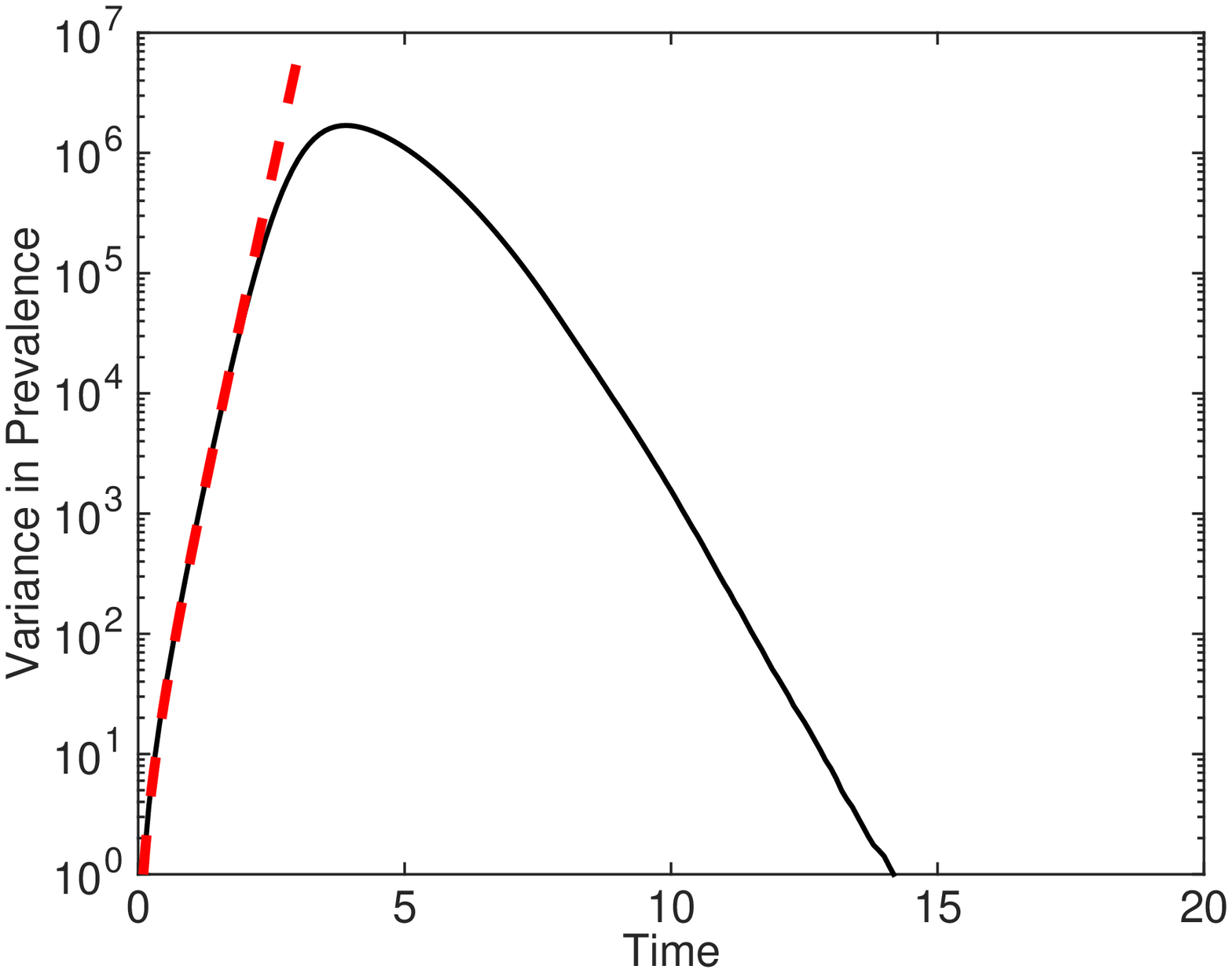} }} 
	{\resizebox{\fs}{!}{\includegraphics{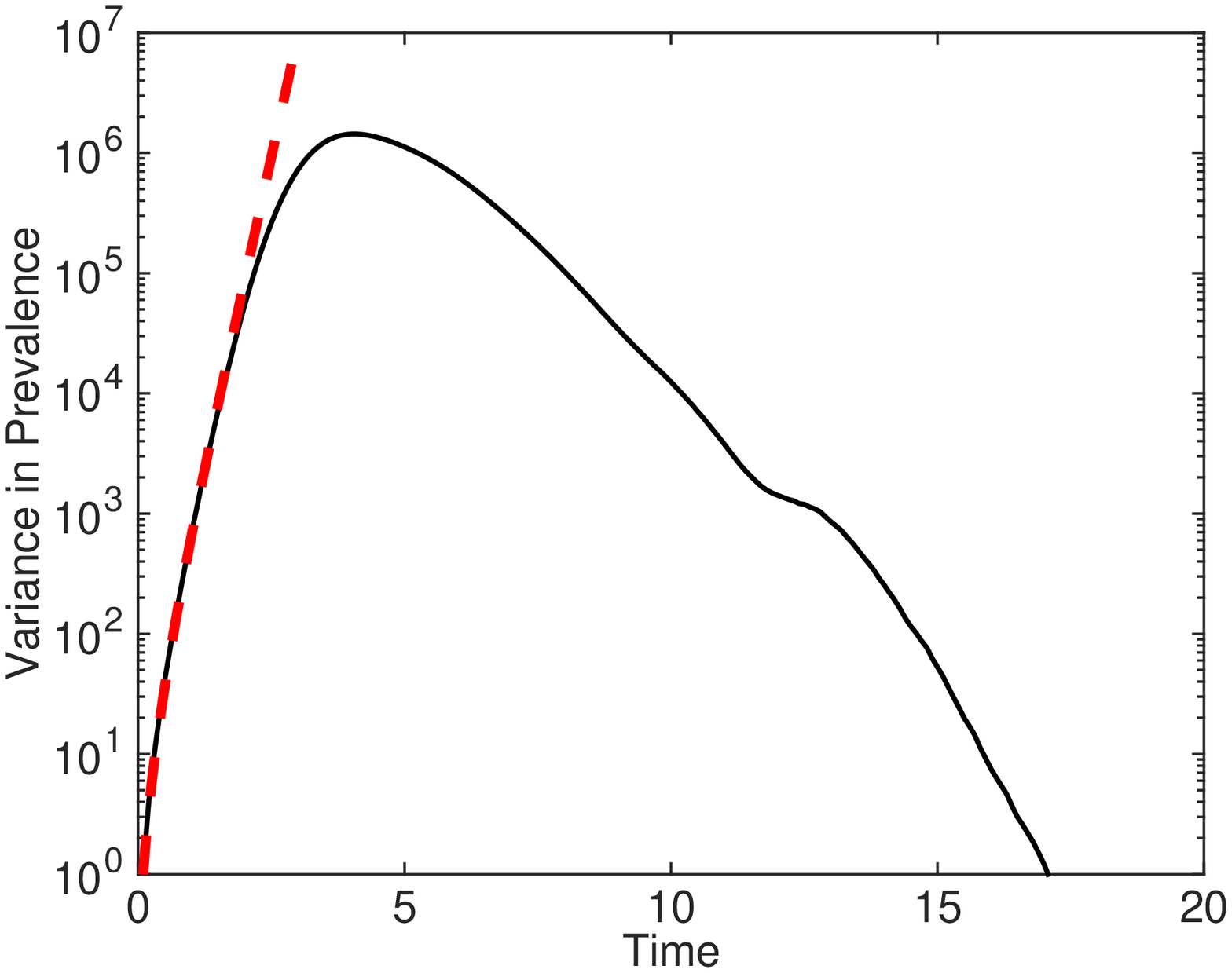} }}} \\
	\caption{Epidemic simulations starting from one node selected uniformly at
	random. Parameters are $\tau=2$, $\gamma=1$ throughout. Degree distributions
are as for Figure~\ref{fig:restart} above.}
	\label{fig:sims}
\end{figure}

\begin{figure}
	\centering
	\subfloat[Simulated epidemic]{%
	{\resizebox{.9\textwidth}{!}{\includegraphics{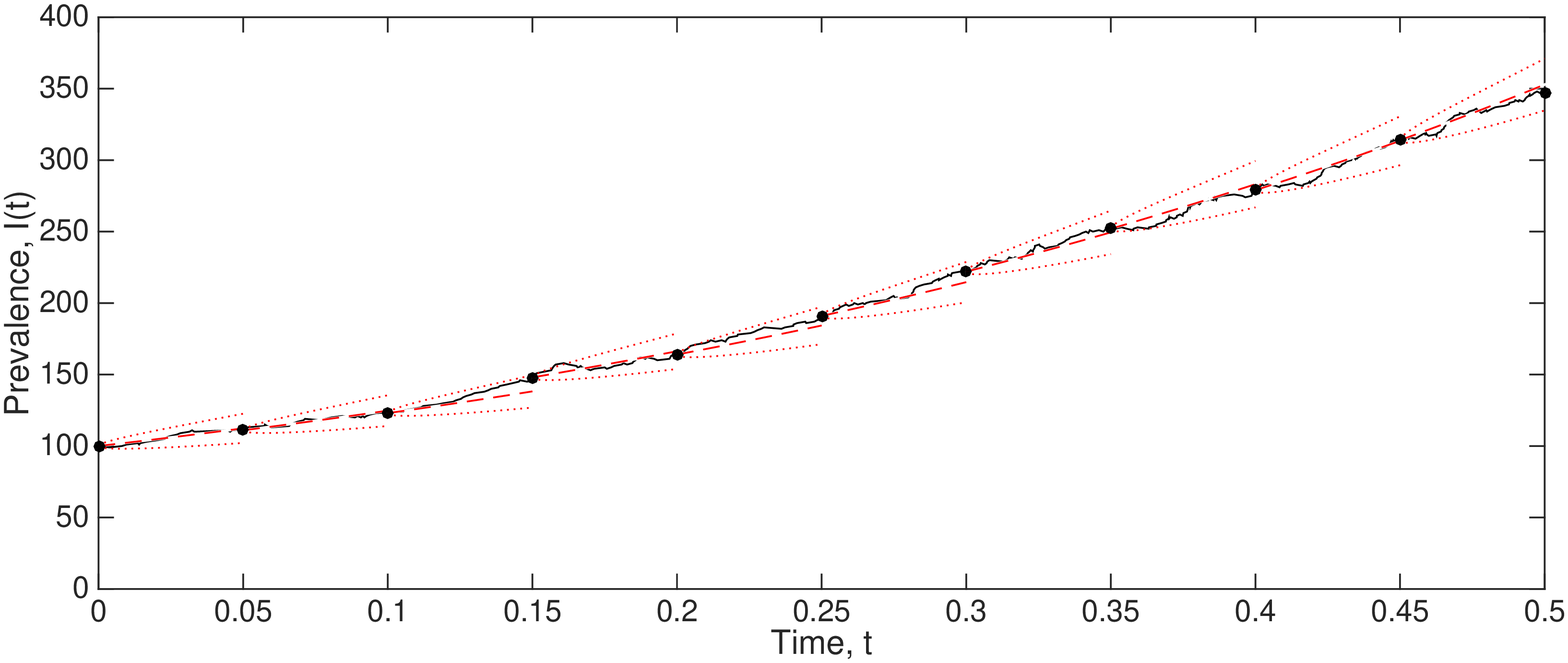} }}}\\
	\subfloat[Likelihood surfaces]{%
	{\resizebox{.45\textwidth}{!}{\includegraphics{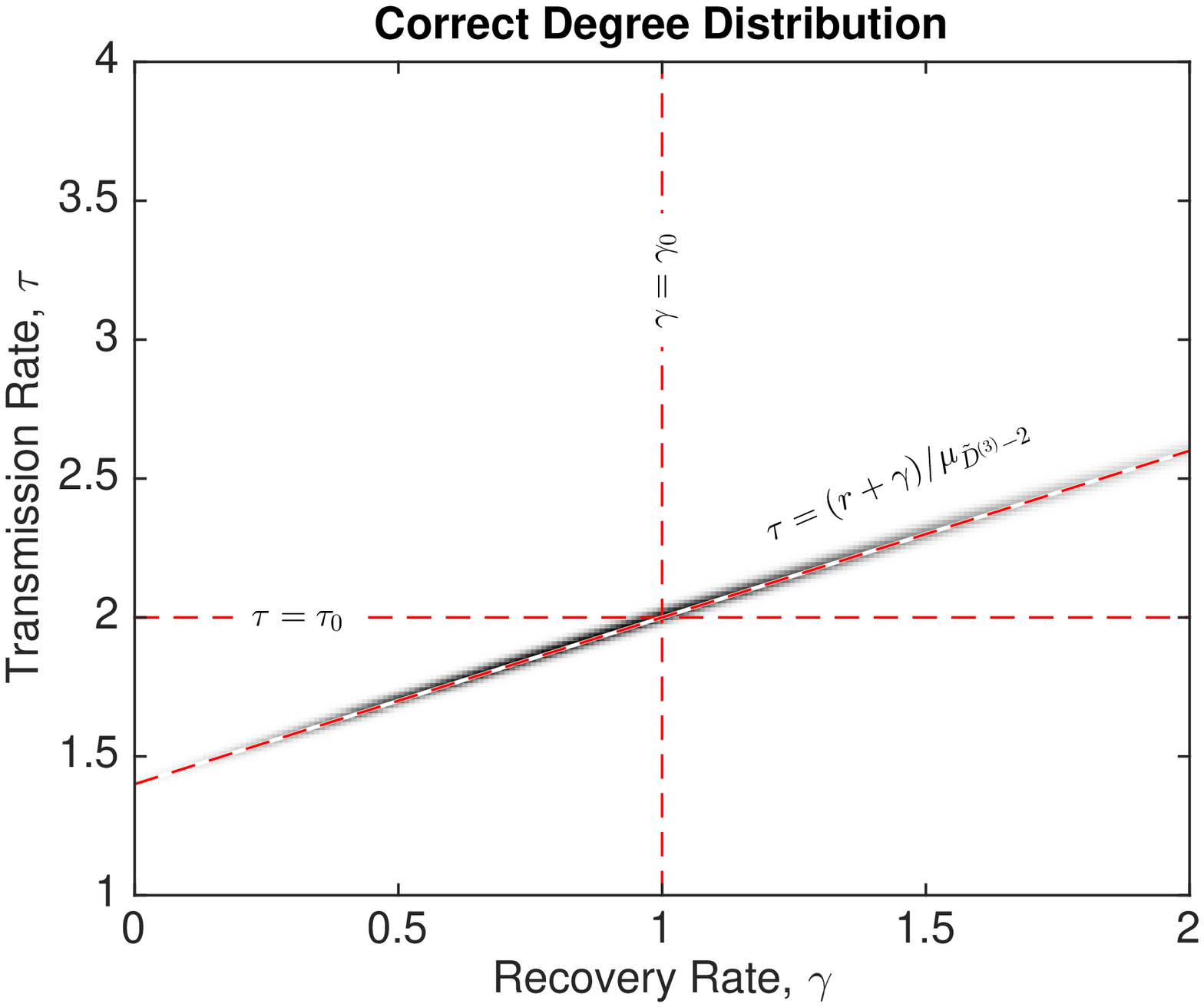} }}
{\resizebox{.45\textwidth}{!}{\includegraphics{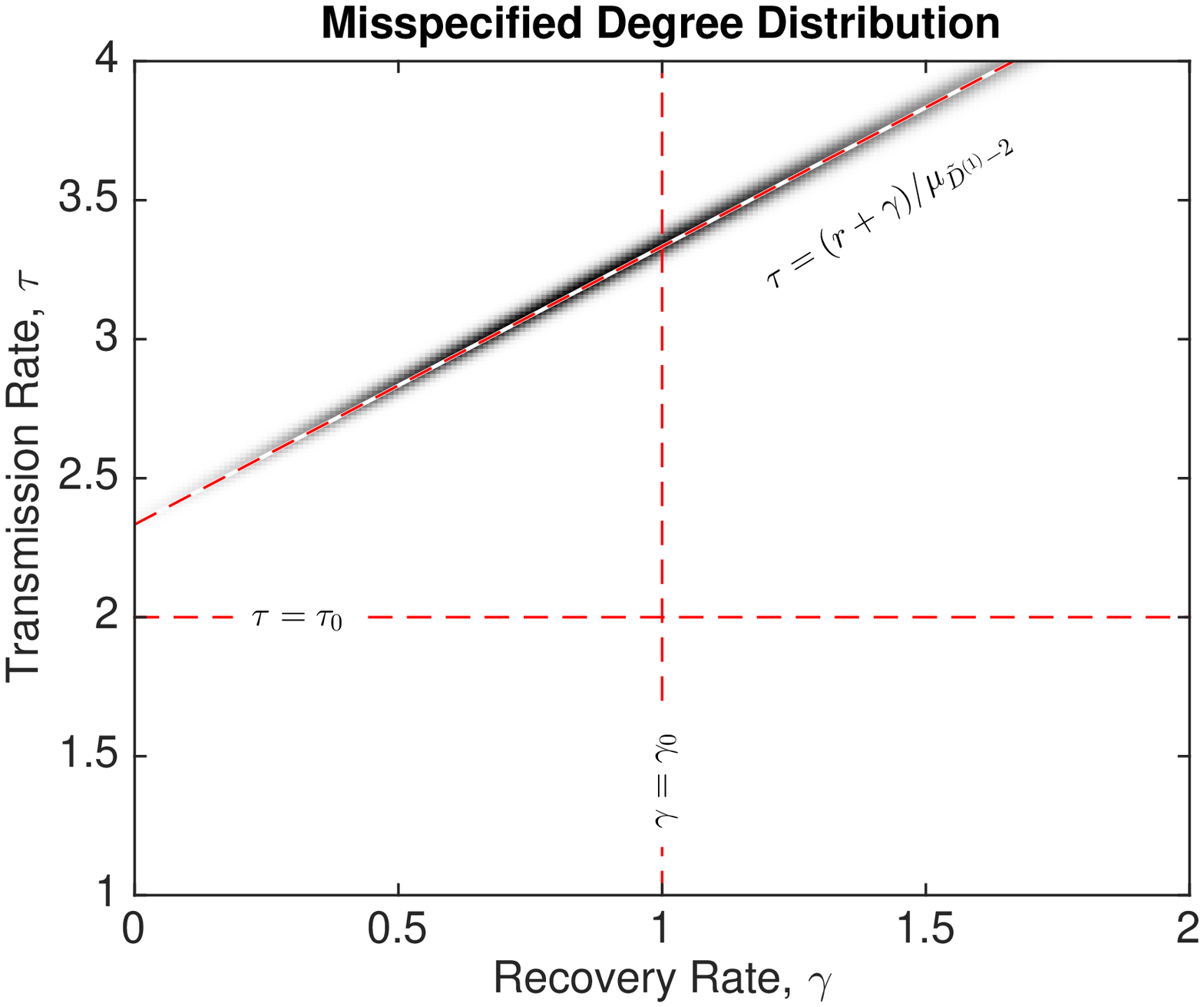}}} }
\caption{Simulation study. The top plot (i) shows the first quarter of the
	timepoints (observations as black dots, full trajectory as black solid line,
	Gaussian approximation mean as red dashed line, Gaussian approximation 95\%
	prediction interval as red dotted line). The bottom plots (ii) show
	likelihood surfaces for a Gaussian approximation model as described in
	Section~\ref{sec:infer}. True parameters are $\hat{\tau}=2$, $\hat{\gamma}=1$.
	Correct degree distribution $D^{(3)}$ is as given in the third columns of
	Figures \ref{fig:restart} and \ref{fig:sims} above, and the misspecified
	distribution $D^{(1)}$ is the distribution from the first column. Likelihood
	at a point is proportional to the intensity of shading, and three curves are
shown in each plot as dashed red lines. }
	\label{fig:infer}
\end{figure}


\begin{thebibliography}{10}

\bibitem{Bailey:1957}
N~T~J Bailey.
\newblock {\em The Mathematical Theory of Epidemics}.
\newblock Griffin, London, 1957.

\bibitem{Heesterbeek:2015}
Hans Heesterbeek, Roy~M. Anderson, Viggo Andreasen, Shweta Bansal, Daniela
  De~Angelis, Chris Dye, Ken T.~D. Eames, W.~John Edmunds, Simon D.~W. Frost,
  Sebastian Funk, T.~Deirdre Hollingsworth, Thomas House, Valerie Isham, Petra
  Klepac, Justin Lessler, James~O. Lloyd-Smith, C.~Jessica~E. Metcalf, Denis
  Mollison, Lorenzo Pellis, Juliet R.~C. Pulliam, Mick~G. Roberts, Cecile
  Viboud, and {Isaac Newton Institute IDD Collaboration}.
\newblock Modeling infectious disease dynamics in the complex landscape of
  global health.
\newblock {\em Science}, 347(6227):aaa4339, 2015.

\bibitem{Danon:2011}
Leon Danon, Ashley~P Ford, Thomas House, Chris~P Jewell, Matt~J Keeling,
  Gareth~O Roberts, Joshua~V Ross, and Matthew~C Vernon.
\newblock Networks and the epidemiology of infectious disease.
\newblock {\em Interdisciplinary Perspectives on Infectious Diseases},
  2011:1--28, 2011.

\bibitem{Molloy:1995}
Michael Molloy and Bruce Reed.
\newblock {A critical point for random graphs with a given degree sequence}.
\newblock {\em Random Structures and Algorithms}, 6:161--179, 1995.

\bibitem{Newman:2002}
M.~E.~J. Newman.
\newblock Spread of epidemic disease on networks.
\newblock {\em Physical Review E}, 66(1):016128, 2002.

\bibitem{Ball:2008}
Frank Ball and Peter Neal.
\newblock Network epidemic models with two levels of mixing.
\newblock {\em Mathematical Biosciences}, 212(1):69--87, 2008.

\bibitem{Lindquist:2010}
Jennifer Lindquist, Junling Ma, P.~Driessche, and Frederick~H. Willeboordse.
\newblock Effective degree network disease models.
\newblock {\em Journal of Mathematical Biology}, 62(2):143--164, 2010.

\bibitem{Volz:2008}
Erik~M Volz.
\newblock {SIR} dynamics in random networks with heterogeneous connectivity.
\newblock {\em Journal of Mathematical Biology}, 56(3):293--310, 2008.

\bibitem{Miller:2011}
Joel~C. Miller.
\newblock A note on a paper by {{Erik Volz: SIR}} dynamics in random networks.
\newblock {\em Journal of Mathematical Biology}, 62(3):349--358, 2011.

\bibitem{Miller:2012}
Joel~C. Miller, Anja~C. Slim, and Erik~M. Volz.
\newblock Edge-based compartmental modelling for infectious disease spread.
\newblock {\em Journal of The Royal Society Interface}, 9(70):890--906, 2012.

\bibitem{House:2010}
Thomas House and Matt~J. Keeling.
\newblock Insights from unifying modern approximations to infections on
  networks.
\newblock {\em Journal of The Royal Society Interface}, 8(54):67--73, 2010.

\bibitem{Eames:2002}
K~T~D Eames and M~J Keeling.
\newblock Modeling dynamic and network heterogeneities in the spread of
  sexually transmitted diseases.
\newblock {\em PNAS}, 99(20):13330--13335, Jan 2002.

\bibitem{Decreusefond:2012}
Laurent Decreusefond, Jean-St\'{e}phane Dhersin, Pascal Moyal, and Viet~Chi
  Tran.
\newblock Large graph limit for an {{SIR}} process in random network with
  heterogeneous connectivity.
\newblock {\em The Annals of Applied Probability}, 22(2):541--575, 2012.

\bibitem{Bohman:2012}
Tom Bohman and Michael Picollelli.
\newblock {SIR} epidemics on random graphs with a fixed degree sequence.
\newblock {\em Random Structures and Algorithms}, 41(2):179--214, 2012.

\bibitem{Barbour:2013}
Andrew Barbour and Gesine Reinert.
\newblock Approximating the epidemic curve.
\newblock {\em Electronic Journal of Probability}, 18(54):1--30, 2013.

\bibitem{Janson:2014}
Svante Janson, Malwina Luczak, and Peter Windridge.
\newblock Law of large numbers for the {SIR} epidemic on a random graph with
  given degrees.
\newblock {\em Random Structures and Algorithms}, 45(4):726--763, 2014.

\bibitem{Graham:2014}
Matthew Graham and Thomas House.
\newblock Dynamics of stochastic epidemics on heterogeneous networks.
\newblock {\em Journal of Mathematical Biology}, 68(7):1583--1605, 2014.

\bibitem{Ethier:1986}
S.~N. Ethier and T.~G. Kurtz.
\newblock {\em {M}arkov processes: characterization and convergence}.
\newblock {W}iley series in probability and mathematical statistics. John Wiley
  and Sons, Hoboken, New Jersey, 1986.

\bibitem{Ross:2006}
J.~V. Ross, T.~Taimre, and P.~K. Pollett.
\newblock On parameter estimation in population models.
\newblock {\em Theoretical Population Biology}, 70(4):498--510, 2006.

\bibitem{Durrett:2007}
Rick Durrett.
\newblock {\em Random Graph Dynamics}.
\newblock Cambridge University Press, 2007.

\bibitem{Ball:1995}
Frank Ball and Peter Donnelly.
\newblock Strong approximations for epidemic models.
\newblock {\em Stochastic Processes and their Applications}, 55(1):1--21, 1995.

\bibitem{Dorman:2004}
K.~Dorman, J.~Sinsheimer, and K.~Lange.
\newblock In the garden of branching processes.
\newblock {\em SIAM Review}, 46(2):202--229, 2004.

\bibitem{Daley:1988}
D.~J. Daley and D.~Vere-Jones.
\newblock {\em An Introduction to the Theory of Point Processes}.
\newblock Probability and Its Applications. Springer, New York, 1988.

\bibitem{Miller:2014}
J.~C. Miller and I.~Z. Kiss.
\newblock Epidemic spread in networks: Existing methods and current challenges.
\newblock {\em Mathematical Modelling of Natural Phenomena}, 9(2):4--42, 2014.

\bibitem{Athreya:1972}
K.~B. Athreya and P.~E. Ney.
\newblock {\em Branching Processes}.
\newblock Springer-Verlag, Berlin, 1972.

\bibitem{Holme:2013}
Petter Holme.
\newblock Extinction times of epidemic outbreaks in networks.
\newblock {\em PLoS ONE}, 8(12):e84429, 12 2013.

\bibitem{Windridge:2014}
Peter Windridge.
\newblock The extinction time of a subcritical branching process related to the
  {SIR} epidemic on a random graph.
\newblock {\em Journal of Applied Probability}, 52(4):1195--1201, 2015.

\bibitem{Waugh:1958}
W.~A.~O'N. Waugh.
\newblock Conditioned {M}arkov processes.
\newblock {\em Biometrika}, 45(1-2):241--249, 1958.

\bibitem{Heinzmann:2009}
Dominik Heinzmann.
\newblock Extinction times in multitype {M}arkov branching processes.
\newblock {\em Journal of Applied Probability}, 46(1):296--307, 2009.

\bibitem{Murray:2002}
J.~D. Murray.
\newblock {\em {Mathematical Biology I}}.
\newblock Springer, 3rd edition, 2002.

\bibitem{Murray:2003}
J.~D. Murray.
\newblock {\em {Mathematical Biology II}}.
\newblock Springer, 3rd edition, 2003.

\bibitem{Murray:1984}
J.~D. Murray.
\newblock {\em Asymptotic Analysis}, volume~48 of {\em Applied Mathematical
  Sciences}.
\newblock Springer, New York, 1984.

\bibitem{Constable:2014}
George W.~A. Constable and Alan~J. McKane.
\newblock Fast-mode elimination in stochastic metapopulation models.
\newblock {\em Physical Review E}, 89(3):032141, 2014.

\bibitem{Parsons:2015}
T~L Parsons and T~Rogers.
\newblock Dimension reduction via timescale separation in stochastic dynamical
  systems.
\newblock [arXiv:1510.07031], 2015.

\bibitem{ONeill:1999}
P~D O'Neill and G~O Roberts.
\newblock {{B}}ayesian inference for partially observed stochastic epidemics.
\newblock {\em Journal of the Royal Statistical Society A}, 162:121--129, 1999.

\bibitem{Newman:2002b}
M.~E.~J. Newman.
\newblock Assortative mixing in networks.
\newblock {\em Physical Review Letters}, 89(20):208701, 2002.

\bibitem{Klepac:2012}
Petra Klepac, C.~Jessica~E. Metcalf, Angela~R. McLean, and Katie Hampson.
\newblock Towards the endgame and beyond: complexities and challenges for the
  elimination of infectious diseases.
\newblock {\em Philosophical Transactions of the Royal Society of London B:
  Biological Sciences}, 368(1623):20120137, 2013.

\bibitem{Black:2014}
Andrew~J. Black, Thomas House, Matt~J. Keeling, and Joshua~V. Ross.
\newblock The effect of clumped population structure on the variability of
  spreading dynamics.
\newblock {\em Journal of Theoretical Biology}, 359:45--53, 2014.

\bibitem{Miller:2014a}
Joel~C. Miller.
\newblock Epidemics on networks with large initial conditions or changing
  structure.
\newblock {\em PLoS ONE}, 9(7):e101421, 2014.

\bibitem{Grimmett:2001}
G~R Grimmett and D~R Stirzaker.
\newblock {\em Probability and Random Processes}.
\newblock Oxford University Press, Oxford, 3 edition, 2001.

\end{thebibliography}
\end{document}